\documentclass[a4,10pt]{article}
\usepackage{graphicx}
\usepackage{amsmath}
\usepackage{amssymb}
\usepackage{bbm}
\usepackage{caption}
\usepackage{multirow}
\usepackage{geometry}
\usepackage{tikz}
\usepackage{verbatim}
\usetikzlibrary{shapes}
\usepackage{amsthm}
\usepackage[bookmarks=true,colorlinks=true]{hyperref}
\usepackage{xspace}
\newtheorem{theorem}{Theorem}[section]
\newtheorem{lemma}{Lemma}[section]
\newtheorem{proposition}{Proposition}[section]
\newtheorem{definition}{Definition}[section]
\newtheorem{remark}{Remark}[section]
\newtheorem{assumption}{Assumption}[section]

\numberwithin{equation}{section}
\numberwithin{figure}{section}
\numberwithin{table}{section}
\usepackage{float}
\usepackage{indentfirst}
\usepackage{caption}
\usepackage{subfigure}
\allowdisplaybreaks
\usepackage{epstopdf} 
\usepackage[round]{natbib}
\begin{document}

	\title{\bf Effect of Labour Income on the Optimal Bankruptcy Problem} 
	\date{}
	\author{\sffamily Guodong Ding$^{1}$, Daniele Marazzina$^{1}$\thanks{ Corresponding Author. E-mail: daniele.marazzina@polimi.it. ORCID ID: orcid.org/0000-0001-6107-9822}\\
		\sffamily\small $^1$ Department of Mathematics, Politecnico di Milano\\ \sffamily\small I-20133 Milano, Italy}
	\maketitle
	
	{\noindent\small{\bf Abstract:}
		In this paper we deal with the optimal bankruptcy problem for an agent who can optimally allocate her consumption rate, the amount of capital invested in the risky asset as well as her leisure time. In our framework, the agent is endowed by an initial debt, and she is required to repay her debt continuously. Declaring bankruptcy, the debt repayment is exempted at the cost of a wealth shrinkage. We implement the duality method to solve the problem analytically and conduct a sensitivity analysis to the cost and benefit parameters of bankruptcy. Introducing the flexible leisure/working rate, and therefore the labour income, into the bankruptcy model, we investigate its effect on the optimal strategies.   }
	
	\vspace{1ex}
	{\noindent\small{\bf Keywords:} Power Utility Optimization, Bankruptcy Stopping Time, Consumption-Portfolio-Leisure Controls, Legendre-Fenchel Transform, Variational Inequalities}


\section{Introduction}
In this paper, we study an optimal stopping time problem, in which an agent decides her consumption-portfolio-leisure strategy as well as the optimal bankruptcy time. Her utility is described by a power utility function concerning the consumption and leisure rates. Moreover, the leisure rate should be upper bounded by a positive constant ($L$). Relating to the leisure rate, the agent earns the labour income with a fixed wage rate. The sum of labour and leisure rates is assumed to be constant $\bar{L}$. As the complement, the labour rate is lower bounded by a positive constant $\bar{L}-L$ for the consideration of retaining the employment state. By determining the continuous and stopping regions of the corresponding stopping time problem, we prove that the optimal bankruptcy time is the first hitting time of the wealth process downward to a critical wealth boundary.   
	
The idea is directly inspired by \cite{jeanblanc2004optimal}, in which a stochastic control model is constructed to quantify the benefit of filing consumer bankruptcy in the perspective of complete debt erasure. Their research is a response to the sharp growth in bankruptcy cases between 1978 and 2003 due to the promulgation of the 1978 Bankruptcy Reform Act in American. The Act introduced two kinds of consumer bankruptcy mechanisms, which are reflected in its Chapter 7 and Chapter 13 separately: debtors following Chapter 7 to file bankruptcy are granted the debt exemption, but must undertake the liquidation of non-exempt assets. Alternatively, the mechanism in Chapter 13 adopts the reorganization procedure instead of the liquidation. Debtors are permitted to retain assets, but the debt is required to be reorganized and paid continuously from future revenues. The statistical data shows that filing bankruptcy under Chapter 7 predominates in all consumer bankruptcy cases (1,156,274 out of 1,625,208 cases in 2003, accounting for 72\%).\footnote{Administrative Office of the U.S Courts, Table F-2— Bankruptcy Filings (December 31, 2003) [Online]. Available: https://www.uscourts.gov/statistics/table/f-2/bankruptcy-filings/2003/12/31} Following \cite{jeanblanc2004optimal}, an affine loss function, $\alpha(X(\tau)-F)$, is established to serve the fixed and variable costs of filing bankruptcy, which corresponds to the mathematical description for the bankruptcy mechanism under Chapter 7. Here $X(\tau)$ is the wealth level at the moment of bankruptcy $\tau$, and $F$ represents the fixed cost of bankruptcy. Therefore, the bankruptcy option reduces the wealth from $X(\tau)$ to $X(\tau^+)=\alpha(X(\tau)-F)$, with a drop in wealth equal to $(1-\alpha)X(\tau)+\alpha F$. The loss of the $(1-\alpha)$ proportion of the remaining wealth after the bankruptcy liquidation is related, for example, to taxes costs.  

Compared to \cite{jeanblanc2004optimal}, we make an extension in two aspects: firstly, a new control variable, the leisure rate, is inserted for a more realistic consideration; accordingly, the agent earns the labour income. For the introduction of the leisure as a control variable into the optimal stopping time problem, the reader can refer to \citep{choi2008optimal,farhi2007saving}, where authors studied the optimal retirement -from labour- model regarding the consumption-portfolio-leisure strategy. Different from these two researches, we consider the stopping time concerning the bankruptcy issue rather than the retirement: while the optimal retirement is the first hitting time of the wealth process to an upper critical wealth boundary \citep{ barucci2012optimal,choi2008optimal,farhi2007saving}, the optimal bankruptcy is related to a lower boundary. This extension permits us to study the impact of the disutility from full work on the bankruptcy option. Secondly, in order to deal with a utility from consumption and leisure rate, we implement a different method from \cite{jeanblanc2004optimal}, where the utility of the agent only depends on her consumption, solving the optimal problem with the duality method instead of the dynamic programming method, to deduce the solution analytically, as in \citep{barucci2012optimal,choi2008optimal}. We would like to stress that in this work we deal with the duality method applied to intertemporal consumption, for terminal utility problem the reader can refer, for example, to \citep{barucci2019optimal,colaneri2021value,nicolosi2018portfolio}. The duality method throughout this paper can be summarized into four steps. We first tackle the post-stopping time problem to deduce a closed form of the corresponding value function. Then we apply the Legendre-Fenchel transform to the utility function and the value function of post-bankruptcy time problem obtained in the first step. Afterwards, we construct the duality between the optimal control problem with the individual's shadow prices problem, by the aid of the liquidity and budget constraints and the dual transforms acquired before. Finally, we cast the dual shadow price problem as a free boundary problem, which leads to a system of variational inequalities and enables us to solve it analytically. The methodology discussed here refers to \citep{he1993labor,karatzas2000utility}: in \cite{karatzas2000utility} authors applied the duality method to solve a discretionary stopping time problem explicitly, while in \cite{he1993labor} authors used the duality approach to link the individual's shadow price problem with the optimal control problem and investigated the impact of the liquidity constraint on the optimization.

Other related literatures are \cite{karatzas1997explicit}, where authors studied the general optimal control problem involving the consumption and investment, and offered the solution in a closed-form, \cite{sethi1995explicit}, where a general continuous-time consumption and portfolio decision problem with a recoverable bankruptcy option is considered, and \cite{bellalah2019intertemporal} where authirs address the role of labour earnings in optimal asset allocation.

In the numerical results part, we first of all conduct the sensitivity analysis with respect to the key bankruptcy parameters $F$, $\alpha$ and $d$. The fixed toll $F$ and $(1-\alpha)$ can be treated as the fixed and flexible bankruptcy costs; the debt $d$ is the continuous-time debt repayment. As already said, the optimal bankruptcy corresponds to the first hitting time of the wealth process of a downward boundary, the bankruptcy wealth threshold. This threshold, as a function of the debt repayment $d$, is an increasing and convex curve. The rationale of this result is the following: a heavier debt repayment, in fact, implies that the benefit of bankruptcy becomes more attractive, therefore inducing the agent to take a higher threshold to make the bankruptcy requirement more accessible such that she can enjoy the debt exemption easily.  Furthermore, the convexity of the mapping can be explained by the fact that this motivation is diminishing as the debt repayment decreases. Similar results hold true if we consider the bankruptcy threshold as a function of $\alpha$, i.e., the proportion of wealth after the bankruptcy liquidation: a lower value of $\alpha$ indicates a higher flexible cost (a higher value of $(1-\alpha)$) and pushes the agent to set a lower wealth level to avoid suffering the bankruptcy. Our numerical results also show the non-monotonic relationship between the bankruptcy wealth threshold and $F$ itself; this is due to the role of $F$, which is not only the fixed cost of bankruptcy, but, according to the model in \cite{jeanblanc2004optimal}, in order to make the problem feasible, it also has an important role as liquidity constraint in the pre-bankruptcy period. Moreover, comparing the optimal control policies between the model with and without the bankruptcy option, we find that this additional option offers the agent a better circumstance such that the optimal consumption-portfolio-leisure policies dominate the ones without it before the bankruptcy. Whereas, after declaring bankruptcy, the agent suffers the wealth shrinkage and prefers to invest less in the risky asset for the needs of obtaining utility from consumption and leisure. In addition, we also study the impact of introducing the leisure rate as a second control variable, such that the influence of labour income can be disclosed. The numerical result indicates that the optimal consumption and portfolio policies with the flexible leisure option always prevail over the corresponding policies of the model with a full leisure rate, and therefore without labour income.

The paper is organized as follows. Section \ref{Section 2} formulates the corresponding optimization problem, and provides the financial market setting. Section \ref{Section 3} offers the value function of post-bankruptcy problem and its Legendre-Fenchel transform. In Section \ref{Section 4}, we construct the duality between the optimal control problem with the individual's shadow price, and obtain a free boundary problem which endows us the closed-form optimal solutions. Section \ref{Section 5} presents the numerical tests to this model and the sensitivity analysis of the bankruptcy wealth threshold to main parameters. Finally, Section \ref{Section 6} concludes. Most of the proofs and computations are reported in the online appendix. 

\section{Problem Formulation}\label{Section 2}
\subsection{Financial Market}
We first formulate the considered financial market over the infinite-time horizon. Based on the mutual fund theorem from \cite{karatzas1997explicit}, we consider only one risky asset which dynamics follows the Geometric Brownian Motion with constant drift and diffusion coefficients. The agent faces two investment opportunities: the investment in the money market, which endows her a fixed and positive interest rate $r>0$, and the risky asset, which dynamically evolves according to the stochastic differential equation (SDE)
\begin{equation}
\begin{cases}
dS(t)=\mu S(t)dt+\sigma S(t)dB(t),\\
S(0)=S_{0}.
\end{cases}\label{1.1}
\end{equation}
Here, $B(t)$ denotes a standard Brownian motion on the filtered probability space $(\Omega,\mathcal{F},\mathbb{P})$, $\{\mathcal{F}_{t},0\leq t<\infty\}$ is the augmented natural filtration on this Brownian motion, and $S_{0}$ represents the initial stock price, which is assumed to be a positive constant. Since the drift and diffusion terms $\mu$ and $\sigma$ are positive constants, there exists a unique solution to the SDE $(\ref{1.1})$, $S(t)=S_{0}e^{\left(\mu-\frac{1}{2}\sigma^{2}\right)t+\sigma B(t)}$.

Then referring to \cite{karatzas1998methods}, we introduce the state-price density process as $H(t)\triangleq\xi(t)\tilde{Z}(t)$,
with $\xi(t)$, $\tilde{Z}(t)$, the discount process and an exponential martingale, respectively defined as
\begin{equation*}
\begin{cases}
\xi(t)\triangleq e^{-rt}, \quad\mbox{with}\quad \xi(0)=1,\\
\tilde{Z}(t)\triangleq e^{-\frac{1}{2}\theta^{2}t-\theta B(t)}, \quad\mbox{with}\quad \tilde{Z}(0)=1.
\end{cases}
\end{equation*}
Moreover, $\theta\triangleq\frac{\mu-r}{\sigma}$ stands for the market price of risk, that is, the Sharpe-Ratio. Since the exponential martingale $\tilde{Z}(t)$ is, in fact, a $\mathbb{P}$-martingale, and both the number of risky assets and the dimension of the driving Brownian motions are equal to one, the financial market $\mathcal{M}$ defined with the above setting, $\mathcal{M}=\{(\Omega,\mathcal{F},\mathbb{P}),B,r,\mu,\sigma,S_{0}\}$, is standard and complete, based on the result from \cite[Section 1.7, Definition 7.3]{karatzas1998methods}. Additionally, we can define an equivalent martingale measure through $\tilde{\mathbb{P}}(A)\triangleq\mathbb{E}\left[\tilde{Z}(t)\mathbb{I}_{A}\right]$, $\forall A\in \mathcal{F}_{t}$. Then based on the Girsanov Theorem, we can get a standard Brownian motion under the $\tilde{\mathbb{P}}$ measure as
\begin{equation}
\tilde{B}(t)\triangleq B(t)+\theta t,\quad \forall t\ge 0.\label{1.2}
\end{equation}

\subsection{The Optimization Problem} 
The agent optimally chooses the consumption rate, the amount of money allocated in the risky asset and the leisure rate, which are denoted as $c(t)$, $\pi(t)$ and $l(t)$, treated as the three control variables in the optimization. The sum of the labour and leisure rate is constant and equals $\bar{L}$. Therefore, the working rate at time $t$ is $(\bar{L}-l(t))$ that enables the agent to earn a wage of $w(\bar{L}-l(t))$, where $w>0$ represents the constant wage rate. Obviously, the condition $0\leq l(t)\leq \bar{L}$ must be imposed for the positive labour income consideration. Furthermore, a realistic constraint is introduced into the model, that is, the working rate should be lower bounded by a positive constant $(\bar{L}-L)$ for the sake of retaining the employment state. 
 
Following \cite{jeanblanc2004optimal}, an affine loss function is introduced for accommodating fixed and variable costs of filing bankruptcy. Let $\tau$ denote the bankruptcy time, the agent is obliged to repay continuously a positive fixed debt $d$ until the stopping time $\tau$, whereas this debt obligation is exempted after declaring bankruptcy, but with the fixed cost $F>0$ and the variable cost $(1-\alpha)$, where the proportional coefficient $\alpha$ takes the value in $(0,1)$. In more detail, the agent needs to pay a fixed toll $F$ once for all at the time $\tau$, and the $(1-\alpha)$ proportion of the remaining wealth, which is related to the social cost, time cost and taxes cost of declaring bankruptcy. Therefore, the agent is able to keep the amount $\alpha(X(\tau)-F)$ of wealth for the consumption and investment after bankruptcy. For the purpose of making sure that the agent is capable of affording the bankruptcy, the wealth level is required to cover the cost before the stopping time $\tau$, that is, $X(t)\ge F+\eta$, $\forall t\in[0,\tau]$, where $\eta$ is a small non-negative constant to guarantee that there is still a few amounts of wealth left even after the liquidation. The bankruptcy mechanism described above entails the wealth process $X(t)$ to satisfy the following SDE
 \begin{equation*}
 \begin{cases}
 X(0)=x ,\\
 dX(t)=\left[rX(t)+\pi(t)\left(\mu-r\right)-c(t)-d+w(\bar{L}-l(t))\right]dt+\sigma\pi(t)dB(t),&t\leq\tau,\\
 X(\tau^{+})=\alpha(X(\tau)-F),\\
 dX(t)=\left[rX(t)+\pi(t)\left(\mu-r\right)-c(t)+w(\bar{L}-l(t))\right]dt+\sigma\pi(t)dB(t),&t>\tau.
 \end{cases}
 \end{equation*}

Furthermore, we assume that the agent's preference is described by a power utility function of consumption and leisure rate
 \begin{equation}
 u(c,l)=\frac{\left(c^{\delta}l^{1-\delta}\right)^{1-k}}{\delta(1-k)},\quad 0<\delta<1, \quad k>1.\label{1.3}
 \end{equation}
 Setting $k>1$ makes the mixed second partial derivative negative,
 \begin{equation*}
 \frac{\partial^{2}u(c,l)}{\partial c\partial l}=(1-k)(1-\delta)l^{(1-k)(1-\delta)-1}c^{\delta(1-k)-1}<0,
 \end{equation*} 
 which clarifies that consumption and leisure are substitute goods. The following lemma introduces the Legendre-Fenchel transform of the function $u(c,l)$, which will help us to reduce the number of control variables up to a single one. 
 Referring to \cite[Section 2.2]{choi2008optimal}, the Legendre-Fenchel transform of the utility function is defined as
 \begin{equation}
 \tilde{u}(y)\triangleq\sup_{c\ge 0,0\leq l\leq L} \left[u(c,l)-(c+wl)y\right],\label{1.4}
 \end{equation}
 and it is given in the following lemma.
 \begin{lemma}\label{Lemma 1}
 	The Legendre-Fenchel transform of the utility function $u(c,l)$ is
 	\begin{equation*}
 	\tilde{u}(y)=\left[A_{1}y^{\frac{\delta (1-k)}{\delta(1-k)-1}}-w L y\right]\mathbb{I}_{\{0<y<\tilde{y}\}}+\left[A_{2}y^{-\frac{1-k}{k}}\right]\mathbb{I}_{\{y\ge\tilde{y}\}},
 	\end{equation*}
 	with
 	\begin{equation*}
 	A_{1}\triangleq\frac{1\!-\!\delta\!+\!\delta k}{\delta (1\!-\!k)}L^{\!-\!\frac{(1\!-\!k)(1\!-\!\delta)}{\delta(1\!-\!k)\!-\!1}},\quad A_{2}\triangleq\frac{k}{\delta(1\!-\!k)}\left(\frac{1\!-\!\delta}{\delta w}\right)^{\frac{(1\!-\!k)(1\!-\!\delta)}{k}}, \quad \mbox{and}\quad \tilde{y}\triangleq L^{\!-k}\left(\frac{1\!-\!\delta}{\delta w}\right)^{1\!-\!\delta(1\!-\!k)}.
 	\end{equation*}
 	Furthermore, the consumption-leisure policy reaching the supremum in (\ref{1.4}) is
 	\begin{equation*}
 	\begin{cases}
 	\hat{c}=y^{-\frac{1}{k}}\left(\frac{1-\delta}{\delta w}\right)^{\frac{(1-k)(1-\delta)}{k}}\mathbb{I}_{\{y\ge\tilde{y}\}}+L^{-\frac{(1-k)(1-\delta)}{\delta(1-k)-1}}y^{\frac{1}{\delta(1-k)-1}}\mathbb{I}_{\{0<y<\tilde{y}\}},\\
 	\hat{l}=y^{-\frac{1}{k}}\left(\frac{1-\delta}{\delta w}\right)^{-\frac{\delta(1-k)-1}{k}}\mathbb{I}_{\{y\ge\tilde{y}\}}+L\mathbb{I}_{\{0<y<\tilde{y}\}}.
 	\end{cases}
 	\end{equation*}
 \end{lemma}
 \begin{proof}
 See Appendix \ref{A.1}.
 \end{proof}

In this framework, the primal optimization problem, which is denoted as $(P)$, is expressed as
 \begin{equation*}
 \begin{split}
 V(x)&\triangleq\sup_{(\tau,\{c(t),\pi(t),l(t)\}_{\scriptscriptstyle t\ge 0})\in\mathcal{A}(x)}J(x;c,\pi,l,\tau)\\
 &=\sup_{(\tau,\{c(t),\pi(t),l(t)\}_{\scriptscriptstyle t\ge 0})\in\mathcal{A}(x)} \mathbb{E} \left[\int_{0}^{\infty}e^{-\gamma t}u(c(t),l(t))dt\right],
 \end{split}\tag{$P$}
 \end{equation*}
 $\mathcal{A}(x)$ stands for the admissible control set and follows the definition below.
 
 \begin{remark}
 The above framework is consistent with an infinitely lived agent or an agent which death is modelled as the first jump time of an independent Poisson process. In the first case, $\gamma$ is the subjective discount rate. In the second case, we have $\gamma=\hat{\gamma}+\lambda_{D}$, where $\hat{\gamma}$ is the subjective discount rate and $\lambda_{D}$ is the intensity of the Poisson process. In fact, if $\tau_{D}$ is the time in which the death occurs, we have
 \begin{equation*}
\mathbb{E} \left[\int_{0}^{\tau_D{}}e^{-\hat{\gamma} t}u(c(t),l(t))dt\right]= \mathbb{E} \left[\int_{0}^{+\infty}e^{-\hat{\gamma} t}e^{-{\lambda_{D}} t}u(c(t),l(t))dt\right]= \mathbb{E} \left[\int_{0}^{+\infty}e^{-{\gamma} t}u(c(t),l(t))dt\right],
\end{equation*}
 due to the independence of the Poisson process.
 \end{remark}
 
 \begin{definition}\label{Definition 1}
 	Given the initial wealth $x\ge F+\eta$, $\mathcal{A}(x)$ is defined as the set of all admissible policies satisfying:
 	\begin{itemize}
 		\item  $\tau\leq\infty$ is an $\mathcal{F}_{t}$-stopping time,
 		\item  $\{c(t)\!:\!t\!\ge\! 0\}$ is an $\mathcal{F}_{t}$-progressively measurable and non-negative process such that $\int_{0}^{t}\!c(s)ds\!<\!\infty$, a.s., $\forall t\ge 0$,
 		\item $\{\pi(t):t\ge 0\}$ is an $\mathcal{F}_{t}$-progressively measurable process such that $\int_{0}^{t}\pi^{2}(s)ds<\infty$, a.s., $\forall t\ge 0$,
 		\item  $\{l(t):t\ge 0\}$ is an $\mathcal{F}_{t}$-progressively measurable and non-negative process such that $0\leq l(t)\leq L$, $\forall t\ge 0$,
 		\item $X(t)\ge F+\eta$ for $0\leq t\leq\tau$, and $X(t)\ge 0$ for  $t>\tau$ a.s.,
 		\item $\mathbb{E}\left[\int_{0}^{\infty}e^{-\gamma t}u^{-}(c(t),l(t))dt\right]<\infty$ with $u^{-}\triangleq-\min(u,0)$.
 	\end{itemize}
 \end{definition} 
\noindent Moreover, we assume $F+\eta\geq \frac{d-w\bar{L}}{r}$, since $\frac{d-w\bar{L}}{r}$ represents the market value of the debt repayment reduced by the maximum amount to borrow against the future income in the pre-bankruptcy period: the agent is therefore unable to allocate the investment and consumption when the wealth level stays below it. 

The subsequent proposition provides the corresponding budget constraint.
\begin{proposition}\label{Proposition 1}
	Given any initial wealth $x\ge F+\eta$, any strategy $\left(\tau,\{c(t),\pi(t),l(t)\}_{t\ge 0}\right)\in\mathcal{A}(x)$, the budget constraint is given by 
	\begin{equation}
	\mathbb{E}\left[\int_{0}^{\tau}H(t)\left(c(t)+d+wl(t)-w\bar{L}\right)dt+H(\tau)X(\tau)\right]\leq x.\label{1.5}
	\end{equation}
\end{proposition}
\begin{proof}
See Appendix \ref{A.2}.
\end{proof}

Completing the construction of the primal optimization problem, we are going to solve it explicitly in the following sections. The gain function of the primal optimization problem $(P)$ can be rewritten as the expectation of two separated terms representing the pre- and post-bankruptcy part,
\begin{equation*}
\begin{split}
J(x;c,\pi,l,\tau)&=\mathbb{E} \left[\int_{0}^{\infty}e^{-\gamma t}u(c(t),l(t))dt\right]\\
&=\mathbb{E} \left[\int_{0}^{\tau}e^{-\gamma t}u(c(t),l(t))dt+e^{-\gamma \tau}\mathbb{E}\left[\left.\int_{\tau}^{\infty}e^{-\gamma (t-\tau)}u(c(t),l(t))dt\right|\mathcal{F}_{\tau}\right]\right]\\
&=\mathbb{E} \left[\int_{0}^{\tau}e^{-\gamma t}u(c(t),l(t))dt+e^{-\gamma \tau}J_{\scriptscriptstyle PB}(\alpha(X(\tau)-F);c,\pi,l)\right],
\end{split}
\end{equation*}
where we define $J_{\scriptscriptstyle PB}(X(t);c,\pi,l)\triangleq\mathbb{E} \left[\left.\int_{t}^{\infty}e^{-\gamma (s-t)}u(c(s),l(s))ds\right| \mathcal{F}_{t}\right]$ in the post-bankruptcy framework, i.e., no debt repayment. We perform the backward approach, hence, begin with the post-bankruptcy part by means of the dynamic programming principle.

\section{Post-Bankruptcy Problem}\label{Section 3}
We first tackle the post-bankruptcy problem, assuming without loss of generalization $\tau=0$, which is the optimization over an infinite time horizon through controlling the investment amount, consumption and leisure rate. Since the debt repayment is removed from the wealth process after the bankruptcy, the corresponding dynamics becomes
\begin{equation*}
\begin{cases}
d X(t)=[rX(t)+\pi(t)(\mu-r)-c(t)+w(\bar{L}-l(t))]dt+\sigma \pi(t)dB(t),\\
X(0)=x.
\end{cases}
\end{equation*}
Afterwards, based on the gain function $J_{\scriptscriptstyle PB}(\cdot)$ defined at the end of the previous section, we can express the value function of the post-bankruptcy part as follows,
\begin{equation*}
\begin{split}
V_{\scriptscriptstyle PB}(x)&\triangleq\sup_{\{c(t),\pi(t),l(t)\}_{\scriptscriptstyle t\ge 0}\in\mathcal{A}_{\scriptscriptstyle PB}(x)}J_{\scriptscriptstyle PB}(x;c,\pi,l)\\
&=\sup_{\{c(t),\pi(t),l(t)\}_{\scriptscriptstyle t\ge 0}\in\mathcal{A}_{\scriptscriptstyle PB}(x)}\mathbb{E} \left[\int_{0}^{\infty}e^{-\gamma t}u(c(t),l(t))dt\right].
\end{split}\tag{$P_{\scriptscriptstyle PB}$}
\end{equation*}
The admissible control set $\mathcal{A}_{\scriptscriptstyle PB}(x)$ is compatible with Definition \ref{Definition 1}, only removing the condition about the stopping time and changing the liquidity condition from ``$X(t)\ge F+\eta$ for $0\leq t\leq\tau$, and $X(t)\ge 0$ for  $t>\tau$ a.s.'' to ``$X(t)\ge 0$ for $t\ge 0$ a.s.''. Additionally, for any given initial endowment $x\ge 0$ and admissible consumption-portfolio-leisure strategy $\{c(t),\pi(t),l(t)\}_{\scriptscriptstyle t\ge 0}\in \mathcal{A}_{\scriptscriptstyle PB}(x)$, the following propositions provide us the budget and liquidity constraint to the post-bankruptcy problem.
\begin{proposition}
	The infinite horizon budget constraint of the post-bankruptcy problem is
	\begin{equation*}
	\mathbb{E}\left[\int_{0}^{\infty}H(t)(c(t)+wl(t)-w\bar{L})dt\right]\leq x.
	\end{equation*}
\end{proposition} 
\begin{proof}
Similarly as in Appendix \ref{A.2}, we can prove that $\mathbb{E}\left[\int_{0}^{t}H(s)\left(c(s)+wl(s)-w\bar{L}\right)ds\right]\leq x$. Then the above budget constraint can be obtained by taking the limit as $t\to\infty$.
\end{proof}

\begin{proposition}
	The infinite horizon liquidity constraint of the post-bankruptcy problem is
	\begin{equation*}
	\mathbb{E}\left[\left.\int_{t}^{\infty}\frac{H(s)}{H(t)}(c(s)+wl(s)-w\bar{L})ds\right|\mathcal{F}_{t}\right]\ge 0.
	\end{equation*}
\end{proposition} 
\begin{proof}
The result directly comes from \cite[Section 3.9, Theorem 9.4]{karatzas1998methods}, more precisely, the non-negative property of $X(t)$ and $\mathbb{E}\left[\left.\int_{t}^{\infty}\frac{H(s)}{H(t)}(c(s)+wl(s)-w\bar{L})ds\right|\mathcal{F}_{t}\right]=X(t)$.
\end{proof}
Then, we implement the methodology presented in \citep{karatzas2000utility,he1993labor} to establish a duality between the optimal control problem and the individual's shadow price problem through the Lagrange method. We make the following derivation of $J_{\scriptscriptstyle PB}(x;c,\pi,l)$, introducing a non-increasing process $D_{\scriptscriptstyle PB}(t)\!\ge\! 0$ and a Lagrange multiplier $\lambda$,
\begin{equation*}
\begin{split}
J_{\scriptscriptstyle PB}(x;c,\pi,l)
&=\mathbb{E}\left[\int_{0}^{\infty}e^{-\gamma t}\left(u(c(t),l(t))-(c(t)+wl(t))\lambda e^{\gamma t}D_{\scriptscriptstyle PB}(t)H(t) \right)dt\right]\\
&\qquad+\lambda \mathbb{E}\left[\int_{0}^{\infty}(c(t)+wl(t))D_{\scriptscriptstyle PB}(t)H(t)dt\right]\\
&\leq \mathbb{E} 
\left[\int_{0}^{\infty}e^{-\gamma t}\tilde{u}(\lambda e^{\gamma t}D_{\scriptscriptstyle PB}(t)H(t))dt\right]+\lambda \mathbb{E}\left[\int_{0}^{\infty}(c(t)+wl(t))D_{\scriptscriptstyle PB}(t)H(t)dt\right]\\
&=\mathbb{E}
\left[\int_{0}^{\infty}e^{-\gamma t}\left(\tilde{u}(\lambda e^{\gamma t}D_{\scriptscriptstyle PB}(t)H(t))+w\bar{L}\lambda e^{\gamma t}D_{\scriptscriptstyle PB}(t)H(t)\right)dt\right]\\
&\qquad+\lambda \mathbb{E}\left[\int_{0}^{\infty}(c(t)+wl(t)-w\bar{L})D_{\scriptscriptstyle PB}(t)H(t)dt\right].
\end{split}
\end{equation*}
By the Fubini's Theorem, see e.g. \cite[Appendix A, Theorem A.48]{bjork2009arbitrage}, we have
\begin{equation*}
\begin{split}
\int_{0}^{\infty}H(t)D_{\scriptscriptstyle PB}(t)(c(t)+wl(t)-w\bar{L})dt=& \int_{0}^{\infty}H(t)(c(t)+wl(t)-w\bar{L})dt\\
&\quad+\!\int_{0}^{\infty}\!H(t)\!\int_{t}^{\infty}\!\frac{H(s)}{H(t)}(c(s)\!+\!wl(s)\!-\!w\bar{L})dsdD_{\scriptscriptstyle PB}(t),
\end{split}
\end{equation*}
and the inequality concerning $J_{\scriptscriptstyle PB}(x;c,\pi,l)$ can be rewritten as
\begin{equation*}
\begin{split}
J_{\scriptscriptstyle PB}(x;c,\pi,l)
&\leq\mathbb{E}\left[\int_{0}^{\infty}e^{-\gamma t}\left(\tilde{u}(\lambda e^{\gamma t}D_{\scriptscriptstyle PB}(t)H(t))+w\bar{L}\lambda e^{\gamma t}D_{\scriptscriptstyle PB}(t)H(t)\right)dt\right]\\
&\qquad+\lambda\mathbb{E}\left[\int_{0}^{\infty}H(t)(c(t)+wl(t)-w\bar{L})dt\right]\\
&\qquad+\lambda\mathbb{E}\left[\int_{0}^{\infty}H(t)\mathbb{E}\left[\left.\int_{t}^{\infty}\frac{H(s)}{H(t)}(c(s)+wl(s)-w\bar{L})ds\right|\mathcal{F}_{t}\right]dD_{\scriptscriptstyle PB}(t)\right]\\
&\leq\mathbb{E}\left[\int_{0}^{\infty}e^{-\gamma t}\left(\tilde{u}(\lambda e^{\gamma t}D_{\scriptscriptstyle PB}(t)H(t))+w\bar{L}\lambda e^{\gamma t}D_{\scriptscriptstyle PB}(t)H(t)\right)dt\right]+\lambda x,
\end{split}
\end{equation*}
the last inequality is derived from the budget constraint, the liquidity constraint and the non-increasing property of $D_{\scriptscriptstyle PB}(t)$.
Referring to \cite[Section 4]{he1993labor}, we can define the corresponding individual's shadow price problem $(S_{\scriptscriptstyle PB})$ as below,
\begin{equation*}\tag{$S_{\scriptscriptstyle PB}$}
\begin{split}
\tilde{V}_{\scriptscriptstyle PB}(\lambda)\triangleq\inf_{\{D_{\scriptscriptstyle PB}(t)\}_{\scriptscriptstyle t\ge 0}\in\mathcal{D}}& \mathbb{E}\left[\int_{0}^{\infty}e^{-\gamma t}\left(\tilde{u}(\lambda e^{\gamma t}D_{\scriptscriptstyle PB}(t)H(t))+w\bar{L}\lambda e^{\gamma t}D_{\scriptscriptstyle PB}(t)H(t)\right)dt\right],
\end{split}
\end{equation*} 
where $\mathcal{D}$ represents the set of
non-negative, non-increasing and progressively measurable processes. Then, we put forward a theorem to construct the duality between the optimal consumption-portfolio-leisure problem $(P_{\scriptscriptstyle PB})$ and the individual's shadow price problem $(S_{\scriptscriptstyle PB})$.

\begin{theorem}\label{Theorem 1}
	\textbf{(Duality Theorem)}
	Suppose $D_{\scriptscriptstyle PB}^{*}(t)$ is the optimal solution (the $\text{arg\,inf}$) to the dual shadow price problem $(S_{\scriptscriptstyle PB})$, then the optimal consumption-leisure strategy to the primal problem $(P_{\scriptscriptstyle PB})$ satisfies $c^{*}(t)+wl^{*}(t)=-\tilde{u}^{\prime}(\lambda^* e^{\gamma t}D_{\scriptscriptstyle PB}^{*}(t)H(t))$, and we have the following relation
	\begin{equation*}
	V_{\scriptscriptstyle PB}(x)=\inf_{\lambda>0} \left\{\tilde{V}_{\scriptscriptstyle PB}(\lambda)+\lambda x\right\}, \qquad\forall x\ge 0.
	\end{equation*} 
	Here $\lambda^*$ is the parameter $\lambda$ which gives the infimum in the above equation.
\end{theorem}
\begin{proof}
See Appendix \ref{B.1}.
\end{proof}

\noindent It can be known from the above duality theorem that the optimal solution of the post-bankruptcy problem is transformed into finding the optimal $D_{\scriptscriptstyle PB}^{*}
(t)$. In order to solve the problem $(S_{\scriptscriptstyle PB})$ explicitly, we follow the approach in \cite{davis1990portfolio} and first provide the subsequent assumption.
\begin{assumption}
	The non-increasing process $D_{\scriptscriptstyle PB}(t)$ is absolutely continuous with respect to t. Hence, there exists a process $\psi_{\scriptscriptstyle PB}(t)\ge 0$ such that $dD_{\scriptscriptstyle PB}(t)=-\psi_{\scriptscriptstyle PB}(t)D_{\scriptscriptstyle PB}(t)dt$.
\end{assumption}

Defining a new process $Z(t) \triangleq \lambda e^{\gamma t} D_{\scriptscriptstyle PB}(t)H(t)$, the value function of Problem $(S_{\scriptscriptstyle PB})$ can be rewritten as
$\tilde{V}_{\scriptscriptstyle PB}(\lambda)=\inf\limits_{\psi_{\scriptscriptstyle PB}(t)\ge 0} \mathbb{E}\left[\int_{0}^{\infty}e^{-\gamma t}(\tilde{u}(Z(t))+w\bar{L}Z(t))dt\right]$, where $Z(t)$ is the state variable, and follows the dynamics
\begin{equation}\label{Zeta}
\begin{cases}
\frac{dZ(t)}{Z(t)}=\frac{dD_{\scriptscriptstyle PB}(t)}{D_{\scriptscriptstyle PB}(t)}+(\gamma-r)dt-\theta dB(t)=-\psi_{\scriptscriptstyle PB}(t)dt+(\gamma-r)dt-\theta dB(t),\\
Z(0)=\lambda>0.
\end{cases}
\end{equation}
Then we introduce a new function $\phi_{\scriptscriptstyle PB}:(\mathbb{R}^{+},\mathbb{R}^{+})\mapsto \mathbb{R}$ as
\begin{equation}
\phi_{\scriptscriptstyle PB}(t,z)\triangleq\inf_{\psi_{\scriptscriptstyle PB}(t)\ge 0} \mathbb{E}\left[\left.\int_{t}^{\infty} e^{-\gamma s} \left(\tilde{u}(Z(s))+w\bar{L}Z(s)\right)ds\right|Z(t)=z\right],\label{1.6}
\end{equation}
which implies that $\tilde{V}_{\scriptscriptstyle PB}(\lambda)=\phi_{\scriptscriptstyle PB}(0,\lambda)$. The Bellman equation associated to $\phi_{\scriptscriptstyle PB}(t,z)$ is
\begin{equation}
\min_{\psi_{\scriptscriptstyle PB}\ge 0}\left\{\tilde{\mathcal{L}}\phi_{\scriptscriptstyle PB}(t,z)+e^{-\gamma t}(\tilde{u}(z)+w\bar{L}z)-\gamma\phi_{\scriptscriptstyle PB}(t,z)-\psi_{\scriptscriptstyle PB} z\frac{\partial \phi_{\scriptscriptstyle PB}}{\partial z}(t,z) \right\}=0,\label{1.7}
\end{equation}
with the operator $\tilde{\mathcal{L}}=(\gamma-r)z\frac{\partial}{\partial z}+\frac{1}{2}\theta^{2}z^{2}\frac{\partial^{2}}{\partial z^{2}}$. The above Bellman equation makes the optimum $\psi_{\scriptscriptstyle PB}^{*}$ possess the following characterizations:
$$\frac{\partial \phi_{\scriptscriptstyle PB}}{\partial z}(t,z)=0 \Rightarrow \psi_{\scriptscriptstyle PB}^{*}\ge 0;\qquad \frac{\partial \phi_{\scriptscriptstyle PB}}{\partial z}(t,z)\leq 0 \Rightarrow \psi_{\scriptscriptstyle PB}^{*}= 0.$$
Moreover, the Bellman Equation (\ref{1.7}) can be transformed into
\begin{equation*}
\min \left\{\tilde{\mathcal{L}}\phi_{\scriptscriptstyle PB}(t,z)-\gamma\phi_{\scriptscriptstyle PB}(t,z)+e^{-\gamma t}(\tilde{u}(z)+w\bar{L}z),-\frac{\partial \phi_{\scriptscriptstyle PB}}{\partial z}(t,z)\right\}=0,
\end{equation*}
which corresponds to the variational inequalities: find a function ${\phi}_{\scriptscriptstyle PB}(\cdot,\cdot)\in C^{2}((0,\infty)\times\mathbb{R}^{+})$ and a free-boundary $\hat{z}_{PB}$, satisfying
\begin{equation}\label{1.8}
\begin{cases}
(V1) \quad\frac{\partial {\phi}_{\scriptscriptstyle PB}}{\partial z}(t,z)=0, & z\ge\hat{z}_{\scriptscriptstyle PB},\\
(V2) \quad \frac{\partial{\phi}_{\scriptscriptstyle PB}}{\partial z}(t,z)\leq 0, & 0<z<\hat{z}_{\scriptscriptstyle PB},\\
(V3)\quad \tilde{\mathcal{L}}{\phi}_{\scriptscriptstyle PB}(t,z)-\gamma{\phi}_{\scriptscriptstyle PB}(t,z)+e^{-\gamma t}(\tilde{u}(z)+w\bar{L}z)=0, & 0<z<\hat{z}_{\scriptscriptstyle PB},\\
(V4)\quad \tilde{\mathcal{L}}{\phi}_{\scriptscriptstyle PB}(t,z)-\gamma{\phi}_{\scriptscriptstyle PB}(t,z)+e^{-\gamma t}(\tilde{u}(z)+w\bar{L}z)\ge 0, & z\ge\hat{z}_{\scriptscriptstyle PB},
\end{cases}
\end{equation}
for any $t\ge 0$, with the smooth fit conditions $\frac{\partial{\phi}_{\scriptscriptstyle PB}}{\partial z}(t,\hat{z}_{\scriptscriptstyle PB})=0$, $\frac{\partial^{2}{\phi}_{\scriptscriptstyle PB}}{\partial z^{2}}(t,\hat{z}_{\scriptscriptstyle PB})=0$. In line with \cite[Appendix A]{choi2008optimal}, we assume that ${\phi}_{\scriptscriptstyle PB}(t,z)$ takes the time-independent form ${\phi}_{\scriptscriptstyle PB}(t,z)=e^{-\gamma t}v_{\scriptscriptstyle PB}(z)$ for solving the above variational inequalities
explicitly. The solution is computed in Appendix \ref{B.2} in a semi-analytical framework, i.e., as the solution of a non-linear system of equations. Once the function  ${\phi}_{\scriptscriptstyle PB}(\cdot,\cdot)$ is computed, and therefore $\tilde{V}_{\scriptscriptstyle PB}(\cdot)$ is known, we can recover ${V}_{\scriptscriptstyle PB}(\cdot)$ from Theorem \ref{Theorem 1}.

\section{Primal Optimization Problem}\label{Section 4}
Once we solved the post-bankruptcy problem, we first of all have to deal with the jump in the wealth at bankruptcy time. We introduce a subset of the primal optimization problem's admissible control set, $\mathcal{A}_{1}(x)\subset\mathcal{A}(x)$, inside which any policy maximises the gain function of the post-bankruptcy problem. That is to say, for any $(\tau, \{c(t),\pi(t),l(t) \})\in\mathcal{A}_{1}(x)$, the following holds,
\begin{equation*}
\mathbb{E}\left[\int_{\tau}^{\infty}e^{-\gamma t}u(c(t),l(t))dt\right]=\mathbb{E}\left[e^{-\gamma \tau}V_{\scriptscriptstyle PB}(\alpha(X^{x,c,\pi,l}(\tau)-F))\mathbb{I}_{\{\tau<\infty\}}\right].
\end{equation*}
Here $X^{x,c,\pi,l}(\tau)$ is the wealth at time $\tau$ given an initial wealth $x$ and assuming the policies $c,\pi,l$ for consumption, allocation in the risky asset and leisure rate, respectively.
Then, from the dynamic programming principle, the whole optimization problem is converted into
\begin{equation*}
V(x)\triangleq \sup_{(\tau, \{c(t),\pi(t),l(t) \}_{\scriptscriptstyle t\ge 0})\in\mathcal{A}_{1}(x)} \mathbb{E}\left[\int_{0}^{\tau}e^{-\gamma t}u(c(t),l(t))dt+e^{-\gamma\tau}U(X^{x,c,\pi,l}(\tau))\right],
\end{equation*} 
denoting the value function at the moment of bankruptcy as
\begin{equation*}
U(X^{x,c,\pi,l}(\tau))\triangleq \sup_{\{c(t),\pi(t),l(t) \}_{\scriptscriptstyle t\ge 0}\in\mathcal{A}_{1}(x)} \mathbb{E}\left[\left.\int_{\tau}^{\infty}e^{-\gamma (s-\tau)}u(c(s),l(s))ds\right
|\mathcal{F}_{\tau}\right].
\end{equation*}
Therefore, it can be observed that the relationship between $U(\cdot)$ and the post-bankruptcy value function $V_{\scriptscriptstyle PB}(\cdot)$ is
$$U(X^{x,c,\pi,l}(\tau))=V_{\scriptscriptstyle PB}(\alpha(X^{x,c,\pi,l}(\tau)-F)).$$
Simple computations give us the Legendre-Fenchel transform of $U(x)$, that is, $\tilde{U}(z)$
\begin{equation}\label{pico}
\tilde{U}(z)=\tilde{V}_{\scriptscriptstyle PB}\left(\frac{z}{\alpha}\right)-Fz={\phi}_{\scriptscriptstyle PB}\left(0,\frac{z}{\alpha}\right)-Fz.
\end{equation}

After obtaining the optimal solution for the post-bankruptcy problem, we now reduce the primal optimization problem by fixing the stopping time. Defining a set of admissible controls corresponding to a fixed stopping time $\tau\in\mathcal{T}$ as
\begin{equation*}
\mathcal{A}_{\tau}(x)\triangleq\left\{\{c(t),\pi(t),l(t)\}_{t\ge 0}: \left(\tau,\{c(t),\pi(t),l(t)\}_{t\ge 0}\right)\in\mathcal{A}(x) \,\mbox{for any fixed} \, \tau\in\mathcal{T} \right\},
\end{equation*}
and a utility maximization problem as
\begin{equation*}
V_{\tau}(x)\triangleq\sup_{\{c(t),\pi(t),l(t)\}_{t\ge 0}\in\mathcal{A}_{\tau}(x) } J(x;c,\pi,l,\tau),\tag{$P_{\tau}$}
\end{equation*} 
the problem $(P)$ can be transformed into an optimal stopping time problem, $V(x)=\sup\limits_{\tau\in\mathcal{T}}V_{\tau}(x)$. Then, we put forward the liquidity constraint for the optimal bankruptcy problem.
\begin{proposition}\label{Proposition 2}
	The liquidity constraint of the considered problem is
	\begin{equation}
	\mathbb{E}\left[\left.\int_{t}^{\tau}\frac{H(s)}{H(t)}\left(c(s)+d+wl(s)-w\bar{L}\right)ds+\frac{H(\tau)}{H(t)}X(\tau)\right|\mathcal{F}_{t}\right]\ge F+\eta, \quad\forall t\in[0,\tau].\label{1.9}
	\end{equation}
\end{proposition}
\begin{proof}
See Appendix \ref{C.1}.
\end{proof}
Following the same technique as in the post-bankruptcy problem, considering the budget and liquidity constraints (\ref{1.5}) and (\ref{1.9}), we introduce a real number $\lambda > 0$, the Lagrange multiplier, and a non-increasing continuous process $D(t)>0$. We obtain
\begin{equation*}
\begin{split}
J(x;c,\pi,l,\tau)
&\leq\mathbb{E}\bigg[\int_{0}^{\tau}e^{-\gamma t}\left(\tilde{u}(\lambda D(t)e^{\gamma t}H(t))-(d-w\bar{L})\lambda e^{\gamma t}D(t)H(t)\right)dt\\
&\qquad+e^{-\gamma\tau}\tilde{U}(\lambda D(\tau)e^{\gamma\tau}H(\tau))\bigg]
+\lambda\mathbb{E}\left[\int_{0}^{\tau}(F+\eta)H(t)dD(t)\right]+\lambda x.
\end{split}
\end{equation*}
As in \cite[Section 4]{he1993labor}, the individual's shadow price problem inspired by the above inequality is defined as below,
\begin{equation*}\tag{$S_{\tau}$}
\begin{split}
\tilde{V}_{\tau}(x)
&\triangleq \inf_{\{D(t)\}_{t\ge 0}\in\mathcal{D}}\mathbb{E}\left[\int_{0}^{\tau}e^{-\gamma t}\left(\tilde{u}(\lambda D(t)e^{\gamma t}H(t))-(d-w\bar{L})\lambda e^{\gamma t}D(t)H(t)\right)dt\right.\\
&\qquad\left.+e^{-\gamma\tau}\tilde{U}(\lambda D(\tau)e^{\gamma\tau}H(\tau))\right]
+\lambda\mathbb{E}\left[\int_{0}^{\tau}(F+\eta)H(t)dD(t)\right].
\end{split} 
\end{equation*}
Hereafter, we provide a theorem to construct the duality between the individual's shadow price problem $(S_{\tau})$ and the optimal consumption-portfolio-leisure problem $(P_{\tau})$.
\begin{theorem}\label{Theorem 2}
	\textbf{(Duality Theorem)}
	Suppose $D^{*}(t)$ is the optimal solution (the $\text{arg\,inf}$) to Problem $(S_{\tau})$, then, $c^{*}(t)+wl^{*}(t)=-\tilde{u}^{\prime}(\lambda^* e^{\gamma t}D^{*}(t)H(t))$ and $X^{x,c^{*},\pi^{*},l^{*}}(\tau)=-\tilde{U}^{\prime}(\lambda^* e^{\gamma\tau}D^{*}(\tau)H(\tau))$ coincide with the optimal solution to Problem $(P_{\tau})$, and we have the following relationship
	\begin{equation*}
	V_{\tau}(x)=\inf_{\lambda>0} \left\{\tilde{V}_{\tau}(\lambda)+\lambda x\right\}, \qquad\forall x\ge F+\eta.
	\end{equation*} 
	Here $\lambda^*$ is the parameter $\lambda$ which gives the infimum in the above equation.
\end{theorem}
\begin{proof}
See Appendix \ref{C.2}.
\end{proof} 

Furthermore, the duality theorem makes the value function of Problem $(P)$ conforms to the following derivation,
\begin{equation*}
V(x)=\sup_{\tau\in\mathcal{T}}V_{\tau}(x)=\sup_{\tau\in\mathcal{T}}\inf_{\lambda>0}\{\tilde{V}_{\tau}(\lambda)+\lambda x\}\leq\inf_{\lambda>0}\sup_{\tau\in\mathcal{T}}\{\tilde{V}_{\tau}(\lambda)+\lambda x\}=\inf_{\lambda>0}\{\sup_{\tau\in\mathcal{T}}\tilde{V}_{\tau}(\lambda)+\lambda x\}.
\end{equation*}
Let us introduce a new function 
$\tilde{V}(\lambda)\triangleq\sup\limits_{\tau\in\mathcal{T}}\tilde{V}_{\tau}(\lambda)$.
As in \cite[Section 8, Theorem 8.5]{karatzas2000utility}, the value function $V(x)$ satisfies $V(x)=\inf\limits_{\lambda>0} [\tilde{V}(\lambda)+\lambda x]$ under the condition that $\tilde{V}(\lambda)$ exists and is differentiable for any $\lambda> 0$. Hence, the objective optimization problem contains two steps:
\begin{equation*}
\begin{cases}
\tilde{V}(\lambda)=\sup\limits_{\tau \in \mathcal{T}}\tilde{V}_{\tau}(\lambda),\\
V(x)=\inf\limits_{\lambda>0} [\tilde{V}(\lambda)+\lambda x]=\tilde{V}(\lambda^{*})+\lambda^{*} x.
\end{cases}
\end{equation*}
The first step involves an optimal stopping time problem, and the second refers to obtain the
optimum $\lambda^{*}$ achieving the infimum part. We begin with the first optimization problem related to
the individual's shadow price. Before this, following the method of \cite{davis1990portfolio}, an assumption is imposed on the process $D(t)$.
\begin{assumption}
	The non-increasing process $D(t)$ is absolutely continuous with respect to t. Hence, there is a non-negative process $\psi(t)$ such that $dD(t)=-\psi(t)D(t)dt$.
\end{assumption}

Introducing a new process $Z(t)\triangleq\lambda D(t)e^{\gamma t}H(t)$, the value function of the dual problem $(S_{\tau})$ is rewritten as
\begin{equation*}
\tilde{V}_{\tau}(\lambda)=\inf_{\psi(t)\ge 0}\mathbb{E}\left[\int_{0}^{\tau} e^{-\gamma t} \left(\tilde{u}(Z(t))-(d-w\bar{L})Z(t)-(F+\eta)\psi(t)Z(t)\right)dt+e^{-\gamma \tau}\tilde{U}(Z(\tau))\right].
\end{equation*}
We consider a new generalized optimization problem 
\begin{equation*}
\phi(t,\!z)\!\triangleq\!\sup_{\tau\ge t} \!\inf_{\psi(t)>0}\! \mathbb{E}\!\left[\!\left.\int_{t}^{\tau}\!\!\!\!\!e^{\!-\!\gamma s}\! \left(\tilde{u}(Z(s))\!-\!(d\!-\!w\bar{L})Z(s)\!-\!(F\!+\!\eta)\psi(s)Z(s)\right)ds\!+\!e^{\!-\!\gamma \tau}\!\tilde{U}(Z(\tau))\right|\!Z(t)\!=\!z\right],
\end{equation*}
it can be observed that  $\tilde{V}(\lambda)=\phi(0,\lambda)$, which indicates the solution of $\tilde{V}(\lambda)$ is resorted to solve the above generalized problem.  We start with the infimum part through defining
\begin{equation*}
\phi_{\scriptscriptstyle inf}(t,z)\!\triangleq\!\!\!\inf_{\psi(t)>0} \!\!\mathbb{E}\!\left[\left.\int_{t}^{\tau} \!e^{-\gamma s}\! \left(\tilde{u}(Z(s))\!\!-\!\!(d\!-\!w\bar{L})Z(s)\!\!-\!\!(F\!+\!\eta)\psi(s)Z(s)\right)\!ds\!+\!e^{-\gamma \tau}\!\tilde{U}(Z(\tau))\right|\!Z(t)\!=\!z\right],
\end{equation*}
the dynamic programming principle gives us the subsequent Bellman equation
\begin{equation*}
\min_{\psi\ge 0} \left\{\mathcal{L}\phi_{\scriptscriptstyle inf}(t,z)+e^{-\gamma t}\left(\tilde{u}(z)-(d-w\bar{L})z\right)-\psi z\left[\frac{\partial \phi_{\scriptscriptstyle inf}}{\partial z}(t,z)+(F+\eta)e^{-\gamma t}\right]\right\}=0,
\end{equation*}
with introducing the operator $\mathcal{L}=\frac{\partial}{\partial  t}+(\gamma-r)z\frac{\partial}{\partial z}+\frac{1}{2}\theta^{2}z^{2}\frac{\partial^{2}}{\partial z^{2}}$. The following characterizations hold for the optimum $\psi^{*}$ (the $\text{arg\,min}$):
\begin{itemize}
	\item $\frac{\partial \phi_{\scriptscriptstyle inf}}{\partial z}(t,z)+(F+\eta)e^{-\gamma t}=0 \Longrightarrow \psi^{*}\ge 0$: then we get that there exists a boundary $\hat{z}$ such that
	\begin{equation}
	\frac{\partial \phi}{\partial z}(t,z)=\frac{\partial \phi_{\scriptscriptstyle inf}}{\partial z}(t,z)=-(F+\eta)e^{-\gamma t},\quad z\ge\hat{z}.\label{1.11}
	\end{equation}
	\item $\frac{\partial \phi_{\scriptscriptstyle inf}}{\partial z}(t,z)+(F+\eta)e^{-\gamma t}\leq 0 \Longrightarrow \psi^{*}=0$: in this case, $\phi(t,z)$ is reduced to a pure optimal stopping time problem,
	\begin{equation}
	\phi(t,z)\!=\!\sup_{\tau\ge t} \mathbb{E}\left[\left.\int_{t}^{\tau}\!\!\! e^{\!-\!\gamma s} \left(\tilde{u}(Z(s))\!-\!(d\!-\!w\bar{L})Z(s)\right)ds\!+\!e^{-\gamma \tau}\tilde{U}(Z(\tau))\right|Z(t)\!=\!z\right],\quad 0<z<\hat{z}.\label{1.12}
	\end{equation} 
\end{itemize}

We first focus on the optimal stopping time problem (\ref{1.12}) and present a lemma to determine its continuous and stopping regions. But before this, one relationship should be noticed. With the optimum $\psi^{*}(t)$, the function $ \phi_{\scriptscriptstyle inf}(t,z)$ can be rewritten as
\begin{equation*}
\phi_{\scriptscriptstyle inf}(t,z)= \mathbb{E}\left[\left.\int_{t}^{\tau} \!e^{-\gamma s}\! \left(\tilde{u}(Z(s))\!-\!(d\!-\!w\bar{L})Z(s)\!-\!(F\!+\!\eta)\psi^{*}(s)Z(s)\right)\!ds\!+\!e^{-\gamma \tau}\!\tilde{U}(Z(\tau))\right|\!Z(t)\!=\!z\right].
\end{equation*}
Fixing the time with $t=\tau$, $\phi_{\scriptscriptstyle inf}(\tau,z)$ can be treated as a single variable function of $z$, that is, $$\phi_{\scriptscriptstyle inf}(\tau,z)\!=\!e^{\!-\!\gamma \tau}\tilde{U}(z).$$ From this equation, and Equation (\ref{pico}) and (\ref{1.11}), we obtain
$\tilde{U}^{\prime}(\hat{z})\!=\!-\!(F\!+\!\eta)\!\leq\!-\!F\!=\!\tilde{U}^{\prime}(\alpha\hat{z}_{\scriptscriptstyle PB})$;
considering the convex property of $\tilde{U}(z)$, we directly obtain the relationship
\begin{equation}
\hat{z}\leq\alpha\hat{z}_{\scriptscriptstyle PB}.\label{1.13}
\end{equation}

\begin{lemma}\label{Lemma 3}
	Assuming
	\begin{equation}
	rF-d+w\bar{L}-\frac{w\bar{L}}{\alpha}<0,\label{1.14}
	\end{equation}
	there exists $\bar{z}$ such that the continuous region of the optimal stopping problem (\ref{1.12}) with the state variable $Z(t)$ is $\Omega_{1}=\{0<z<\bar{z}\}$, and the stopping region is $\Omega_{2}=\{z\ge\bar{z}\}$. $\bar{z}$ is the boundary that separates $\Omega_{1}$ and $\Omega_{2}$.
\end{lemma}
\begin{proof}
See Appendix \ref{C.3}
\end{proof}

\noindent After obtaining the continuous region $\Omega_{1}\!=\!\{0\!<\!z\!<\!\bar{z}\}$, we can treat the stopping time of bankruptcy as the first hitting time of process $Z(t)$ to the boundary $\bar{z}$ from the inner region of $\Omega_{1}$. Figure \ref{Figure 1} describes the relationship between the optimal stopping time and the continuous and stopping regions.  
\begin{figure}[t]
	\centering
	\caption{ Relationship between Optimal Bankruptcy Time and Continuous and Stopping Regions}\label{Figure 1}
	\includegraphics[width=1\linewidth,height=0.3\textheight]{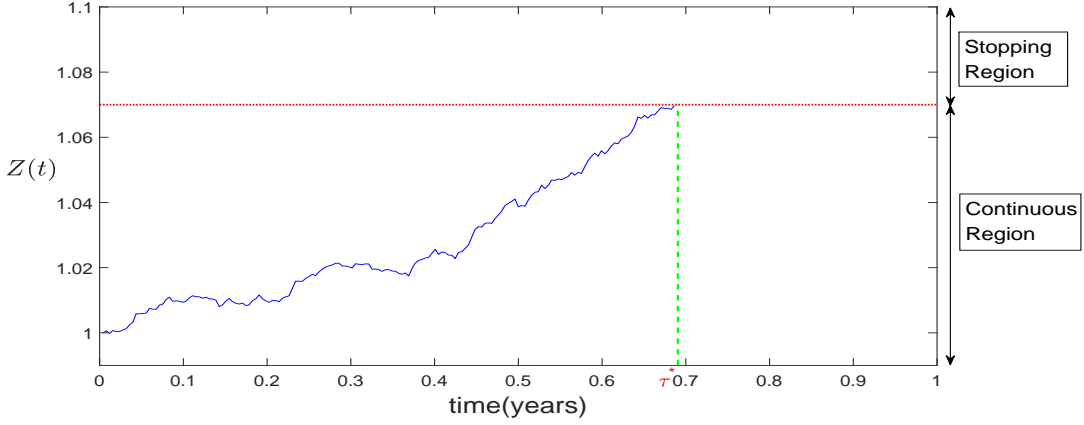}
\end{figure}
\noindent The optimal bankruptcy time is the moment when process $Z(t)$ first touches the boundary, which is represented with the red dotted line, from within the continuous region. Hence, the stopping time satisfies $\tau=\inf\{t\ge 0:Z(t)\ge\bar{z}\}$ and will be proved to be finite with probability one under a sufficient constraint with the following lemma. Besides, it should be clear that $\bar{z}$, corresponding to the bankruptcy threshold, is an upper barrier for the process $Z(t)$, and therefore a lower barrier of the wealth process, as we will show in Remark \ref{Rem1}.
\begin{lemma}\label{Lemma 4}
	Under the assumption
	\begin{equation}
	\gamma>r+\frac{\theta^{2}}{2},\label{1.15}
	\end{equation}
	we have $\mathbb{P}\left(\tau_{\bar{z}}\!<\!\tau_{0}\right)\!=\!1$,
	with two stopping times, $\tau_{\bar{z}}\!=\!\inf\limits_{t\ge 0}\{t:Z(t)\!=\!\bar{z}\}$ and $\tau_{0}\!=\!\inf\limits_{t\ge 0}\{t:Z(t)\!=\!0\}$.
\end{lemma}
\begin{proof}
See Appendix \ref{C.4}.
\end{proof}

Subsequently, combining Condition (\ref{1.11}) and the optimal stopping time problem (\ref{1.12}), we get the free boundary problem which characterizes the function $\phi(\cdot,\cdot)$, considering two different cases:\\
\textbf{(1) Variational Inequalities assuming $\bar{z}<\hat{z}$:} Find the free boundaries $\bar{z}>0$ (Bankruptcy), $\hat{z}>0$ ($(F+\eta)$-wealth level), and a function ${\phi}(\cdot,\cdot)\in C^{1}((0,\infty)\times\mathbb{R}^{+})\cap C^{2}((0,\infty)\times\mathbb{R}^{+}\setminus\{\bar{z}\})$ satisfying
\begin{equation}
\begin{cases}
(V1) \quad\frac{\partial {\phi}}{\partial z}(t,z)+(F+\eta)e^{-\gamma t}=0, & z\ge\hat{z},\\
(V2) \quad \frac{\partial {\phi}}{\partial z}(t,z)+(F+\eta)e^{-\gamma t}\leq 0, & 0<z<\hat{z},\\
(V3)\quad \mathcal{L}{\phi}(t,z)+e^{-\gamma t}\left(\tilde{u}(z)-(d-w\bar{L})z\right)=0, & 0<z<\bar{z},\\
(V4)\quad \mathcal{L}{\phi}(t,z)+e^{-\gamma t}\left(\tilde{u}(z)-(d-w\bar{L})z\right)\leq 0, & \bar{z}\leq z<\hat{z},\\
(V5) \quad {\phi}(t,z)\ge e^{-\gamma t}\tilde{U}(z), & 0<z<\bar{z},\\
(V6) \quad {\phi}(t,z)= e^{-\gamma t}\tilde{U}(z), &\bar{z}\leq z<\hat{z},\label{1.16}
\end{cases}
\end{equation}
for any $t\ge 0$, with the smooth fit conditions
\begin{equation*}
\quad\frac{\partial {\phi}}{\partial z}(t,\hat{z})=-(F+\eta)e^{-\gamma t},\quad \frac{\partial^{2}{\phi}}{\partial z^{2}}(t,\hat{z})=0,
\end{equation*}
\begin{equation*}
\quad\frac{\partial {\phi}}{\partial z}(t,\bar{z})=e^{-\gamma t}\tilde{U}^{\prime}(\bar{z}),\quad {\phi}(t,\bar{z})=e^{-\gamma t}\tilde{U}(\bar{z}).
\end{equation*}

Furthermore, in the period up to the stopping time $\tau$, which corresponds to the interval $0<z<\bar{z}$, we need to consider that whether the constraint $0\leq l(t)\leq L$ is trigged or not, which is related to the boundary $\tilde{y}$ introduced in Lemma \ref{Lemma 1}. Hence, the problem of the pre-bankruptcy part is divided into two cases, $0<\tilde{y}\leq\bar{z}$ and $0<\bar{z}<\tilde{y}$. Then combining with the two cases of the post-bankruptcy part, we have the following framework of partition for the primal optimization problem $(P)$. 

\begin{itemize}
	\item $0<\bar{z}<\tilde{y}$:
	\begin{center}
		\begin{tikzpicture}[
		grow=right,
		level 1/.style={sibling distance=1.7cm,level distance=3cm},
		level 2/.style={sibling distance=2cm, level distance=6cm},
		edge from parent/.style={very thick,draw=blue!40!black!60,
			shorten >=5pt, shorten <=5pt},
		edge from parent path={(\tikzparentnode.east) -- (\tikzchildnode.west)},
		kant/.style={text width=2cm, text centered, sloped},
		every node/.style={text ragged, inner sep=2mm},
		punkt/.style={rectangle, rounded corners, shade, top color=white,
			bottom color=blue!50!black!20, draw=blue!40!black!60, very
			thick }
		]
		
		\node[punkt, text width=4.5em] {$0<\bar{z}<\tilde{y}$}
		child{ node[punkt, text width=7em] { $0<\hat{z}_{\scriptscriptstyle PB}<\tilde{y}$ }
			child {
				node[punkt,text width=19em] { Case 4. $0\!<\!\bar{z}\!<\!\hat{z}\!\leq\!\alpha\hat{z}_{\scriptscriptstyle PB}\!<\!\alpha\tilde{y}\!<\!\tilde{y}$}
				edge from parent
				node[kant, below, pos=.6] {}
			}
			edge from parent{
				node[kant, above] {}}
		}
		child{ node[punkt, text width=7em] {$0<\tilde{y}\leq\hat{z}_{\scriptscriptstyle PB}$}
			child {
				node [punkt,text width=19em] [rectangle split, rectangle split,rectangle split parts=3,
				text ragged] {
					\text{ Case 1. $0\!<\!\bar{z}\!<\!\hat{z}\!<\!\alpha\tilde{y}\!\leq\!\alpha\hat{z}_{\scriptscriptstyle PB}$ }
					\nodepart{second}
					\text{ Case 2. $0\!<\!\bar{z}\!<\!\alpha\tilde{y}\!\leq\!\hat{z}\!\leq\!\alpha\hat{z}_{\scriptscriptstyle PB}$ }
					\nodepart{third}
					\text{ Case 3. $0\!<\!\alpha\tilde{y}\!<\!\bar{z}\!<\!\hat{z}\!\leq\!\alpha\hat{z}_{\scriptscriptstyle PB}$ $\&$ $\bar{z}\!<\!\tilde{y}$ }
			}}
			edge from parent{
				node[kant, above] {}}
		};
		\end{tikzpicture}
	\end{center}
	
	\item $0<\tilde{y}\leq\bar{z}$:
	\begin{center}
		\begin{tikzpicture}[
		grow=right,
		level 1/.style={sibling distance=1cm,level distance=2.7cm},
		level 2/.style={sibling distance=4.5cm, level distance=6cm},
		edge from parent/.style={very thick,draw=blue!40!black!60,
			shorten >=5pt, shorten <=5pt},
		edge from parent path={(\tikzparentnode.east) -- (\tikzchildnode.west)},
		kant/.style={text width=2cm, text centered, sloped},
		every node/.style={text ragged, inner sep=2mm},
		punkt/.style={rectangle, rounded corners, shade, top color=white,
			bottom color=blue!50!black!20, draw=blue!40!black!60, very
			thick }
		]
		
		\node[punkt, text width=4.3em] {$0<\tilde{y}\leq\bar{z}$}
		child{ node[punkt, text width=5.3em] { $0<\hat{z}_{\scriptscriptstyle PB}<\tilde{y}$ }
			child {
				node[punkt,text width=22.3em] {No Case (Since $\hat{z}_{\scriptscriptstyle PB}\!<\!\tilde{y}\!\leq\!\bar{z}\!<\!\hat{z}\!\leq\!\alpha\hat{z}_{\scriptscriptstyle PB}$ $\&$ $\alpha\in(0,1)$)}
				edge from parent
				node[kant, below, pos=.6] {}
			}
			edge from parent{
				node[kant, above] {}}
		}
		child{ node[punkt, text width=5.3em] { $0<\tilde{y}\leq\hat{z}_{\scriptscriptstyle PB}$ }
			child {
				node[punkt,text width=22.3em]{ Case 5. $0\!<\!\alpha\tilde{y}\!<\!\tilde{y}\!\leq\!\bar{z}\!<\!\hat{z}\!\leq\!\alpha\hat{z}_{\scriptscriptstyle PB}$}
				edge from parent
				node[kant, below, pos=.6] {}
			}
			edge from parent{
				node[kant, above] {}}
		};
		\end{tikzpicture}	
	\end{center}
	
\end{itemize}
\textbf{(2) Variational Inequalities assuming $\bar{z}\ge\hat{z}$:} Find the free boundaries $\bar{z}>0$ (Bankruptcy), $\hat{z}>0$ ($(F+\eta)$-wealth level), and a function $\phi(\cdot,\cdot)\in C^{1}((0,\infty)\times\mathbb{R}^{+})\cap C^{2}((0,\infty)\times\mathbb{R}^{+}\setminus\{\bar{z}\})$ satisfying
\begin{equation}
\begin{cases}
(V1) \quad\frac{\partial \phi}{\partial z}(t,z)+(F+\eta)e^{-\gamma t}=0, & \hat{z}\leq z<\bar{z},\\
(V2) \quad \frac{\partial \phi}{\partial z}(t,z)+(F+\eta)e^{-\gamma t}\leq 0, & 0<z<\hat{z},\\
(V3)\quad \mathcal{L}\phi(t,z)+e^{-\gamma t}\left(\tilde{u}(z)-(d-w\bar{L})z\right)=0, & 0<z<\hat{z},\\
(V4)\quad \mathcal{L}\phi(t,z)+e^{-\gamma t}\left(\tilde{u}(z)-(d-w\bar{L})z\right)\leq 0, &\hat{z}\leq z<\bar{z},\\
(V5) \quad \phi(t,z)\ge e^{-\gamma t}\tilde{U}(z), & 0<z<\bar{z},\\
(V6) \quad \phi(t,z)= e^{-\gamma t}\tilde{U}(z), &z\ge\bar{z},\label{1.17}
\end{cases}
\end{equation}
for any $t\ge 0$, with the smooth fit conditions $\frac{\partial \phi}{\partial z}(t,\hat{z})=-(F+\eta)e^{-\gamma t}$ and $\frac{\partial^{2}\phi}{\partial z^{2}}(t,\hat{z})=0$.
Same as before, we provide a diagram to show the possible situations under this prerequisite.\\
\begin{center}
	\begin{tikzpicture}[
	grow=right,
	level 1/.style={sibling distance=1cm,level distance=3.7cm},
	level 2/.style={sibling distance=4.5cm, level distance=4.5cm},
	edge from parent/.style={very thick,draw=blue!40!black!60,
		shorten >=5pt, shorten <=5pt},
	edge from parent path={(\tikzparentnode.east) -- (\tikzchildnode.west)},
	kant/.style={text width=2cm, text centered, sloped},
	every node/.style={text ragged, inner sep=2mm},
	punkt/.style={rectangle, rounded corners, shade, top color=white,
		bottom color=blue!50!black!20, draw=blue!40!black!60, very
		thick }
	]
	
	\node[punkt, text width=4.5em] {$\bar{z}\ge\hat{z}>0$}
	child{ node[punkt, text width=5.5em] { $0<\hat{z}<\tilde{y}$ }
		child {
			node[punkt,text width=12em] {Case 7. $0<\hat{z}<\tilde{y}$ $\&$ $\bar{z}\ge\hat{z}$}
			edge from parent
			node[kant, below, pos=.6] {}
		}
		edge from parent{
			node[kant, above] {}}
	}
	child{ node[punkt, text width=5.5em] { $0<\tilde{y}\leq\hat{z}$ }
		child {
			node[punkt,text width=12em]{Case 6. $0<\tilde{y}\leq\hat{z}$ $\&$ $\bar{z}\ge\hat{z}$}
			edge from parent
			node[kant, below, pos=.6] {}
		}
		edge from parent{
			node[kant, above] {}}
	};
	\end{tikzpicture}	
	
\end{center}

\noindent As in the post-bankruptcy problem, we assume that ${\phi}(t,z)$ adopts a time-independent form, ${\phi}(t,z)= e^{-\gamma t} v(z)$, then the solution is obtained, see Appendix \ref{C.5}. We would like to stress that only one of the seven cases admits a solution.

After acquiring the closed form of $v(z)$,  \cite[Section 8, Theorem 8.5]{karatzas2000utility} indicates that $$V(x)=\inf\limits_{\lambda>0}[\tilde{V}(\lambda)+\lambda x]=\inf\limits_{\lambda>0}[v(\lambda)+\lambda x]$$ keeps true for a unique $\lambda^*>0$ under the differentiable property of $v(\cdot)$. Then the wealth threshold of bankruptcy, namely $\bar{x}$, can be calculated from the relationship $\bar{x}=-v^{\prime}(\bar{z})$. Therefore, given any initial wealth $x\ge F+\eta$, we get the optimal Lagrange multiplier $\lambda^{*}$ through
solving the equation, $x=-v^{\prime} (\lambda^{*})$. Furthermore, since the optimum $\lambda^{*}$ is the initial value of process (\ref{Zeta}), $Z^{*}(t)$, the optimal wealth process follows $X^{*}(t)=-v^{\prime}(Z^{*}(t))$. The optimal bankruptcy time satisfies $\tau^{*}=\inf\limits_{t\ge 0}\{X^{*}(t)\leq \bar{x}\}$. Moreover, recalling Lemma \ref{Lemma 1}, the optimal consumption and leisure strategies are
\begin{equation*}
	\begin{cases}
		c^{*}(t)=(Z^{*}(t))^{-\frac{1}{k}}\left(\frac{1-\delta}{\delta w}\right)^{\frac{(1-k)(1-\delta)}{k}}\mathbb{I}_{\{Z^{*}(t)\ge\tilde{y}\}}+L^{-\frac{(1-k)(1-\delta)}{\delta(1-k)-1}}(Z^{*}(t))^{\frac{1}{\delta(1-k)-1}}\mathbb{I}_{\{0<Z^{*}(t)<\tilde{y}\}},\\
		l^{*}(t)=(Z^{*}(t))^{-\frac{1}{k}}\left(\frac{1-\delta}{\delta w}\right)^{-\frac{\delta(1-k)-1}{k}}\mathbb{I}_{\{Z^{*}(t)\ge\tilde{y}\}}+L\mathbb{I}_{\{0<Z^{*}(t)<\tilde{y}\}},
	\end{cases}
\end{equation*}
and the optimal portfolio strategy is $\pi^{*}(t)=\frac{\theta}{\sigma}Z^{*}(t)v^{\prime\prime}(Z^{*}(t))$, which can be obtained from \cite[Section 5, Theorem 3]{he1993labor}.

\begin{remark}\label{Rem1}
In Figure \ref{Figure 1}, we plot the relationship between the optimal bankruptcy time and the continuous and stopping regions with respect to $Z^*(t)$, showing that the optimal bankruptcy time is the first time the process $Z^*(t)$ touch the upper barrier $\bar{z}$. The same plot can be done with respect to $X^{*}(t)=-v^{\prime}(Z^{*}(t))$: the convex property of $v(\cdot)$, see \cite[Section 3.4, Lemma 4.3]{karatzas1998methods}, indicates that $X^{*}(t)$ is a decreasing function of $Z^{*}(t)$, therefore, in this case the optimal bankruptcy time is the first time the process $X^*(t)$ touch a lower barrier $\bar{x}=-v^{\prime}(\bar{z})$.
\end{remark}

\section{Numerical Analysis}\label{Section 5}
We now implement the sensitivity analysis with respect to the input parameters. As baseline parameters, we consider the ones listed in Table \ref{Table 1}. These inputs satisfy conditions (\ref{1.3}), (\ref{1.14}) and (\ref{1.15}).
\begin{table}[t]
	\renewcommand\arraystretch{1.3}
	\centering
	\caption{Baseline Input Parameters}\label{Table 1}
	\begin{tabular}{cccccccccccccc} \hline
		\textbf{ $\delta$} & \textbf{ $k$ } & \textbf{$r$} &\textbf{$\mu$} & \textbf{ $\sigma$ } & \textbf{ $\gamma$ } & \textbf{ $d$ } & \textbf{ $w$ } & \textbf{ $F$ }  & \textbf{ $\eta$ }  & \textbf{ $\alpha$ }& \textbf{ $\bar{L}$ } & \textbf{ $L$ } & \textbf{ $x$ }  \\ \hline
		0.6 & 3 & 0.05 & 0.1 & 0.2 & 0.3 & 0.3 & 1.5 & 0.96 & 0.0001 & 0.9 & 1 & 0.8 & 6.6  \\ \hline
	\end{tabular}
\end{table} 
\noindent The parameters $r$, $\mu$, $\sigma$, $\gamma$ and $\alpha$ are directly taken from \cite{jeanblanc2004optimal}. Whereas the fixed bankruptcy toll and the debt repayment amount in their study are $F=400\$$ and $d=125\$$, we set $F=0.96$ and $d=0.3$ such that the ratios of $F$ and $d$ in our and their research keep consistent. The same consideration is also applied for the setting of the initial wealth $x$.
As discussed in Section \ref{Section 4}, seven cases should be considered simultaneously and only one must be verified: in fact with these parameters, only Case 2, ``$0<\bar{z}<\alpha \tilde{y}\leq\hat{z}<\alpha\hat{z}_{PB}$'', admits a solution. With the above input parameters, we derive the output parameters: $B_{2}$, $\bar{z}$, $\hat{z}$, $\tilde{y}$, the set of wealth thresholds $\{\hat{x}, \bar{x}, \tilde{x}\}$, the optimal Lagrange multiplier $\lambda^{*}$ and the value function $V(x)$. The results are listed in Table \ref{Table 2} and \ref{Table 3}. The $fval$, i.e., the value of the function at its zero, in Tables \ref{Table 2} represents the maximum error generating from using the $fsolve$ function of MATLAB to solve the non-linear equations (\ref{Ap.19}), (\ref{Ap.20}) and (\ref{Ap.21}): as expected, the $fval$ value is close to zero, i.e., the algorithm correctly solve the system of equations. 
\begin{table}[t]
	\renewcommand\arraystretch{1.3}
	\centering
	\caption{Output Parameters, Optimal Lagrange Multiplier and Value Function}\label{Table 2}
	\begin{tabular}{cccccc|c} \hline
		\textbf{ $B_{2}$} & \textbf{ $\hat{z}$} & \textbf{ $\tilde{y}$} & \textbf{ $\bar{z}$} & \textbf{ $V(x)$ } &\textbf{$\lambda^{*}$} & $fval$ \\ \hline
		5.3104 & 2.6126 & 0.3280  & 0.1591 & -2.7741 & 0.3110 & 3.6973e-11 \\ \hline
	\end{tabular}
\end{table}

\begin{table}[t]
	\renewcommand\arraystretch{1.3}
	\centering
	\caption{Wealth Thresholds}\label{Table 3}
	\begin{tabular}{cccc} \hline
		\textbf{ $\hat{x}$} & \multicolumn{2}{c}{\textbf{ $\tilde{x}$ }} & \textbf{$\bar{x}$} \\ 
		\textbf{ $F\!+\!\eta$-wealth level} & \multicolumn{2}{c}{\textbf{ minimum wealth level for $l^{*}(t)\!=\!L$ }} & \textbf{bankruptcy wealth level}\\ \hline
		0.9601  & \multicolumn{2}{c}{6.3164} & 10.0651 \\ \hline
	\end{tabular}
\end{table}

We first of all want to stress that in this case the optimal bankruptcy time is 0, since the initial wealth, 6.6, is lower than the optimal bankruptcy barrier $\bar{x}= 10.0651$, i.e., the starting wealth is inside the stopping region.

\paragraph{Sensitivity of optimal solutions with respect to the risk aversion coefficient $k$:}~{}

In this part, we use the Monte Carlo method to simulate the single path of optimal wealth process and consumption-portfolio-leisure strategy by taking different values of $k$ for discovering the sensitivity of optimal solutions to the risk aversion coefficient. Parameters are the ones in Table \ref{Table 1}, with the exception of the initial wealth which is set to 25 instead of 6.6 for observing the bankruptcy mechanism. The agent with a higher value of $k$ prefers to have a higher wealth threshold for declaring bankruptcy ($\bar{x}=7.4777$ for $k=2$, $\bar{x}=10.0651$ for $k=3$ and $\bar{x}=12.4239$ for $k=4$). This is reasonable, as shown in Figure \ref{Figure 2}: the more risk-averse agent tends to smooth the consumption and leisure, and invest less in the risky asset such that maintaining a relatively higher wealth level. Moreover, from the optimal trajectories of leisure, it can be observed that the relatively low fixed and flexible bankruptcy costs, $F = 0.96$ and $1-\alpha= 0.1$, and the high bankruptcy
wealth thresholds enable the agent to enjoy the maximum leisure rate $L=0.8$ even after suffering the wealth shrinkage caused by declaring bankruptcy. Therefore, the leisure processes corresponding to different $k$ values are fully identical. Finally, in Figure \ref{Figure 2} we zoom close to the bankruptcy time, to show the discontinuity of the optimal strategies.

\begin{figure}[t]
	\centering
	\caption{Trajectories of Optimal Wealth and Control Strategies w.r.t. Risk Aversion Coefficient}\label{Figure 2}
	\includegraphics[width=1\linewidth,height=0.45\textheight]{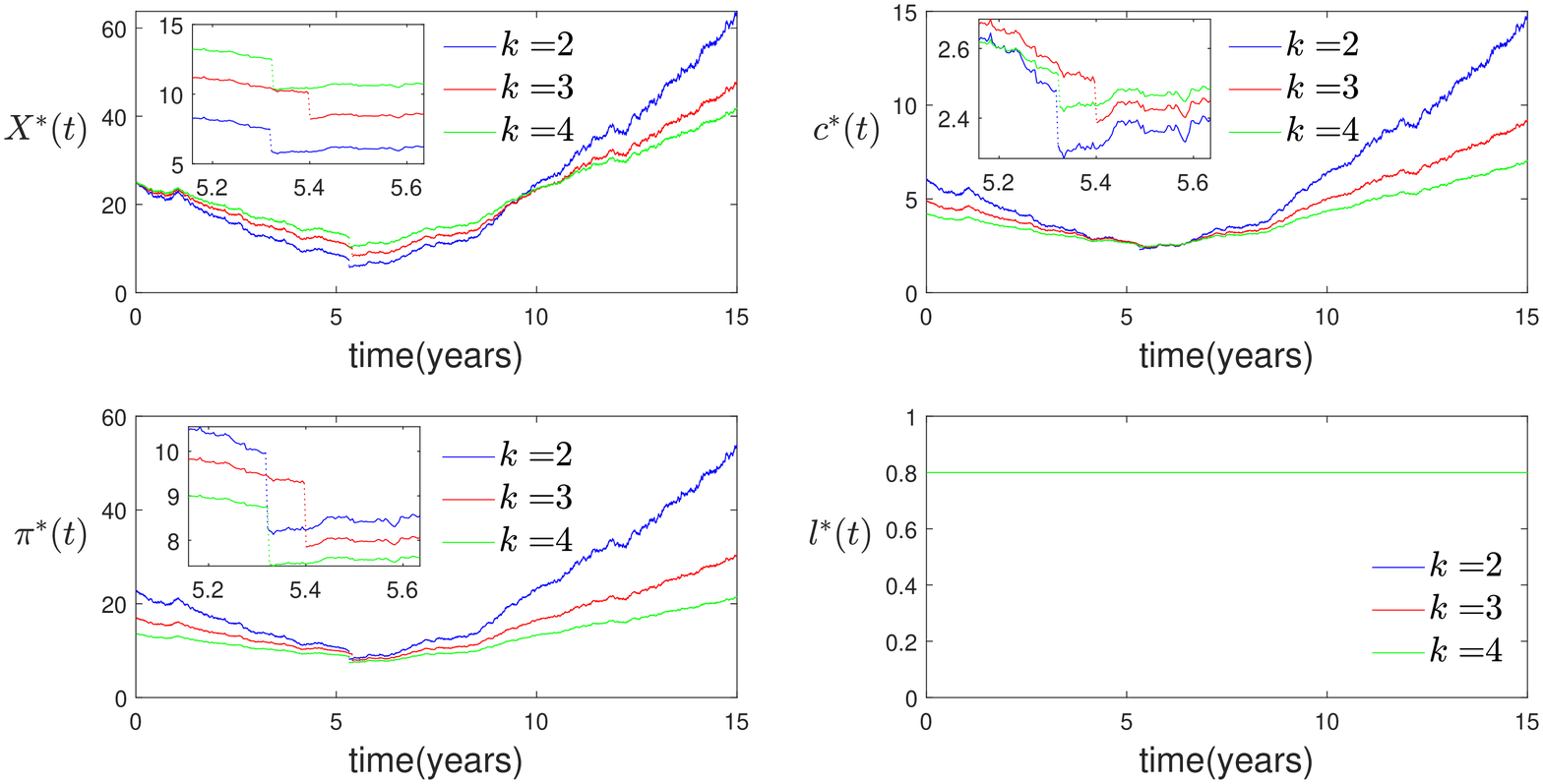}
\end{figure} 

\paragraph{Sensitivity of the optimal bankruptcy threshold with respect to the market risk premium jointly with the risk aversion coefficient:}~{}

The parameter $\theta=\frac{\mu-r}{\sigma}$, which is the Sharpe-Ratio, measures the market risk premium. For the purpose of discovering the relationship between $\bar{x}$ and $\theta$, we keep $r$, $\sigma$ constant and change the value of $\mu$ from 0.035 to 0.2, with an interval of 0.0075, which leads the value of $\theta$ to change from 0.1 to 1.2. Figure \ref{Figure 3} shows that, with a lower risk premium, that is, a smaller value of $\theta$, the agent prefers to set a higher wealth threshold to more easily get rid of the debt thanks to the bankruptcy. Contrarily, a better market performance entails the agent a stronger ability to bear the debt repayment; hence, she sets a lower wealth threshold to avoid suffering the bankruptcy costs. Furthermore, there is a positive relationship between the optimal wealth threshold of bankruptcy and the risk aversion level $k$, which is already clarified through Figure \ref{Figure 2}.

\begin{figure}[t]
	\centering
	\caption{Market Risk Premium, Risk Aversion Coefficient and Bankruptcy Threshold}\label{Figure 3}
	\includegraphics[width=1\linewidth,height=0.29\textheight]{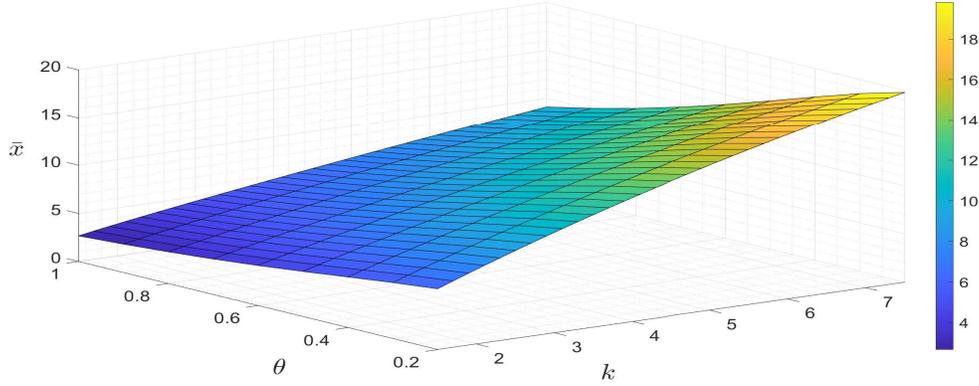}
\end{figure}
\noindent 

\paragraph{Sensitivity of the optimal bankruptcy threshold with respect to the debt repayment, the fixed and flexible bankruptcy cost:}~{}

The wealth process suffers a shrinkage through an affine function $\alpha(X(\tau)-F)$ declaring bankruptcy, and the debt repayment is exempted. Hence, $F$ and $(1-\alpha)$ can be treated as the fixed and flexible cost of bankruptcy, and $d$ is the benefit of bankruptcy. In order to discover the influence of the bankruptcy option, we provide Figures \ref{Figure 4}-\ref{Figure 5} to illustrate the sensitivity of optimal wealth threshold of bankruptcy with respect to the coefficients $\alpha$, $d$ and $F$. 
\begin{figure}[t]
	\centering
	\caption{Flexible Bankruptcy Cost, Debt Repayment and Bankruptcy Wealth Threshold}\label{Figure 4}
	\includegraphics[width=1\linewidth,height=0.29\textheight]{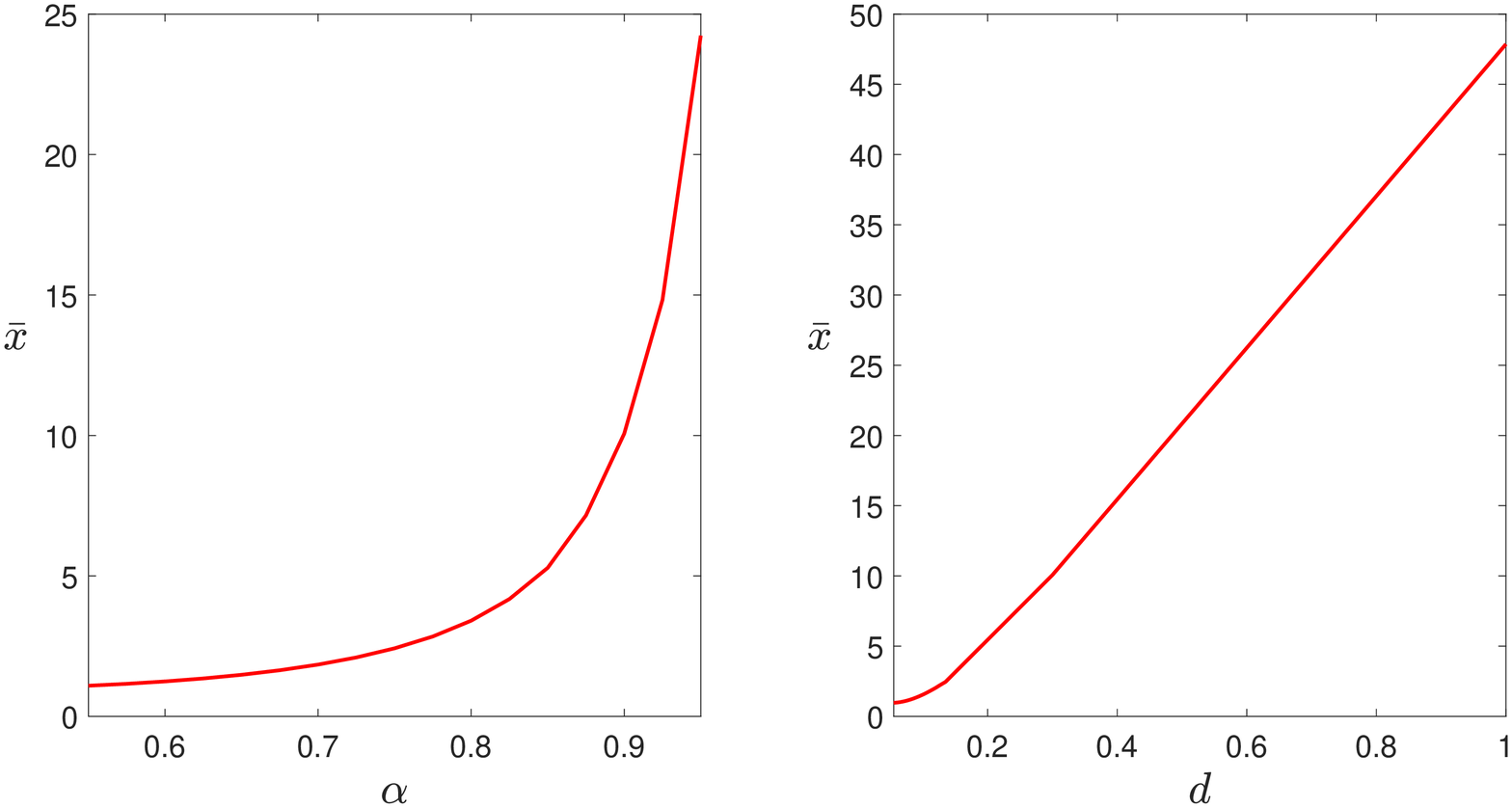}
\end{figure}
\noindent First considering the sensitivity of bankruptcy wealth threshold to the flexible cost coefficient, we can observe that $\bar{x}$ is an increasing function of $\alpha$. The rationale is the following: since $\alpha$ represents the proportion of wealth held after bankruptcy, a lower value of $\alpha$ indicates a higher cost such that the agent prefers to set a lower wealth threshold to avoid suffering the wealth shrinkage from bankruptcy. As for the relationship between $\bar{x}$ and $d$, it can be observed the same increasing and convex curve. When the debt repayment is higher, which implies that the benefit of bankruptcy is more attractive, the agent tends to take a higher threshold such that the wealth process satisfies the bankruptcy requirement $\bar{x}$ more easily to enjoy the debt exemption. However, this incentive becomes weaker as the debt repayment decreases, which leads to the convexity of the considering mapping. Finally, in Figure \ref{Figure 5} we provide a three-dimensional image to explain the sensitivity of bankruptcy wealth threshold to the parameter $F$ jointly with $k$. 
\begin{figure}[t]
	\centering
	\caption{Risk Aversion Coefficient, Fixed Bankruptcy Cost and Bankruptcy Wealth Threshold}\label{Figure 5}
	\includegraphics[width=1\linewidth,height=0.3\textheight]{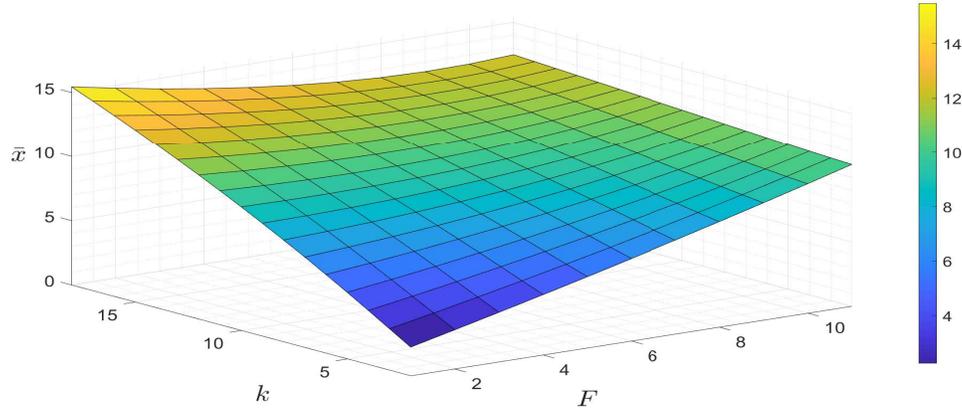}
\end{figure}
\noindent Since the value of $F$ adjusted according to \cite{jeanblanc2004optimal} is relatively low, the liquidity constraint triggered by $F+\eta$ is easy to be covered by the labour income. Thus, the role of $F$ is more related to the fixed cost of bankruptcy such that a higher value of $F$ will make the agent set a lower wealth threshold to avoid suffering the bankruptcy, which results in a monotonic decreasing relationship between $\bar{x}$ and $F$. However, $F$ is not only the cost of bankruptcy, but also can be regarded as the liquidity constraint to limit the agent's investment behaviour. In order to reflect the phenomenon that the role of $F$ is the trade-off between the liquidity constraint boundary of the pre-bankruptcy period and the fixed cost of bankruptcy, we conduct the same three-dimensional image with different input parameters, particularly, we follow the baseline of inputs in Table \ref{Table 1} and only change the values of $r=0.02$, $\mu=0.07$, $\sigma=0.15$, $\gamma=0.1$ and $\alpha=0.7$ to satisfy Condition (\ref{1.14}) also for larger values of $F$. From Figure \ref{Figure 6}, we find that the relationship between $\bar{x}$ and $F$ is not monotonous. When the risk aversion level is low, a positive relationship between $F$ and the bankruptcy threshold of wealth is observed. This is because a larger $F$ value, which is treated as the collateral recalling that $X(t)\ge F+\eta$ before bankruptcy, will reduce the agent's available capital and limit her investment behaviour. Therefore, the agent will set a higher bankruptcy wealth threshold to get rid of the limitation of the liquidity constraint. Whereas for the agent with deep risk aversion, liquidity constraints are less restrictive, and the parameter $F$ plays more as the role of the bankruptcy cost.


\begin{figure}[t]
	\centering
	\caption{Risk Aversion Coefficient, Fixed Bankruptcy Cost and Bankruptcy Wealth Threshold}\label{Figure 6}
	\includegraphics[width=1\linewidth,height=0.3\textheight]{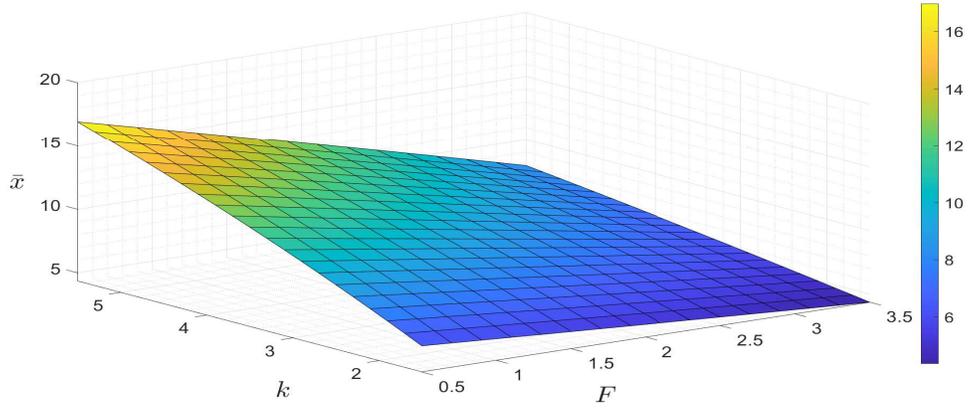}
\end{figure}

\paragraph{Influence of the bankruptcy option:}~{}

To study the influence of introducing the bankruptcy option, we also solve a pure optimal control problem without optimal stopping, see Appendix \ref{nob}. In Figure \ref{Figure 7}, we take the inputs baseline in Table \ref{Table 1} and show the optimal controls (as a function of the initial wealth) for the case with and without bankruptcy option models to reveal its influence. In the second case, we consider both the case with a liquidity constraint $X(t)>F+\eta$ and without liquidity constraint.
\begin{figure}[t]
	\centering
	\caption{Influence of Bankruptcy Option}\label{Figure 7}
	\hspace*{-1cm}\includegraphics[width=1.1\linewidth,height=0.3\textheight]{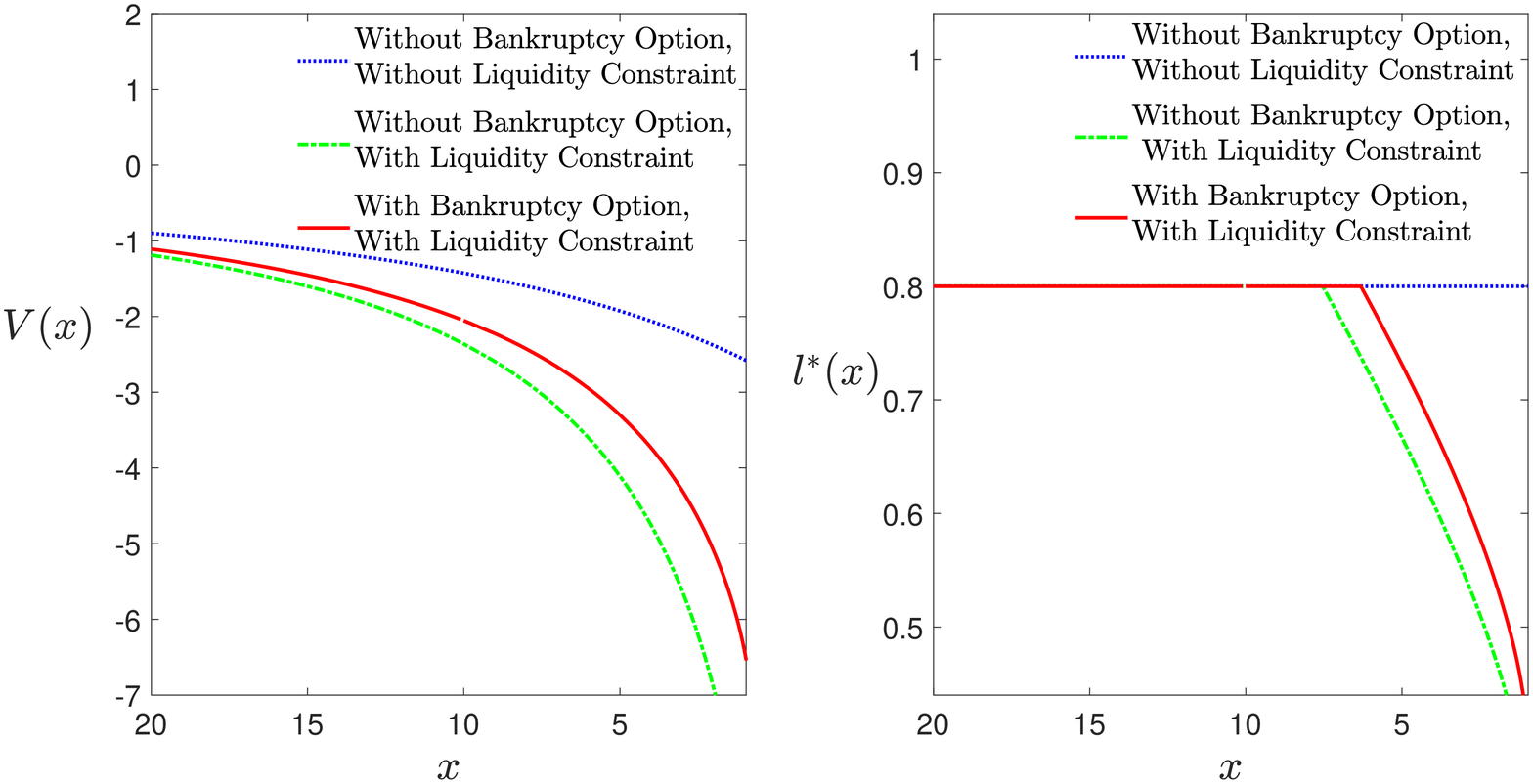}
	\hspace*{-1cm}\includegraphics[width=1.1\linewidth,height=0.3\textheight]{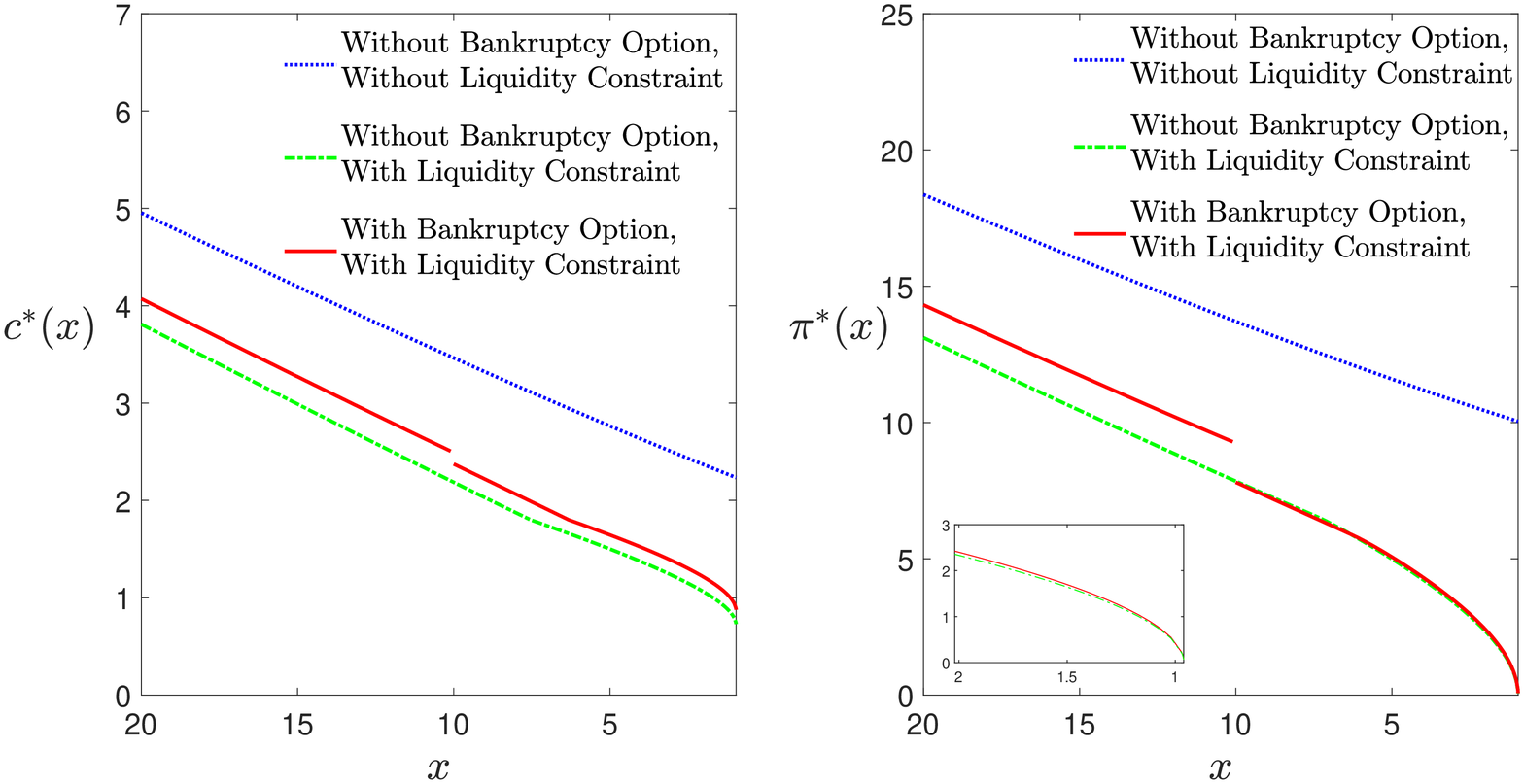}
\end{figure}
\noindent From Figure \ref{Figure 7}, five phenomena should be noticed by comparing the two curves with liquidity constraint. First, due to the additional bankruptcy option, the value function is always greater than the value function without such an option at any given initial wealth level. Second, before the occurring of optimal stopping time, the additional option offers the agent a better circumstance, and the optimal consumption-portfolio-leisure policies always dominates the corresponding policies without bankruptcy option model. Third, in order to meet the needs of obtaining utility from consumption and leisure, the agent with the bankruptcy option and a low initial wealth immediately file bankruptcy, facing a shrinkage of wealth, which causes a downward jump in consumption and allocation in the risky asset. Fourth, when a liquidity constraint is considered, the optimal leisure rate decreases for low values of the initial wealth $x$, since the agent needs to spend more time working to get a larger wage, therefore not exploiting the full leisure, to face debt and liquidity constraint. Finally, the amount of money allocated in the risky asset is 0 as the wealth level drops to the liquidity constraint boundary ($F+\eta=0.9601$). This is to avoid that the wealth process violates the liquidity constraint due to the risky asset's fluctuation. However, the optimal consumption always keeps positive even as the wealth approaches the liquidity boundary since the agent continues to obtain the labour income. Additionally, we can observe that the optimal solutions without the liquidity constraint always dominates the solutions with this extra constraint.

\paragraph{Influence of the labour income:}~{}

In order to investigate the influence of introducing the leisure rate as an additional control variable and thus the labour income, we first of all compare the results with the optimal bankruptcy model in \cite{jeanblanc2004optimal}, therefore with full leisure and no labour income, to conduct the numerical analysis to discover the sensitivity with respect to the presence of leisure rate and labour income. 

\begin{figure}[t]
	\centering
	\caption{Influence of Labour Income}\label{Figure 8}
	\includegraphics[width=1\linewidth,height=0.25\textheight]{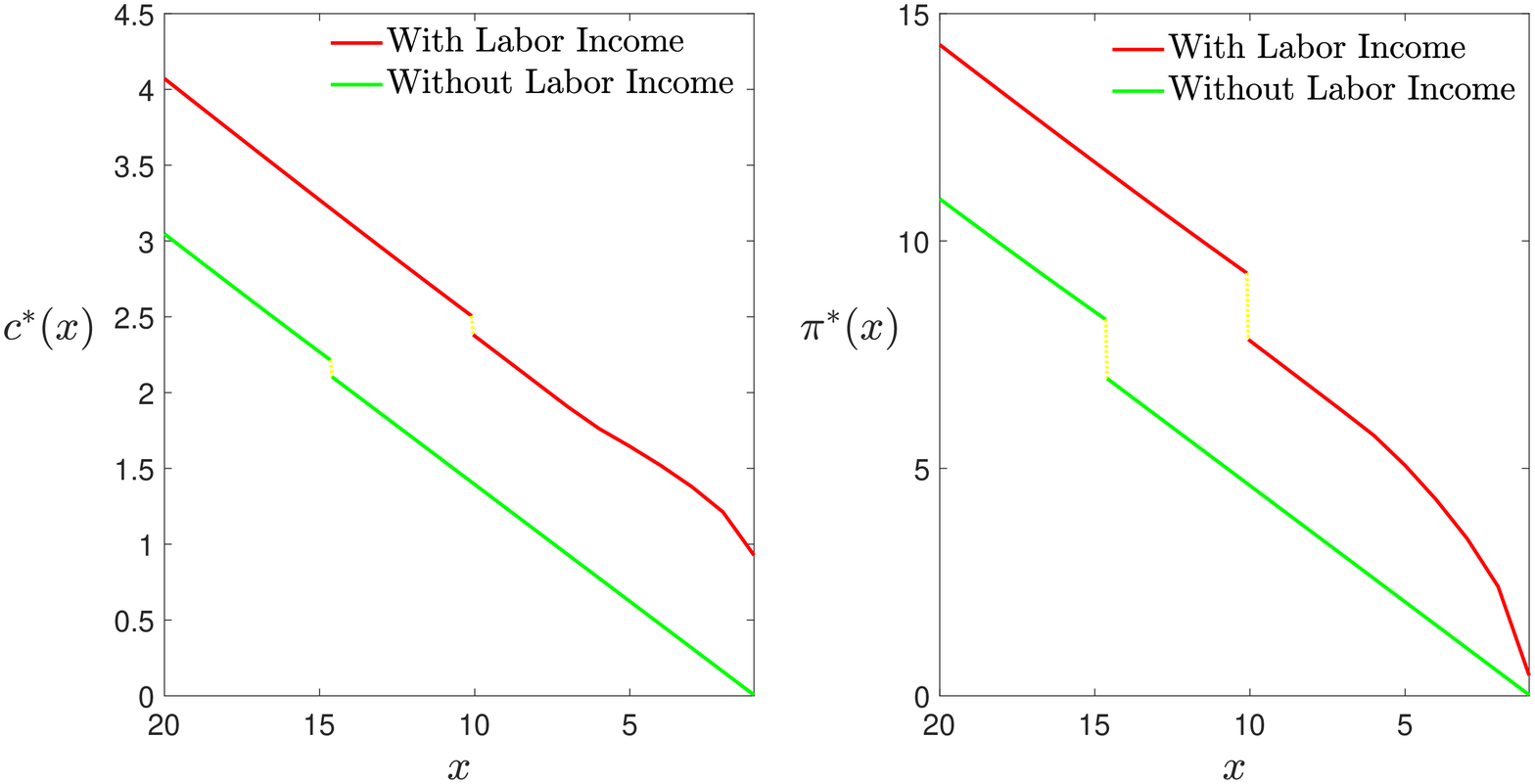}
\end{figure}
\noindent Figure \ref{Figure 8} shows that there exists a downward jumps of optimal control strategies for both full and selectable leisure models due to the shrinkage of wealth at the moment of declaring bankruptcy. Moreover, it can be observed that the optimal consumption and portfolio policies of the optimization problem introducing the leisure as a control variable always keep dominating the corresponding policies of the model with full leisure rate since the agent can earn the additional income from labour. Meanwhile, comparing the wealth levels corresponding to the downward jumps, we can observe that the agent with the full leisure rate tends to have a higher bankruptcy wealth threshold such that more wealth can be taken into the post-bankruptcy period to support the further consumption ($\bar{x}=10.0651$ for the ``with labour income case'', $\bar{x}=14.6091$ for the ``without labour income -full leisure- case''). Secondly, in Figure \ref{Figure 9} we discover the sensitivity of bankruptcy wealth threshold with respect to different values of $L$ within the optimization model considering leisure selection.
\begin{figure}[t]
	\centering
	\caption{Optimal Bankruptcy Threshold w.r.t. Maximum Leisure Rate}\label{Figure 9}
	\includegraphics[width=1\linewidth,height=0.25\textheight]{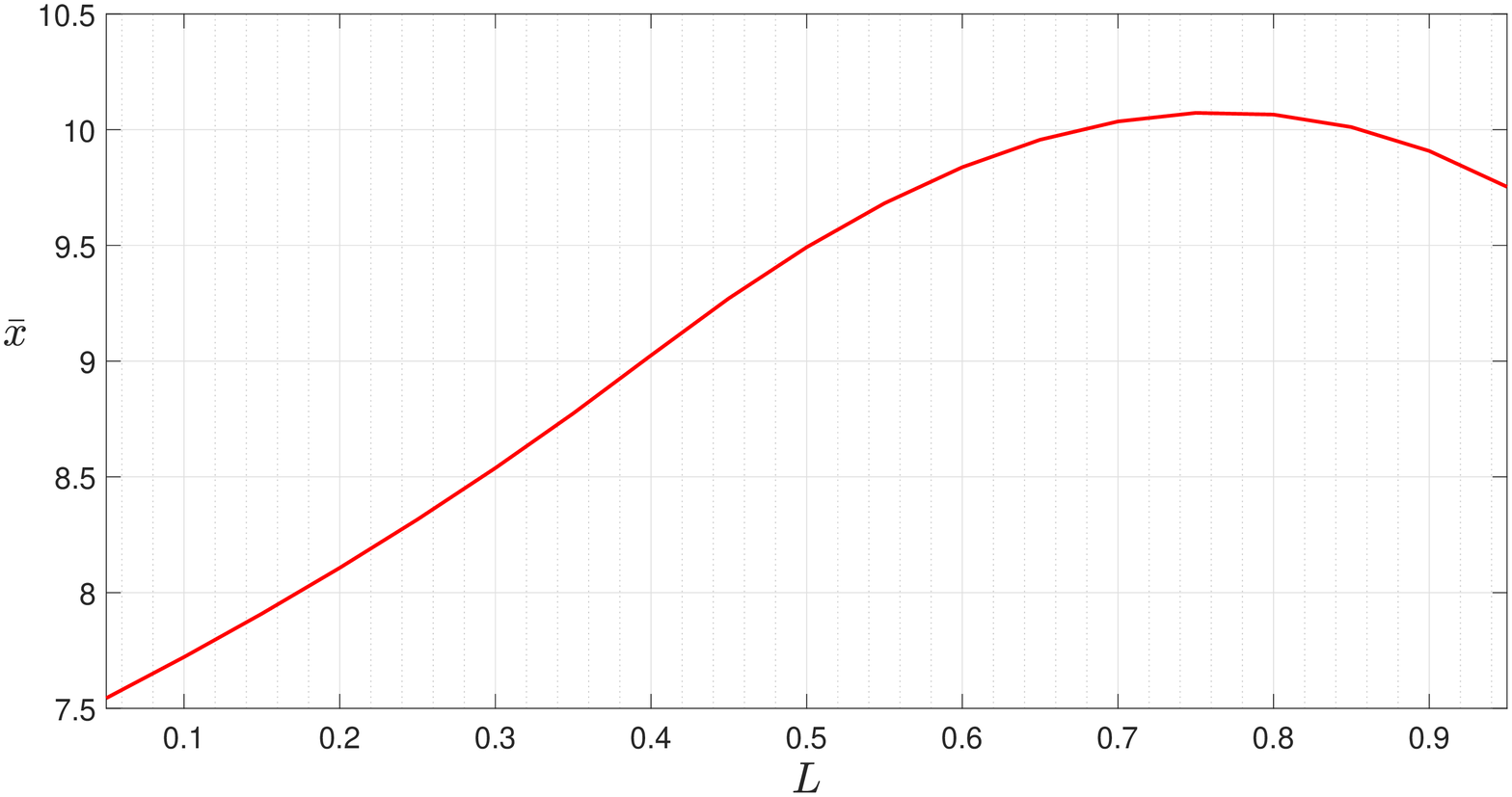}
\end{figure}
\noindent Figure \ref{Figure 9} shows that the bankruptcy wealth threshold is not a monotonic function of $L$; the relationship between these two variables works in a complex way and can be analysed briefly into two separated pieces: one piece is with a relatively low value of $L$, and another is with higher $L$ value. Since $L$ represents the upper boundary of the leisure variable $l(t)$, $(\bar{L}-L)$ represents the minimum mandatory working rate. In the first part, the low enough value of $L$ obliges the agent to allocate most of the time on work, thereby gaining enough labour income to afford the debt repayment. Hence, she prefers to set a smaller wealth threshold for bankruptcy to avoid encountering the wealth shrinkage. However, for the second part, the restriction on the optimal leisure choice caused by $L$ becomes weaker as its value increases, the agent gains the utility from taking more leisure, which leads to a decreasing labour income, a decreasing consumption since leisure and consumption are substitute, and a corresponding lower bankruptcy wealth threshold. Finally, we conduct the sensitivity analysis to the wage rate and present the result in Figure \ref{Figure 10}, in which an inverse relationship occurs between $\bar{x}$ and $w$. Because the economic situation of the agent becomes worse with a lower wage rate $w$, she tends to set a higher critical wealth level such that more wealth is retained for supporting the post-bankruptcy life.

\begin{figure}[t]
	\centering
	\caption{Optimal Bankruptcy Threshold w.r.t. Wage Rate}\label{Figure 10}
	\includegraphics[width=1\linewidth,height=0.25\textheight]{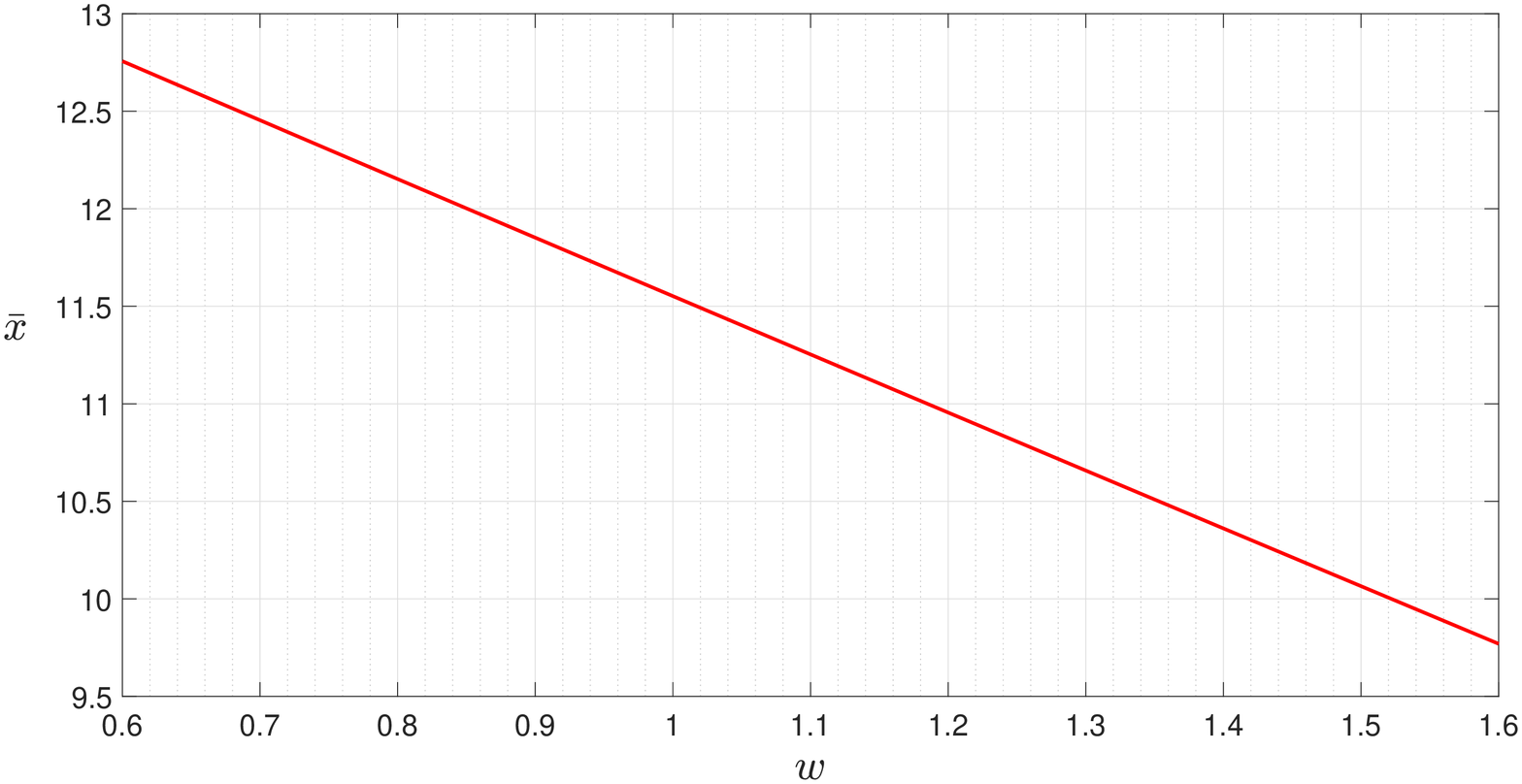}
\end{figure}

\section{Conclusion}\label{Section 6}
In this work the optimal consumption-portfolio-leisure and bankruptcy problem concerning a power utility function has been solved semi-analytically. By the Legendre-Fenchel transform, we have established the duality between the optimization problem with the individual's shadow price problem, which results in a system of variational inequalities then enables us to obtain the closed-form solutions. The optimal wealth and control strategies are represented as functions of wealth's dual process, $Z(t)$. Then we have proved that the optimal policy for the agent is to file bankruptcy at the first hitting time of the optimal wealth process to a critical wealth level, which is the boundary separating the continuous and stopping regions of the corresponding stopping time model. We have also conducted the sensitivity analysis of this wealth threshold to critical parameters. The bankruptcy wealth threshold is the increasing function of both $d$, which can be treated as the benefit of declaring bankruptcy, and $\alpha$, with $1-\alpha$ representing the flexible cost of bankruptcy. Whereas, the non-monotonic relationship between the bankruptcy wealth threshold and $F$ is because $F$ performs a trade-off between the liquidity constraint boundary in the pre-bankruptcy period and the fixed cost of bankruptcy. Regarding the effect of labour income, we show that the bankruptcy wealth threshold is a concave function of the upper bound of the leisure rate, $L$, that is, it first increases and then decreases: a high value for $L$ permits the agent to have large utility from leisure, while a low value of $L$ ``forces'' the agent to get a large wage, even if low utility from leisure. Moreover, the bankruptcy wealth threshold is strictly decreasing for the wage rate since a worse economic situation requires more wealth to support the post-bankruptcy period.

\bibliographystyle{plainnat} 

\bibliography{ReferenceOfP1}

\newpage
\setcounter{page}{1}
\begin{center}
	\title{\bf Effect of Labour Income on the Optimal Bankruptcy Problem - Online Appendix} 
	\author{\sffamily  Guodong Ding$^{1}$, Daniele Marazzina$^{1,2}$\\ \sffamily\small $^1$ Department of Mathematics, Politecnico di Milano \\ \sffamily\small I-20133, Milano, Italy \\ \sffamily\small $^2${ Corresponding Author. E-mail: daniele.marazzina@polimi.it}}
\end{center}

\appendix
\section{Appendix of Section \ref{Section 2}}

\subsection{Proof of Lemma \ref{Lemma 1}} \label{A.1}
In order to get rid of the constraint $0\leq l\leq L$, we begin with the maximization of $\{u(c,l)-(c+wl)y\}$ with $l\in \mathbb{R}$. From the first-order derivative conditions with respect to $c$ and $l$, we obtain the following equations
	\begin{equation}
	\begin{cases}
	l^{(1-\delta)(1-k)}c^{\delta(1-k)-1}=y,\\
	\frac{1-\delta}{\delta}l^{(1-\delta)(1-k)-1}c^{\delta(1-k)}=w y.\label{Ap.1}
	\end{cases}
	\end{equation}
	The above system entails the optimal consumption and leisure policies as
	\begin{equation*}
	\begin{cases}
	\hat{c}=y^{-\frac{1}{k}}\left(\frac{1-\delta}{\delta w}\right)^{\frac{(1-k)(1-\delta)}{k}},\\
	\hat{l}=y^{-\frac{1}{k}}\left(\frac{1-\delta}{\delta w}\right)^{\frac{1-\delta(1-k)}{k}},
	\end{cases}
	\end{equation*}
	$y>0$ guarantees the positive values of $\hat{c}$ and $\hat{l}$. 
	Besides, with treating $\hat{c}$ and $\hat{l}$ as functions of $y$, we get
	\begin{equation*}
	\tilde{u}^{\prime}(y)=\frac{\partial u}{\partial \hat{c}}\frac{d \hat{c}}{dy}+\frac{\partial u}{\partial \hat{l}}\frac{d \hat{l}}{dy}-(\hat{c}+w\hat{l})-\frac{d\hat{c}}{dy}y-wy\frac{d\hat{l}}{dy}=-(\hat{c}+w\hat{l}).
	\end{equation*}
	Then the remaining constraint of the Legendre-Fenchel transform of Equation (\ref{1.4}) is $\hat{l}\leq L$. Since
	\begin{equation*}
	\hat{l}\leq L \quad\Leftrightarrow\quad y\ge \left(\frac{1-\delta }{\delta w}\right)^{1-\delta(1-k)}L^{-k}\triangleq\tilde{y},
	\end{equation*}
	the optimal leisure plan $\hat{l}$ also satisfies the constraint $l\leq L$ under the condition $y\ge\tilde{y}$. Conversely, this constraint comes into force to make the optimal leisure to be $L$ for the interval $y<\tilde{y}$. Thereafter, we can summarize as follows,
	\begin{equation}
	\hat{l}=y^{-\frac{1}{k}}\left(\frac{1-\delta}{\delta w}\right)^{-\frac{\delta(1-k)-1}{k}}\mathbb{I}_{\{y\ge\tilde{y}\}}+L\mathbb{I}_{\{0<y<\tilde{y}\}}.\label{Ap.2}
	\end{equation}
	The first equation in (\ref{Ap.1}) implies the relationship between the optimal consumption and leisure, $\hat{c}=\left[yl^{-(1-k)(1-\delta)}\right]^{\frac{1}{\delta(1-k)-1}}$. Taking Equation (\ref{Ap.2}) into this relationship, we obtain
	\begin{equation*}
	\hat{c}=y^{-\frac{1}{k}}\left(\frac{1-\delta}{\delta w}\right)^{\frac{(1-k)(1-\delta)}{k}}\mathbb{I}_{\{y\ge\tilde{y}\}}+L^{-\frac{(1-k)(1-\delta)}{\delta(1-k)-1}}y^{\frac{1}{\delta(1-k)-1}}\mathbb{I}_{\{0<y<\tilde{y}\}}.
	\end{equation*}
	Finally, we can deduce the Legendre-Fenchel transform $\tilde{u}(y)$ directly by substituting $\hat{c}$ and $\hat{l}$ into Equation (\ref{1.4}).

\subsection{Proof of Proposition \ref{Proposition 1}} \label{A.2}
Referring to \cite[Section 3.3, Remark 3.3]{karatzas1998methods}, we first apply the It\^o's formula to $\xi(t)X(t)$, $\forall t\in[0,\tau]$,
	\begin{equation*}
	\begin{split}
	d\left(\xi(t)X(t)\right)&=\xi(t)dX(t)+X(t)d\xi(t)\\
	&=-\xi(t)\left(c(t)+d+wl(t)-w\bar{L}\right)dt+\pi(t)\left(\mu-r\right)\xi(t)dt+\sigma\xi(t)\pi(t)dB(t)\\
	&=-\xi(t)\left(c(t)+d+wl(t)-w\bar{L}\right)dt++\sigma\xi(t)\pi(t)d\tilde{B}(t),
	\end{split}
	\end{equation*}
	in which $\tilde{B}(t)$ is the Brownian motion under the $\tilde{\mathbb{P}}$ measure mentioned in Equation (\ref{1.2}). Taking the integral on both sides of the above equation from 0 to $\tau$, we obtain
	\begin{equation*}
	\int_{0}^{\tau}\xi(t)\left(c(t)+d+wl(t)-w\bar{L}\right)dt+\xi(\tau)X(\tau)=x+\int_{0}^{\tau}\sigma\xi(t)\pi(t)d\tilde{B}(t).
	\end{equation*}
	The left-hand side can be rewritten as
	\begin{equation*}
	\int_{0}^{\tau}\!\!\xi(t)(c(t)\!+\!d\!+\!wl(t)\!-\!w\bar{L})dt\!+\!\xi(\tau)X(\tau)\!=\!\int_{0}^{\tau}\!\!\xi(t)(c(t)\!+\!wl(t))dt\!+\!\xi(\tau)\left(X(\tau)\!-\!\frac{d\!-\!w\bar{L}}{r}\right)\!+\!\frac{d\!-\!w\bar{L}}{r},
	\end{equation*}
	the condition $X(\tau)\ge F+\eta\ge\frac{d-w\bar{L}}{r}$ from the definition of admissible control set ensures that the left-hand side is bounded below by the constant $\frac{d-w\bar{L}}{r}$, so the It\^o integral on the right-hand side is proved to be a $\tilde{\mathbb{P}}$-supermartingale by means of Fatou's Lemma. Then, taking the expectation on both sides under the $\tilde{\mathbb{P}}$ measure, we have
	\begin{equation*}
	\tilde{\mathbb{E}}\left[\int_{0}^{\tau}\xi(t)(c(t)+d+wl(t)-w\bar{L})dt+\xi(\tau)X(\tau)\right]=x+\tilde{\mathbb{E}}\left[\int_{0}^{\tau}\sigma\xi(t)\pi(t)d\tilde{B}(t)\right]\leq x,
	\end{equation*}
	which endows us with the desired budget constraint through converting the measure to $\mathbb{P}$,
	\begin{equation*}
	\mathbb{E}\!\left[\!\!\int_{0}^{\tau}\!\!H(t)\!\left(c(t)\!+\!d\!+\!wl(t)\!-\!w\bar{L}\right)dt\!+\!H(\tau)X(\tau)\!\right]\!\!=\!\tilde{\mathbb{E}}\!\left[\!\!\int_{0}^{\tau}\!\!\xi(t)(c(t)\!+\!d\!+\!wl(t)\!-\!w\bar{L})dt\!+\!\xi(\tau)X(\tau)\!\right]\!\!\leq\! x. 
	\end{equation*}
	
\section{Appendix of Section \ref{Section 3}}

\subsection{Proof of Theorem \ref{Theorem 1}} \label{B.1}
Before proving Theorem \ref{Theorem 1}, we insert a lemma which helps us to prove the theorem. 

\begin{lemma}\label{Lemma B.1}
	For any given initial wealth $x\ge 0$, and any given and progressively measurable consumption and leisure processes, $c(t)$, $l(t)$, satisfying
	$\sup\limits_{\tau\in\mathcal{T}}\mathbb{E}\left[\int_{0}^{\tau}\!H(t)(c(t)\!+\!wl(t)\!-\!w\bar{L})dt\right]\!\leq\! x$,
	with $\mathcal{T}$ standing for the set of $\mathcal{F}_{t}$-stopping times,
	there exists a portfolio process $\pi(t)$ making
	\begin{equation*}
	X^{x,c,\pi,l}(t)\ge 0, \quad\forall t\ge 0,
	\end{equation*} 
	holds almost surely.
\end{lemma}
\begin{proof}
	Referring to \cite[Appendix, Lemma 1]{he1993labor}, we first define a new process
	\begin{equation*}
	K(t)\triangleq\int_{0}^{t}(c(s)+wl(s)-w\bar{L})H(s)ds,\quad\forall t\ge 0.
	\end{equation*}
	From the properties of processes $c(t)$, $l(t)$ and $H(t)$, it is directly observed that $\mathbb{E}[K(t)]<\infty$, which implies that $\{K(\tau)\}_{\tau\in\mathcal{T}}$ is uniformly integrable. Then, as in \cite[Appendix I]{dellacherie1982probabilities}, there exists a Snell envelope of $K(t)$ denoted as $\bar{K}(t)$. It is a super-martingale under the $\mathbb{P}$ measure and satisfies 
	\begin{equation*}
	\bar{K}(0)=\sup_{\tau\in\mathcal{T}}\mathbb{E}[K(\tau)],\quad\bar{K}(\infty)=K(\infty).
	\end{equation*}
	By the Doob-Meyer Decomposition Theorem from \cite[Section 1.4, Theorem 4.10]{karatzas1998}, the super-martingale $\bar{K}(t)$ can be represented as
	\begin{equation*}
	\bar{K}(t)=\bar{K}(0)+\bar{M}(t)-\bar{A}(t),
	\end{equation*}
	where $\bar{M}(t)$ is a uniformly integrable martingale under the $\mathbb{P}$ measure with the initial value $\bar{M}(0)=0$, $\bar{A}(t)$ is a strictly increasing process with the initial value $\bar{A}(0)=0$. According to the Martingale Representation Theorem from \cite[Section 11.1, Theorem 11.2]{bjork2009arbitrage}, $\bar{M}(t)$ can be expressed as
	\begin{equation*}
	\bar{M}(t)=\int_{0}^{t}\bar{\rho}(s)dB(s),\quad \forall t\ge 0,
	\end{equation*}
	with an $\mathbb{F}$-adapted process $\bar{\rho}(t)$ satisfying $\int_{0}^{\infty}\bar{\rho}^{2}(s)ds<\infty$ a.s.. Let us define a new process 
	\begin{equation*}
	\bar{X}(t)\triangleq\frac{1}{H(t)}\left[x-\bar{K}(0)+\bar{K}(t)-K(t)+\bar{A}(t)\right].
	\end{equation*}
	Based on the fact that $\bar{K}(0)=\sup\limits_{\tau\in\mathcal{T}}\mathbb{E}[K(\tau)]=	\sup\limits_{\tau\in\mathcal{T}}\mathbb{E}\left[\int_{0}^{\tau}H(t)(c(t)+wl(t)-w\bar{L})dt\right]\leq x$, we can conclude that $\bar{X}(t)$ is a non-negative process with the initial wealth $\bar{X}(0)=x$. We express this process with the martingale $\bar{M}(t)$ as
	\begin{equation*}
	\bar{X}(t)\!=\!\frac{1}{H(t)}\left[x\!+\!\int_{0}^{t}\bar{\rho}(s)dB(s)\!-\!K(t)\right]\!=\!\frac{1}{H(t)}\left[x\!+\!\int_{0}^{t}\bar{\rho}(s)dB(s)\!-\!\int_{0}^{t}(c(s)\!+\!wl(s)\!-\!w\bar{L})H(s)ds\right].
	\end{equation*}
	As for the dynamics of wealth process 
	\begin{equation*}
	dX^{x,c,\pi,l}(t)=rX^{x,c,\pi,l}dt+\pi(t)(\mu-r)dt-c(t)dt+w(\bar{L}-l(t))dt+\sigma\pi(t)dB(t),
	\end{equation*}
	we implement the It\^o's formula to $H(t)X^{x,c,\pi,l}(t)$ to get
	\begin{equation*}
	d(H(t)X^{x,c,\pi,l}(t))=-H(t)X^{x,c,\pi,l}(t)\theta dB(t)-(c(t)+wl(t)-w\bar{L})H(t)dt+\sigma\pi(t)H(t)dB(t).
	\end{equation*}
	If we take the portfolio strategy as $\pi(t)=\frac{\bar{\rho}(t)}{\sigma H(t)}+\frac{\theta X^{x,c,\pi,l}(t)}{\sigma}$, the wealth process is rewritten as
	$X^{x,c,\pi,l}(t)\!=\!\frac{1}{H(t)}\!\left[x\!+\!\int_{0}^{t}\bar{\rho}(s)dB(s)\!-\!\int_{0}^{t}(c(s)\!+\!wl(s)\!-\!w\bar{L})H(s)ds\right]$, which shows that $\bar{X}(t)\!=\!X^{x,c,\pi,l}(t)$, a.s.. The non-negativity of $\bar{X}(t)$ claims $X^{x,c,\pi,l}(t)\ge 0$, a.s., $\forall t\ge 0.$
\end{proof}

Then we move to the proof of Duality Theorem \ref{Theorem 1}. Following \cite[Section 4, Theorem 1]{he1993labor}, the proof mainly contains two aspects: the first part is to show the admissibility of $c^{*}(t)$ and $l^{*}(t)$, and the second part is to claim that they are the optimal consumption-leisure strategy to Problem $(P_{\scriptscriptstyle PB})$.
	\\(1) We begin with verifying that any consumption-leisure strategy satisfying $c^{*}(t)+wl^{*}(t)=-\tilde{u}^{\prime}(\lambda^{*} e^{\gamma t}D_{\scriptscriptstyle PB}^{*}(t)H(t))$ is admissible. Taking any stopping time $\tau$ from $\mathcal{T}$, which is the set of $\mathcal{F}_{t}$-stopping times, we can define a process
	\begin{equation*}
	D^{\epsilon}(t)\triangleq D_{\scriptscriptstyle PB}^{*}(t)+\epsilon\mathbb{I}_{[0,\tau)}(t),
	\end{equation*} 
	where $\epsilon$ a positive constant. It is evident that $D^{\epsilon}(t)$ is a non-negative, non-increasing, and progressively measurable process, that is, $D^{\epsilon}(t)\in\mathcal{D}$. Let us define a function
	\begin{equation*}
	\mathfrak{L}(D(t))\triangleq\mathbb{E}\left[\int_{0}^{\infty}e^{-\gamma t}(\tilde{u}(\lambda^{*} e^{\gamma t}D(t)H(t))+w\bar{L}\lambda^{*} e^{\gamma t}D(t)H(t))dt\right]+\lambda^{*} x D(0).
	\end{equation*}
	Since $D_{\scriptscriptstyle PB}^{*}(t)$ is the optimal solution of Problem $(S_{\scriptscriptstyle PB})$, and $x\ge 0$, we get
	\begin{equation*}
	\begin{split}
	\mathfrak{L}(D_{\scriptscriptstyle PB}^{*}(t))&=\mathbb{E}\left[\int_{0}^{\infty}e^{-\gamma t}(\tilde{u}(\lambda^{*} e^{\gamma t}D_{\scriptscriptstyle PB}^{*}(t)H(t))+w\bar{L}\lambda^{*} e^{\gamma t}D_{\scriptscriptstyle PB}^{*}(t)H(t))dt\right]+\lambda^{*} x D_{\scriptscriptstyle PB}^{*}(0)\\
	&\leq \mathbb{E}\left[\int_{0}^{\infty}e^{-\gamma t}(\tilde{u}(\lambda^{*} e^{\gamma t}D^{\epsilon}(t)H(t))+w\bar{L}\lambda^{*} e^{\gamma t}D^{\epsilon}(t)H(t))dt\right]+\lambda^{*} x (D_{\scriptscriptstyle PB}^{*}(0)+\epsilon)\\
	&=\mathfrak{L}(D^{\epsilon}(t)),\quad\forall t\ge 0.
	\end{split}
	\end{equation*}
	The above inequalities give us
	$\limsup\limits_{\epsilon\downarrow 0}\frac{\mathfrak{L}(D^{\epsilon}(t))-\mathfrak{L}(D_{\scriptscriptstyle PB}^{*}(t))}{\epsilon}\ge 0$ and
	\begin{equation*}
	\limsup_{\epsilon\downarrow 0}\mathbb{E}\left[\int_{0}^{\tau}\left(e^{-\gamma t}\frac{\tilde{u}(\lambda^{*} e^{\gamma t}D^{\epsilon }(t)H(t))-\tilde{u}(\lambda^{*} e^{\gamma t}D_{\scriptscriptstyle PB}^{*}(t)H(t))}{\epsilon}+w\bar{L}\lambda^{*} H(t)\right)dt\right]+\lambda^{*} x\ge 0.
	\end{equation*}
	The decreasing property of $\tilde{u}(\cdot)$ endows us with
	$\tilde{u}(\lambda^{*} e^{\gamma t}D^{\epsilon }(t)H(t))\leq\tilde{u}(\lambda^{*} e^{\gamma t}D_{\scriptscriptstyle PB}^{*}(t)H(t))$. Applying the Fatou's lemma, we have
	\begin{equation*}
	\mathbb{E}\!\left[\!\int_{0}^{\tau}\!\!\lambda^{*} H(t)\tilde{u}^{\prime}(\lambda^{*} e^{\gamma t}\!D_{\scriptscriptstyle PB}^{*}(t)H(t))dt\!\right]\!\ge\!\limsup_{\epsilon\downarrow 0}\mathbb{E}\!\left[\!\int_{0}^{\tau}\!\!e^{\!-\!\gamma t}\frac{\tilde{u}(\lambda^{*} e^{\gamma t}\!D^{\epsilon }(t)H(t))\!-\!\tilde{u}(\lambda^{*} e^{\gamma t}\!D_{\scriptscriptstyle PB}^{*}(t)H(t))}{\epsilon}dt\right].
	\end{equation*}
	Because of $c^{*}(t)+wl^{*}(t)=-\tilde{u}^{\prime}(\lambda^{*} e^{\gamma t}D_{\scriptscriptstyle PB}^{*}(t)H(t))$, we get
	$\mathbb{E}\left[\int_{0}^{\tau}H(t)(c^{*}(t)+wl^{*}(t)-w\bar{L})dt\right]\leq x$.
	Since $\tau$ can be any $\mathcal{F}$-stopping time in the set $\mathcal{T}$, there exists a portfolio strategy $\pi^{*}(t)$ that makes the corresponding wealth process satisfying $X^{x,c^{*},\pi^{*},l^{*}}(t)\ge 0$, $\forall t\ge 0$ based on the result from Lemma \ref{Lemma B.1}.\\
	(2) In this part, we claim that $c^{*}(t)$ and $l^{*}(t)$ are the optimal consumption and leisure to Problem $(P_{\scriptscriptstyle PB})$ under the liquidity constraint. Taking an arbitrary consumption strategy $\{c(t),\pi(t),l(t)\}\in\mathcal{A}_{\scriptscriptstyle PB}(x)$, the proof of Lemma \ref{Lemma B.1} guarantees that there exists a process $\zeta(t)$ satisfying
	\begin{equation}
	\int_{0}^{t}H(s)(c(s)+wl(s)-w\bar{L})ds+H(t)X^{x,c,\pi,l}(t)=x+\int_{0}^{t}\zeta(s)dB(s).\label{Ap.3}
	\end{equation}
	Since $X^{x,c,\pi,l}(t)\ge 0$ a.s., we obtain the following inequality with any process $D(t)\in\mathcal{D}$,
	\begin{equation*}
	\int_{0}^{T}\int_{0}^{t}H(s)(c(s)+wl(s)-w\bar{L})ds dD(t)\ge\int_{0}^{T}\left[x+\int_{0}^{t}\zeta(s)dB(s)\right]dD(t),
	\end{equation*}
	where $T$ is any time satisfying $T\ge t$. Because $D(t)$ is bounded variation, we can integrate by parts and get
	\begin{equation*}
	\begin{split}
	&\int_{0}^{T}D(s)H(s)(c(s)+wl(s)-w\bar{L})ds- \int_{0}^{T}D(s)\zeta(s)dB(s)\leq\\ & \qquad D(0)x+D(T)\left[\int_{0}^{T}H(s)(c(s)+wl(s)-w\bar{L})ds\!-\!x\!-\!\int_{0}^{T}\zeta(s)dB(s)\right].
	\end{split}
	\end{equation*}
	Taking the expectation under the $\mathbb{P}$ measure on both sides and replacing Equation (\ref{Ap.3}), we obtain
	\begin{equation*}
	\mathbb{E}\left[\int_{0}^{T}D(s)H(s)(c(s)+wl(s)-w\bar{L})ds\right]\leq D(0)x-\mathbb{E}\left[D(T)H(T)X^{x,c,\pi,l}(T)\right]\leq D(0)x,
	\end{equation*} 
	then, by Lebesgue's Monotone Convergence Theorem, we have
	\begin{equation*}
	\mathbb{E}\left[\int_{0}^{\infty}D(s)H(s)(c(s)+wl(s)-w\bar{L})ds\right]\leq D(0)x.
	\end{equation*}
	The above inequality keeps true for any admissible control strategy $c(t)$, $\pi(t)$, $l(t)$ and any non-negative, non-increasing process $D(t)$; furthermore, we will show that the inequality changes into equality with the given $c^{*}(t)$, $l^{*}(t)$ and $D_{\scriptscriptstyle PB}^{*}(t)$. We define a new process
	\begin{equation*}
	\bar{D}^{\epsilon}(t)\triangleq D_{\scriptscriptstyle PB}^{*}(t)(1+\epsilon)\in\mathcal{D},
	\end{equation*}
	where $\epsilon$ is a small enough constant. Following the same argument as in the first part of the proof, we have $\mathfrak{L}(\bar{D}^{\epsilon}(t))\ge\mathfrak{L}(D_{\scriptscriptstyle PB}^{*}(t))$ and
	\begin{equation*}
	\begin{split}
	\limsup_{\epsilon\downarrow 0}\frac{\mathfrak{L}(\bar{D}^{\epsilon}(t))\!-\!\mathfrak{L}(D_{\scriptscriptstyle PB}^{*}(t))}{\epsilon}&\!=\!\limsup_{\epsilon\downarrow 0}\mathbb{E}\!\bigg[\int_{0}^{\infty}\!\!\!\!e^{\!-\!\gamma t}\frac{\tilde{u}(\lambda^{*} e^{\gamma t}D_{\scriptscriptstyle PB}^{*}(t)(1\!+\!\epsilon)H(t))\!-\!\tilde{u}(\lambda^{*} e^{\gamma t}D_{\scriptscriptstyle PB}^{*}(t)H(t))}{\epsilon}dt\\
	&\qquad +\int_{0}^{\infty}w\bar{L}\lambda^{*} D_{\scriptscriptstyle PB}^{*}(t)H(t)dt\bigg]+\lambda^{*} x D_{\scriptscriptstyle PB}^{*}(0)\ge 0,
	\end{split}
	\end{equation*}
	\begin{equation*}
	\begin{split}
	\liminf_{\epsilon\uparrow 0}\frac{\mathfrak{L}(\bar{D}^{\epsilon}(t))\!-\!\mathfrak{L}(D_{\scriptscriptstyle PB}^{*}(t))}{\epsilon}&\!=\!\liminf_{\epsilon\uparrow 0}\mathbb{E}\!\bigg[\int_{0}^{\infty}\!\!\!\!e^{\!-\!\gamma t}\frac{\tilde{u}(\lambda^{*} e^{\gamma t}D_{\scriptscriptstyle PB}^{*}(t)(1\!+\!\epsilon)H(t))\!-\!\tilde{u}(\lambda^{*} e^{\gamma t}D_{\scriptscriptstyle PB}^{*}(t)H(t))}{\epsilon}dt\\
	&\qquad +\int_{0}^{\infty}w\bar{L}\lambda^{*} D_{\scriptscriptstyle PB}^{*}(t)H(t)dt\bigg]+\lambda^{*} x D_{\scriptscriptstyle PB}^{*}(0)\leq 0.
	\end{split}
	\end{equation*}
	Applying the Fatou's lemma, we obtain separately
	\begin{equation*}
	\begin{split}
	& \mathbb{E}\left[\int_{0}^{\infty}D_{\scriptscriptstyle PB}^{*}(t)H(t)(c^{*}(t)+wl^{*}(t)-w\bar{L})dt\right]\leq xD_{\scriptscriptstyle PB}^{*}(0),\\
	&\mathbb{E}\left[\int_{0}^{\infty}D_{\scriptscriptstyle PB}^{*}(t)H(t)(c^{*}(t)+wl^{*}(t)-w\bar{L})dt\right]\ge xD_{\scriptscriptstyle PB}^{*}(0),
	\end{split}
	\end{equation*}
	which give us
	$\mathbb{E}\left[\int_{0}^{\infty}D_{\scriptscriptstyle PB}^{*}(t)H(t)(c^{*}(t)+wl^{*}(t)-w\bar{L})dt\right]=xD_{\scriptscriptstyle PB}^{*}(0)$. Subsequently, we define a new optimization problem named $(P_{\scriptscriptstyle PB}^{\prime})$ as
	\begin{equation*}
	\max_{c(t)\ge 0,l(t)\ge 0} \mathbb{E}\left[\int_{0}^{\infty}e^{-\gamma t}u(c(t),l(t))dt\right] \tag{$P_{\scriptscriptstyle PB}^{\prime}$}
	\end{equation*}
	\begin{equation*}
	s.t. \quad\mathbb{E} \left[\int_{0}^{\infty}D_{\scriptscriptstyle PB}^{*}(t)H(t)(c(t)+wl(t)-w\bar{L})dt\right]\leq xD_{\scriptscriptstyle PB}^{*}(0).
	\end{equation*}
	We denote the optimal consumption solution of the above problem is $\tilde{c}^{*}(t)$. From the Lagrange method, we know that $\tilde{c}^{*}(t)+w\tilde{l}^{*}(t)=-\tilde{u}^{\prime}(\tilde{\lambda}e^{\gamma t}D_{\scriptscriptstyle PB}^{*}(t)H(t))$,
	where $\tilde{\lambda}>0$ is the Lagrange multiplier. The constraint of problem $(P_{\scriptscriptstyle PB}^{\prime})$ takes equality when $\tilde{\lambda}=\lambda^{*}$. Then, the condition $c^{*}(t)+wl^{*}(t)=-\tilde{u}^{\prime}(\lambda^{*} e^{\gamma t}D_{\scriptscriptstyle PB}^{*}(t)H(t))$ implies that $c^{*}(t)$ and $l^{*}(t)$ are the optimal control policies of Problem $(P_{\scriptscriptstyle PB}^{\prime})$. Moreover, since the maximum utility of primal problem $(P_{\scriptscriptstyle PB})$ is upper bounded by the maximum utility of $(P_{\scriptscriptstyle PB}^{\prime})$, we can conclude that $c^{*}(t)$ and $l^{*}(t)$ are also the optimal consumption and leisure solutions of Problem $(P_{\scriptscriptstyle PB})$.
	
	\subsection{Solutions of Variational Inequalities (\ref{1.8})} \label{B.2}
	The solution is computed considering two cases: $0<\tilde{y}\leq\hat{z}_{\scriptscriptstyle PB}$ and $0<\hat{z}_{\scriptscriptstyle PB}<\tilde{y}$.

	\paragraph{\textbf{Case 1. $0<\tilde{y}\leq\hat{z}_{\scriptscriptstyle PB}$:}}
	\noindent Condition $(V3)$ in (\ref{1.8}) results in a differential equation,
	\begin{equation}
	-\gamma v_{\scriptscriptstyle PB}(z)+(\gamma-r)z v_{\scriptscriptstyle PB}^{\prime}(z)+\frac{1}{2}\theta^{2}z^{2}v_{\scriptscriptstyle PB}^{\prime\prime}(z)+\tilde{u}(z)+w\bar{L}z=0, \qquad 0<z<\hat{z}_{\scriptscriptstyle PB},\label{Ap.18}
	\end{equation}
	which has the solution
	\begin{equation*}
	v_{\scriptscriptstyle PB}(z)=\begin{cases}
	B_{11,\scriptscriptstyle PB}z^{n_{1}}+B_{21,\scriptscriptstyle PB}z^{n_{2}}+\frac{A_{1}}{\Gamma_{1}}z^{\frac{\delta(1-k)}{\delta(1-k)-1}}+\frac{w(\bar{L}-L)}{r}z, & 0< z<\tilde{y},\\
	
	B_{12,\scriptscriptstyle PB}z^{n_{1}}+B_{22,\scriptscriptstyle PB}z^{n_{2}}+\frac{A_{2}}{\Gamma_{2}}z^{-\frac{1-k}{k}}+\frac{w\bar{L}}{r}z,& \tilde{y}\leq z< \hat{z}_{\scriptscriptstyle PB}.
	\end{cases}
	\end{equation*}
	For avoiding the explosion of term $z^{n_{1}}$ when $z$ goes to 0, we set $B_{11,\scriptscriptstyle PB}=0$. Then four parameters are left to be determined, which are $B_{21,\scriptscriptstyle PB}$, $B_{12,\scriptscriptstyle PB}$, $B_{22,\scriptscriptstyle PB}$ and $\hat{z}_{\scriptscriptstyle PB}$. To accomplish this task, we use the smooth conditions at $z=\tilde{y}$ and $z=\hat{z}_{\scriptscriptstyle PB}$ to construct a four-equation system:
	\begin{itemize}
		\item $\mathcal{C}^{0}$ condition at $z=\tilde{y}$
		\begin{equation*}
		B_{21,\scriptscriptstyle PB}\tilde{y}^{n_{2}}+\frac{A_{1}}{\Gamma_{1}}\tilde{y}^{\frac{\delta(1-k)}{\delta(1-k)-1}}-\frac{wL}{r}\tilde{y}=B_{12,\scriptscriptstyle PB}\tilde{y}^{n_{1}}+B_{22,\scriptscriptstyle PB}\tilde{y}^{n_{2}}+\frac{A_{2}}{\Gamma_{2}}\tilde{y}^{-\frac{1-k}{k}};\label{Ap.4}
		\end{equation*}
		\item $\mathcal{C}^{1}$ condition at $z=\tilde{y}$
		\begin{equation*}
		n_{2}B_{21,\scriptscriptstyle PB}\tilde{y}^{n_{2}\!-\!1}\!+\!\frac{\delta(1\!-\!k)}{\delta(1\!-\!k)\!-\!1}\frac{A_{1}}{\Gamma_{1}}\tilde{y}^{\frac{1}{\delta(1\!-\!k)\!-\!1}}\!-\!\frac{wL}{r}\!=\! n_{1}B_{12,\scriptscriptstyle PB}\tilde{y}^{n_{1}\!-\!1}\!+\!n_{2}B_{22,\scriptscriptstyle PB}\tilde{y}^{n_{2}\!-\!1}\!-\!\frac{1\!-\!k}{k}\frac{A_{2}}{\Gamma_{2}}\tilde{y}^{\!-\!\frac{1}{k}};\label{Ap.5}
		\end{equation*}
		\item $\mathcal{C}^{1}$ condition at $z=\hat{z}_{\scriptscriptstyle PB}$
		\begin{equation*}
		n_{1}B_{12,\scriptscriptstyle PB}\hat{z}_{\scriptscriptstyle PB}^{n_{1}-1}+n_{2}B_{22,\scriptscriptstyle PB}\hat{z}_{\scriptscriptstyle PB}^{n_{2}-1}-\frac{1-k}{k}\frac{A_{2}}{\Gamma_{2}}\hat{z}_{\scriptscriptstyle PB}^{-\frac{1}{k}}+\frac{w\bar{L}}{r}=0;\label{Ap.6}
		\end{equation*}
		\item $\mathcal{C}^{2}$ condition at $z=\hat{z}_{\scriptscriptstyle PB}$
		\begin{equation*}
		n_{1}(n_{1}-1)B_{12,\scriptscriptstyle PB}\hat{z}_{\scriptscriptstyle PB}^{n_{1}-2}+n_{2}(n_{2}-1)B_{22,\scriptscriptstyle PB}\hat{z}_{\scriptscriptstyle PB}^{n_{2}-2}+\frac{1-k}{k^{2}}\frac{A_{2}}{\Gamma_{2}}\hat{z}_{\scriptscriptstyle PB}^{-\frac{1+k}{k}}=0.\label{Ap.7}
		\end{equation*}	
	\end{itemize}

	\paragraph{\textbf{Case 2. $0<\hat{z}_{\scriptscriptstyle PB}<\tilde{y}$:}}
	\noindent The same argument with the previous case, we first handle Condition $(V3)$ in (\ref{1.8}). Recalling Lemma \ref{Lemma 1}, the corresponding interval $0<z<\hat{z}_{\scriptscriptstyle PB}$ restricts the function $\tilde{u}(z)$ only takes the form $\tilde{u}(z)=A_{1}z^{\frac{\delta(1-k)}{\delta(1-k)-1}}-wLz$, which is identical with the piece of $0<z<\tilde{y}$ in Case 1. Hence, the differential equation from $(V3)$ has the same solution, only changing the parameters' notations from $B_{11,\scriptscriptstyle PB}$ to $B_{1,\scriptscriptstyle PB}$, and $B_{21,\scriptscriptstyle PB}$ to $B_{2,\scriptscriptstyle PB}$, that is,
	\begin{equation*}
	v_{\scriptscriptstyle PB}(z)=B_{1,\scriptscriptstyle PB}z^{n_{1}}+B_{2,\scriptscriptstyle PB}z^{n_{2}}+\frac{A_{1}}{\Gamma_{1}}z^{\frac{\delta(1-k)}{\delta(1-k)-1}}+\frac{w(\bar{L}-L)}{r}z, \quad 0< z<\hat{z}_{\scriptscriptstyle PB}.
	\end{equation*}
	We set the coefficient $B_{1,\scriptscriptstyle PB}=0$ for the reason that the term $z^{n_{1}}$ goes to $\infty$ as $z$ approaches 0, which violates the boundedness assumption of $v_{\scriptscriptstyle PB}(z)$. Subsequently, using the smooth condition at $z=\hat{z}_{\scriptscriptstyle PB}$, a two-equation system is established to determine the exact values of $B_{2,\scriptscriptstyle PB}$ and $\hat{z}_{\scriptscriptstyle PB}$:
	\begin{itemize}
		\item $\mathcal{C}^{1}$ condition at $z=\hat{z}_{\scriptscriptstyle PB}$
		\begin{equation*}
		n_{2}B_{2,\scriptscriptstyle PB}\hat{z}_{\scriptscriptstyle PB}^{n_{2}-1}+\frac{\delta(1-k)}{\delta(1-k)-1}\frac{A_{1}}{\Gamma_{1}}\hat{z}_{\scriptscriptstyle PB}^{\frac{1}{\delta(1-k)-1}}+\frac{w(\bar{L}-L)}{r}=0;
		\label{Ap.10}
		\end{equation*}
		\item $\mathcal{C}^{2}$ condition at $z=\hat{z}_{\scriptscriptstyle PB}$
		\begin{equation*}
		n_{2}(n_{2}-1)B_{2,\scriptscriptstyle PB}\hat{z}_{\scriptscriptstyle PB}^{n_{2}-2}+\frac{\delta(1-k)}{(\delta(1-k)-1)^{2}}\frac{A_{1}}{\Gamma_{1}}\hat{z}_{\scriptscriptstyle PB}^{\frac{2-\delta(1-k)}{\delta(1-k)-1}}=0.
		\label{Ap.11}
		\end{equation*}
	\end{itemize}
\section{Appendix of Section \ref{Section 4}}

\subsection{Proof of Proposition \ref{Proposition 2}} \label{C.1}
Let us first introduce a process $M_{\tau}(t)\!\triangleq\! H(t)X(t)\!+\!\int_{0}^{t}\left(c(s)\!+\!d\!+\!wl(s)\!-\!w\bar{L}\right)H(s)ds$, $\forall t\in[0,\tau]$.
	Following \cite[Section 3.9, Theorem 9.4]{karatzas1998methods}, we firstly claim that $M_{\tau}(t)$ is a $\mathbb{P}$-martingale. Let $c(t)$ and $l(t)$ be the consumption and leisure processes such that
	\begin{equation*}
	\mathbb{E}\left[\int_{0}^{\infty}H(s)\left(c(s)+d+wl(s)-w\bar{L}\right)ds\right]=x.
	\end{equation*}
	For any fixed stopping time $\tau\in\mathcal{T}$, we define $\zeta(\tau)\!\triangleq\!\frac{1}{H(\tau)}\mathbb{E}\left[\left.\int_{\tau}^{\infty}H(s)\left(c(s)\!+\!d\!+\!wl(s)\!-\!w\bar{L}\right)ds\right|\mathcal{F}_{\tau}\right]$. Then we have $x=\mathbb{E}\left[\int_{0}^{\tau}H(s)\left(c(s)+d+wl(s)-w\bar{L}\right)ds+H(\tau)\zeta(\tau)\right]$.
	\cite[Section 3.3, Theorem 3.5]{karatzas1998methods} implies that there exists a portfolio process $\pi_{\tau}=\{\pi_{\tau}(t):0\leq t\leq\tau\}$ satisfying
	$\zeta(\tau)=X^{x,c,\pi_{\tau},l}(\tau)$. 
	Moreover, since
	\begin{equation*}
	\begin{split}
	\mathbb{E}\left[M_{\tau}(\tau)\right]&=\mathbb{E}\left[H(\tau)X^{x,c,\pi_{\tau},l}(\tau)\right]+\mathbb{E}\left[\int_{0}^{\tau}H(s)\left(c(s)+d+wl(s)-w\bar{L}\right)ds\right]\\
	&=\mathbb{E}\left[H(\tau)\zeta(\tau)\right]+\mathbb{E}\left[\int_{0}^{\tau}H(s)\left(c(s)+d+wl(s)-w\bar{L}\right)ds\right]\\
	&=\mathbb{E}\left[\mathbb{E}\left[\left.\int_{\tau}^{\infty}H(s)\left(c(s)\!+\!d\!+\!wl(s)\!-\!w\bar{L}\right)ds\right|\!\mathcal{F}_{\tau}\!\right]\!+\!\int_{0}^{\tau}\!H(s)\left(c(s)\!+\!d\!+\!wl(s)\!-\!w\bar{L}\right)ds\right]\\
	&=x=M_{\tau}(0),
	\end{split}
	\end{equation*}
	the process $\{M_{\tau}(t):0\leq t\leq\tau\}$ is a $\mathbb{P}$-martingale. Combined with the constraint $X(t)\ge F+\eta$, $\forall t\in[0,\tau]$, we have
	\begin{equation*}
	\mathbb{E}\!\left[\!H(\tau)X(\tau)\!+\!\!\left.\int_{0}^{\tau}\!\!H(s)\!\left(c(s)\!+\!d\!+\!wl(s)\!-\!w\bar{L}\right)\!ds\right|\!\mathcal{F}_{t}\right]\!\!=\!\!H(t)X(t)\!+\!\int_{0}^{t}\!\!H(s)\!\left(c(s)\!+\!d\!+\!wl(s)\!-\!w\bar{L}\right)\!ds,
	\end{equation*}
	then,
	\begin{equation*}
	\mathbb{E}\left[\left.\int_{t}^{\tau}\frac{H(s)}{H(t)}\left(c(s)+d+wl(s)-w\bar{L}\right)ds+\frac{H(\tau)}{H(t)}X(\tau)\right|\mathcal{F}_{t}\right]=X(t)\ge F+\eta,\quad\forall t\in[0,\tau].
	\end{equation*}
	
\subsection{Proof of Theorem \ref{Theorem 2}} \label{C.2}
The proof here is consistent with Appendix \ref{B.1}, but some modification is needed due to the stop-time embedding. We first introduce a lemma, which will be used in the proof of the duality theorem.
\begin{lemma}\label{Lemma C.1}
	For any given initial wealth $x\!\ge\!F\!+\!\eta$, any $\mathcal{F}_{t}$-stopping time $\tau$ with $\mathbb{P}(\tau\!<\!\infty)\!=\!1$, any $\mathcal{F}_{\tau}$-measurable random variable $Q$ with $\mathbb{P}(Q\!\ge\! F\!+\!\eta)\!=\!1$ under the $\mathbb{P}$ measure, and any given progressively measurable consumption and leisure processes $c(t)\ge 0$, $0\leq l(t)\leq L$, satisfying
	$\sup\limits_{\tilde{\tau}\in\mathcal{S}}\!\mathbb{E}\!\left[-\!\int_{\tilde{\tau}}^{\tau}\!H(t)(c(t)\!+\!wl(t)\!-\!w\bar{L}\!+\!d)dt\!-\!H(\tau)Q\right]\leq \!-\!F\!+\!\eta$,
	where $\mathcal{S}$ stands for the set of $\mathcal{F}_{t}$-stopping times before the fixed stopping time $\tau$, and 
	$\mathbb{E}\!\left[\!\int_{0}^{\tau}\!H(t)\left(c(t)\!+\!wl(t)\!-\!w\bar{L}\!+\!d\right)dt\!+\!H(\tau)Q\right]\!=\! x$, there exists a portfolio process $\pi(t)$ making
	$X^{x,c,\pi,l}(t)\!\ge\!F\!+\!\eta$, $\forall t \!\in\![0,\tau]$, and $X^{x,c,\pi,l}(\tau)\!=\!Q$ hold almost surely.
\end{lemma}
\begin{proof}
	Following the similar argument with \cite[Appendix, Lemma 1]{he1993labor}, we first define a new process
	\begin{equation*}
	K(t)\triangleq-\int_{t}^{\tau}H(s)(c(s)+wl(s)-w\bar{L}+d)ds-H(\tau)Q,\quad\forall t\in [0,\tau].
	\end{equation*}
	From the properties of processes $c(t)$, $l(t)$ and $H(t)$, it can be observed that $\mathbb{E}[K(t)]<\infty$, which implies $\{K(\tilde{\tau})\}_{\tilde{\tau}\in\mathcal{S}}$ is uniformly integrable. Therefore, there exists a Snell envelope of $K(t)$ denoted as $\bar{K}(t)$. It is a super-martingale under the $\mathbb{P}$ measure and satisfies 
	\begin{equation*}
	\bar{K}(0)=\sup_{\tilde{\tau}\in\mathcal{S}}\mathbb{E}[K(\tilde{\tau})],\quad\mbox{and}\quad \bar{K}(\tau)=K(\tau).
	\end{equation*}
	By the Doob-Meyer Decomposition Theorem from \cite[Section 1.4, Theorem 4.10]{karatzas1998}, the super-martingale $\bar{K}(t)$ can be decomposed into
	\begin{equation*}
	\bar{K}(t)=\bar{K}(0)+\bar{M}(t)-\bar{A}(t),
	\end{equation*}
	where $\bar{M}(t)$ is a uniformly integrable martingale under the $\mathbb{P}$ measure with the initial value $\bar{M}(0)=0$, $\bar{A}(t)$ is a strictly increasing process with the initial value $\bar{A}(0)=0$. According to the Martingale Representation Theorem from \cite[Section 11.1, Theorem 11.2]{bjork2009arbitrage}, $\bar{M}(t)$ can be expressed as
	\begin{equation*}
	\bar{M}(t)=\int_{0}^{t}\bar{\rho}(s)dB(s),\quad \forall t\in [0,\tau],
	\end{equation*}
	with an $\mathbb{F}$-adapted process $\bar{\rho}(t)$ satisfying $\int_{0}^{\infty}\bar{\rho}^{2}(s)ds<\infty$ a.s.. Let us define a new process 
	\begin{equation*}
	\bar{X}(t)\triangleq\frac{1}{H(t)}\mathbb{E}\left[-\bar{K}(0)+\bar{K}(t)-K(t)+\bar{A}(t)-\bar{M}(\tau)|\mathcal{F}_{t}\right]-F-\eta.
	\end{equation*}
	It can be verified that
	\begin{equation*}
	\bar{X}(\tau)=Q-F-\eta, \quad
	\mbox{and}\quad \bar{X}(0)=x-F-\eta,
	\end{equation*}
	using the condition 
	$\mathbb{E}\left[\int_{0}^{\tau}H(t)\left(c(t)+wl(t)-w\bar{L}+d\right)dt+H(\tau)Q\right]= x$,
	and the martingale property of $\bar{M}(t)$. Further, because of $\bar{K}(0)=\sup\limits_{\tilde{\tau}\in\mathcal{S}}\mathbb{E}[K(\tilde{\tau})]$, we can prove that $\bar{K}(0)\ge K(t)$, $\forall t\in [0,\tau]$ by constructing a contradiction. Let us assume that $\tilde{\tau}^{*}$ attains the supremum within the expression of $\bar{K}(0)$, i.e., $\bar{K}(0)=K(\tilde{\tau}^{*})$, and introduce a stopping time as $\bar{\tau}\triangleq\inf\{0\leq t\leq \tau:K(t)>K(\tilde{\tau}^{*})\}$. Since $\bar{\tau}\in\mathcal{S}$, we have $K(\bar{\tau})\leq K(\tilde{\tau}^{*})$, which is contradictory to the definition of $\bar{\tau}$. Then, based on the fact that $\bar{X}(t)=\frac{1}{H(t)}\mathbb{E}[-\bar{K}(0)|\mathcal{F}_{t}]+\frac{1}{H(t)}\mathbb{E}[\bar{K}(0)-K(t)|\mathcal{F}_{t}]-F-\eta$,
	we can conclude that $\bar{X}(t)\ge 0$, a.s.. Additionally, the process $\bar{X}(t)$ is re-expressed in terms of the martingale $\bar{M}(t)$ as
	\begin{equation}
	\bar{X}(t)\!=\!\frac{1}{H(t)}\mathbb{E}\!\left[\left.\int_{t}^{\tau}\!H(s)(c(s)\!+\!wl(s)\!-\!w\bar{L}\!+\!d)ds\!+\!H(\tau)Q\!-\!\int_{t}^{\tau}\!\bar{\rho}(s)dB(s)\right|\mathcal{F}_{t}\right]\!-\!F\!-\!\eta.\label{Ap.13}
	\end{equation}
	As for the wealth process 
	\begin{equation*}
	\begin{cases}
	dX^{x,c,\pi,l}(t)=rX^{x,c,\pi,l}dt+\pi(t)(\mu-r)dt-(c(t)+wl(t)-w\bar{L}+d)dt+\sigma\pi(t)dB(t),\\
	X^{x,c,\pi,l}(0)=x,
	\end{cases}
	\end{equation*}
	by implementing the It\^o's formula to $H(t)X^{x,c,\pi,l}(t)$ and adopting the portfolio strategy as
	\begin{equation*}
	\pi(t)=\frac{\bar{\rho}(t)}{\sigma H(t)}+\frac{\theta X^{x,c,\pi,l}(t)}{\sigma},
	\end{equation*}
	it is rewritten as $d(H(t)X^{x,c.\pi,l}(t))=\bar{\rho}(t)dB(t)-(c(t)+wl(t)-w\bar{L}+d)H(t)dt$. Taking the integral from $t$ to $\tau$, and then the conditional expectation w.r.t. $\mathcal{F}_{t}$ on both sides of the above equation, we obtain 
	\begin{equation*}
	X^{x,c,\pi,l}(t)\!=\!\frac{1}{H(t)}\mathbb{E}\!\left[\int_{t}^{\tau}\!H(s)(c(s)\!+\!wl(s)\!-\!w\bar{L}\!+\!d)ds\!+\!H(\tau)X^{x,c,\pi,l}(\tau)\!-\!\int_{t}^{\tau}\!\bar{\rho}(s)dB(s)|\mathcal{F}_{t}\right].
	\end{equation*}
	Since $X^{x,c,\pi,l}(0)\!=\!x\!=\!\mathbb{E}\!\left[\int_{0}^{\tau}\!H(t)(c(t)\!+\!wl(t)\!-\!w\bar{L}\!+\!d)dt\!+\!H(\tau)X^{x,c,\pi,l}(\tau)\right]$, we have $X^{x,c,\pi,l}(\tau)\!=\!Q$, a.s.. Finally, through comparing the process of $X^{x,c,\pi,l}(t)$ with Equation (\ref{Ap.13}),
	\begin{equation*}
	\bar{X}(t)\!=\!X^{x,c,\pi}(t)\!-\!F-\eta,\quad  a.s.
	\end{equation*} 
	is observed. The non-negativity of $\bar{X}(t)$ claims $X^{x,c,\pi,l}(t)\!\ge\!F+\eta$, a.s., $\forall t\in[0,\tau]$.
\end{proof}
With the aid of the above lemma, we can complete the statement and proof of Duality Theorem \ref{Theorem 2}.
	Referring to \cite[Section 4, Theorem 1]{he1993labor}, the proof procedure is divided into two aspects: the first part is focused on the admissibility of $c^{*}(t)$ and $l^{*}(t)$, and the second part is revolved around claiming that $c^{*}(t)$ and $l^{*}(t)$ are the optimal consumption-leisure strategy to the primal optimization problem.
	\\(1) We begin verifying that any consumption-leisure strategy satisfying
	\begin{equation*}
	c^{*}(t)+wl^{*}(t)=-\tilde{u}^{\prime}(\lambda^{*} e^{\gamma t}D^{*}(t)H(t)), \quad \mbox{and}\quad X^{x,c^{*},\pi^{*},l^{*}}(\tau)=-\tilde{U}^{\prime}(\lambda^{*} e^{\gamma\tau}D^{*}(\tau)H(\tau)),
	\end{equation*} 
	is admissible. Taking any stopping time $\tilde{\tau}$ from $\mathcal{S}$, we can define a process $D^{\epsilon}(t)\triangleq D^{*}(t)+\epsilon\mathbb{I}_{[0,\tilde{\tau})}(t)$,
	where $\epsilon$ is a positive constant. It is evident that $D^{\epsilon}(t)$ is a non-negative, non-increasing and progressively measurable process, that is, $D^{\epsilon}(t)\in\mathcal{D}$. Let us define a function
	\begin{equation*}
	\begin{split}
	\mathfrak{L}(D(t))&\triangleq\mathbb{E}\!\left[\int_{0}^{\tau}\!e^{\!-\!\gamma t}\left(\tilde{u}(\lambda^{*} D(t)e^{\gamma t}H(t))\!-\!(d\!-\!w\bar{L})\lambda^{*} e^{\gamma t}D(t)H(t)\right)dt\!+\!e^{\!-\!\gamma\tau}\tilde{U}(\lambda^{*} D(\tau)e^{\gamma\tau}H(\tau))\right]
	\\
	&\quad+\lambda^{*}\mathbb{E}\left[\int_{0}^{\tau}(F+\eta)H(t)dD(t)\right]+\lambda^{*}(x-(F+\eta))D(0).
	\end{split}
	\end{equation*}
	Considering $D^{*}(t)$ is the optimal solution of problem $(S_{\tau})$ and the fact $x\ge F+\eta$, we obtain
	\begin{equation*}
	\begin{split}
	\mathfrak{L}(D^{*}(t))&\!=\!\mathbb{E}\!\left[\!\int_{0}^{\tau}\!\!e^{\!-\!\gamma t}\left(\tilde{u}(\lambda^{*} D^{*}(t)e^{\gamma t}H(t))\!-\!(d\!-\!w\bar{L})\lambda^{*} e^{\gamma t}D^{*}(t)H(t)\right)dt\!+\!e^{\!-\!\gamma\tau}\tilde{U}(\lambda^{*} D^{*}(\tau)e^{\gamma\tau}H(\tau))\!\right]
	\\
	&\quad+\lambda^{*}\mathbb{E}\left[\int_{0}^{\tau}(F+\eta)H(t)dD^{*}(t)\right]+\lambda^{*}(x-(F+\eta))D^{*}(0)\\
	&\!\leq\! \mathbb{E}\!\left[\!\int_{0}^{\tau}\!\!e^{\!-\!\gamma t}\left(\tilde{u}(\lambda^{*} D^{\epsilon}(t)e^{\gamma t}H(t))\!-\!(d\!-\!w\bar{L})\lambda^{*} e^{\gamma t}D^{\epsilon}(t)H(t)\right)dt\!+\!e^{\!-\!\gamma\tau}\tilde{U}(\lambda^{*} D^{\epsilon}(\tau)e^{\gamma\tau}H(\tau))\!\right]
	\\
	&\quad+\lambda^{*}\mathbb{E}\left[\int_{0}^{\tau}(F+\eta)H(t)dD^{\epsilon}(t)\right]+\lambda^{*}(x-(F+\eta))(D^{*}(0)+\epsilon)\\
	&=\mathfrak{L}(D^{\epsilon}(t)),\quad\forall t\ge \in[0,\tau].
	\end{split}
	\end{equation*}
	The above inequalities give us
	$\limsup\limits_{\epsilon\downarrow 0}\frac{\mathfrak{L}(D^{\epsilon}(t))-\mathfrak{L}(D^{*}(t))}{\epsilon}\ge 0$, combining with $dD^{\epsilon}(t)=dD^{*}(t)$, $D^{\epsilon}(\tau)=D^{*}(\tau)$, we have
	\begin{equation*}
	\limsup_{\epsilon\downarrow 0}\mathbb{E}\!\left[\!\int_{0}^{\tilde{\tau}}\!\left(e^{-\gamma t}\frac{\tilde{u}(\lambda^{*} e^{\gamma t}D^{\epsilon }(t)H(t))\!-\!\tilde{u}(\lambda^{*} e^{\gamma t}D^{*}(t)H(t))}{\epsilon}\!-\!(d\!-\!w\bar{L})\lambda^{*} H(t)\right)dt\right]\!+\!\lambda^{*}(x\!-\!(F\!+\!\eta))\!\ge\! 0.
	\end{equation*}
	The decreasing property of $\tilde{u}(\cdot)$ endows us with
	$\tilde{u}(\lambda^{*} e^{\gamma t}D^{\epsilon }(t)H(t))\leq\tilde{u}(\lambda^{*} e^{\gamma t}D^{*}(t)H(t))$. Applying the Fatou's lemma, we have
	\begin{equation*}
	\mathbb{E}\!\left[\!\int_{0}^{\tilde{\tau}}\!\!\!\!e^{\!-\!\gamma t}\!\tilde{u}^{\prime}(\lambda^{*} e^{\gamma t}\!D^{*}(t)H(t))\lambda^{*} e^{\gamma t}\!H(t)dt\!\right]\!\!\ge\!\limsup_{\epsilon\downarrow 0}\mathbb{E}\!\left[\!\int_{0}^{\tilde{\tau}}\!\!\!\!e^{\!-\!\gamma t}\!\frac{\tilde{u}(\lambda^{*} e^{\gamma t}\!D^{\epsilon }(t)H(t))\!-\!\tilde{u}(\lambda^{*} e^{\gamma t}\!D^{*}(t)H(t))}{\epsilon}dt\!\right]\!.
	\end{equation*}
	Because of $c^{*}(t)+wl^{*}(t)=-\tilde{u}^{\prime}(\lambda^{*} e^{\gamma t}D^{*}(t)H(t))$, we get
	\begin{equation}
	\mathbb{E}\left[\int_{0}^{\tilde{\tau}}H(t)(c^{*}(t)+wl^{*}(t)-w\bar{L}+d)dt\right]\leq x-(F+\eta).\label{Ap.14}
	\end{equation}
	Following the same technique, if defining $D^{\epsilon}(t)\triangleq D^{*}(t)+\epsilon\mathbb{I}_{[0,\tau)}(t)$, we can obtain that
	\begin{equation}
	\mathbb{E}\left[\int_{0}^{\tau}H(t)(c^{*}(t)+wl^{*}(t)-w\bar{L}+d)dt\right]\leq x-(F+\eta).\label{Ap.15}
	\end{equation}
We now follow the same argument as before, introducing a new process 
	$\tilde{D}^{\epsilon}(t)\triangleq D^{*}(t)+\epsilon\in\mathcal{D}$, here $\epsilon$ is no longer required to be positive, but a sufficiently small real number. Defining a function
	\begin{equation*}
	\begin{split}
	\tilde{\mathfrak{L}}(D(t))&\triangleq\mathbb{E}\left[\int_{0}^{\tau}\!e^{\!-\!\gamma t}\left(\tilde{u}(\lambda^{*} D(t)e^{\gamma t}H(t))\!-\!(d\!-\!w\bar{L})\lambda^{*} e^{\gamma t}D(t)H(t)\right)dt\!+\!e^{\!-\!\gamma\tau}\tilde{U}(\lambda^{*} D(\tau)e^{\gamma\tau}H(\tau))\right]
	\\
	&\quad+\lambda^{*}\mathbb{E}\left[\int_{0}^{\tau}(F+\eta)H(t)dD(t)\right]+\lambda^{*} x D(0),
	\end{split}
	\end{equation*}
	we have $\tilde{\mathfrak{L}}(\tilde{D}^{\epsilon}(t))\ge\tilde{\mathfrak{L}}(D^{*}(t))$, and
	\begin{equation*}
	\begin{split}
	\limsup_{\epsilon\downarrow 0}&\frac{\tilde{\mathfrak{L}}(\tilde{D}^{\epsilon}(t))-\tilde{\mathfrak{L}}(D^{*}(t))}{\epsilon}=\limsup_{\epsilon\downarrow 0}\mathbb{E}\left[\int_{0}^{\tau}e^{-\gamma t}\frac{\tilde{u}(\lambda^{*} e^{\gamma t}\tilde{D}^{\epsilon}(t)H(t))-\tilde{u}(\lambda^{*} e^{\gamma t}D^{*}(t)H(t))}{\epsilon}dt\right.\\
	&-\left.\int_{0}^{\tau}(d-w\bar{L})\lambda^{*} H(t)dt+e^{-\gamma \tau}\frac{\tilde{U}(\lambda^{*} e^{\gamma \tau}\tilde{D}^{\epsilon}(\tau)H(\tau))-\tilde{U}(\lambda^{*} e^{\gamma\tau}D^{*}(\tau)H(\tau))}{\epsilon}\right]+\lambda^{*} x\ge 0,
	\end{split}
	\end{equation*}
	\begin{equation*}
	\begin{split}
	\liminf_{\epsilon\uparrow 0}&\frac{\tilde{\mathfrak{L}}(\tilde{D}^{\epsilon}(t))-\tilde{\mathfrak{L}}(D^{*}(t))}{\epsilon}=\liminf_{\epsilon\uparrow 0}\mathbb{E}\left[\int_{0}^{\tau}e^{-\gamma t}\frac{\tilde{u}(\lambda^{*} e^{\gamma t}\tilde{D}^{\epsilon}(t)H(t))-\tilde{u}(\lambda^{*} e^{\gamma t}D^{*}(t)H(t))}{\epsilon}dt\right.\\
	&-\left.\int_{0}^{\tau}(d-w\bar{L})\lambda^{*} H(t)dt+e^{-\gamma \tau}\frac{\tilde{U}(\lambda^{*} e^{\gamma \tau}\tilde{D}^{\epsilon}(\tau)H(\tau))-\tilde{U}(\lambda^{*} e^{\gamma\tau}D^{*}(\tau)H(\tau))}{\epsilon}\right]+\lambda^{*} x\leq 0.
	\end{split}
	\end{equation*}
	Given the conditions
	\begin{equation*}
	c^{*}(t)+wl^{*}(t)=-\tilde{u}^{\prime}(\lambda^{*} e^{\gamma t}D^{*}(t)H(t)), \quad \mbox{and}\quad X^{x,c^{*},\pi^{*},l^{*}}(\tau)\!=\!-\tilde{U}^{\prime}(\lambda^{*} e^{\gamma\tau}D^{*}(\tau)H(\tau)), 
	\end{equation*} 
	the Fatou's lemma entails the following relations respectively:
	\begin{equation*}
	\mathbb{E}\left[\int_{0}^{\tau}H(t)\left(c^{*}(t)+wl^{*}(t)-w\bar{L}+d\right)dt+H(\tau)X^{x,c^{*},\pi^{*},l^{*}}(\tau)\right]\leq x,
	\end{equation*}
	\begin{equation*}
	\mathbb{E}\left[\int_{0}^{\tau}H(t)\left(c^{*}(t)+wl^{*}(t)-w\bar{L}+d\right)dt+H(\tau)X^{x,c^{*},\pi^{*},l^{*}}(\tau)\right]\ge x,
	\end{equation*}
	which lead to
	\begin{equation}
	\mathbb{E}\left[\int_{0}^{\tau}H(t)\left(c^{*}(t)+wl^{*}(t)-w\bar{L}+d\right)dt+H(\tau)X^{x,c^{*},\pi^{*},l^{*}}(\tau)\right]= x.\label{Ap.16}
	\end{equation}
	Subtracting Equation (\ref{Ap.16}) from (\ref{Ap.14}), we get
	\begin{equation*}
	\mathbb{E}\left[\int_{\tilde{\tau}}^{\tau}H(t)(c(t)+wl(t)-w\bar{L}+d)dt+H(\tau)X^{x,c^{*},\pi^{*},l^{*}}(\tau)\right]\ge F+\eta,
	\end{equation*}
	which is equivalent to
	$\mathbb{E}\left[-\int_{\tilde{\tau}}^{\tau}H(t)(c(t)+wl(t)-w\bar{L}+d)dt-H(\tau)X^{x,c^{*},\pi^{*},l^{*}}(\tau)\right]\leq -(F+\eta)$, for any stopping time $\tilde{\tau}\in\mathcal{S}$. Additionally, subtracting Equation  (\ref{Ap.16}) from (\ref{Ap.15}), we obtain
	$\mathbb{E}[H(\tau)X^{x,c^{*},\pi^{*},l^{*}}(\tau)]\ge F+\eta$,
	which implies that $X^{x,c^{*},\pi^{*},l^{*}}(\tau)\ge F+\eta$ a.s.. Since $\tilde{\tau}$ can be any $\mathcal{F}_{t}$-stopping time in the set $\mathcal{S}$, there exists a portfolio strategy $\pi^{*}(t)$ that makes the corresponding wealth process satisfying $X^{x,c^{*},\pi^{*},l^{*}}(t)\ge F+\eta$, $\forall t\in [0,\tau]$, according to Lemma \ref{Lemma C.1}.\\ 
	(2) Then we turn to claim that $c^{*}(t)$ and $l^{*}(t)$ are the optimal consumption and leisure to the problem $(P_{\tau})$. Taking an arbitrary control strategy $\{c(t),\pi(t),l(t)\}\!\in\!\mathcal{A}_{\tau}(x)$, the proof of Lemma \ref{Lemma C.1} guarantees that there exists a process $\zeta(t)$ satisfying
	\begin{equation}
	\int_{0}^{t}H(s)(c(s)+wl(s)-w\bar{L}+d)ds+H(t)X^{x,c,\pi,l}(t)=x+\int_{0}^{t}\zeta(s)dB(s),\quad \forall t\in[0,\tau].\label{Ap.17}
	\end{equation}
	Since $X^{x,c,\pi,l}(t)\ge F+\eta$ a.s., $\forall t\in[0,\tau]$, we obtain the following inequality with any process $D(t)\in\mathcal{D}$,
	\begin{equation*}
	\int_{0}^{\tau}\int_{0}^{t}H(s)(c(s)+wl(s)-w\bar{L}+d)ds dD(t)+\int_{0}^{\tau}(F+\eta)H(t)dD(t)\ge\int_{0}^{\tau}\left[x+\int_{0}^{t}\zeta(s)dB(s)\right]dD(t).
	\end{equation*}
	Since $D(t)$ is of bounded variation, we can implement the integration by parts and get
	\begin{equation*}
	\begin{split}
	\int_{0}^{\tau}D(s)H(s)&(c(s)+wl(s)-w\bar{L}+d)ds-\int_{0}^{\tau}D(s)\zeta(s)dB(s)\leq D(0)x+\int_{0}^{\tau}(F+\eta)H(s)dD(s)\\
	&\quad+ D(\tau)\left[\int_{0}^{\tau}H(s)(c(s)+wl(s)-w\bar{L}+d)ds-x-\int_{0}^{\tau}\zeta(s)dB(s)\right].
	\end{split}
	\end{equation*}
	Taking the expectation under the $\mathbb{P}$ measure on both sides and replacing Equation (\ref{Ap.17}), we obtain
	\begin{equation*}
	\mathbb{E}\!\left[\!\int_{0}^{\tau}\!\!\!D(s)H(s)(c(s)\!+\!wl(s)\!-\!w\bar{L}\!+\!d)ds\!\right]\!\leq\! D(0)x\!-\mathbb{E}\!\left[\!D(\tau)H(\tau)X^{x,c,\pi,l}\!(\tau)\right]\!+\mathbb{E}\left[\!\int_{0}^{\tau}\!\!\!(F+\eta)H(s)dD(s)\!\right]\!.
	\end{equation*} 
	The above inequality keeps true for any admissible control strategy $\{c(t),\pi(t),l(t)\}$ and any non-negative, non-increasing process $D(t)$. Furthermore, we will show that the inequality changes into equality with the given $c^{*}(t)$, $l^{*}(t)$ and $D^{*}(t)$. We first define a new process
	\begin{equation*}
	\bar{D}^{\epsilon}(t)\triangleq D^{*}(t)(1+\epsilon)\in\mathcal{D},
	\end{equation*}
	where $\epsilon$ is a small enough constant. Following the same argument in the first part proof, 
	we have $\tilde{\mathfrak{L}}(\bar{D}^{\epsilon}(t))\ge\tilde{\mathfrak{L}}(D^{*}(t))$, and
	\begin{equation*}
	\begin{split}
	\limsup_{\epsilon\downarrow 0}&\frac{\tilde{\mathfrak{L}}(\bar{D}^{\epsilon}(t))\!-\!\tilde{\mathfrak{L}}(D^{*}(t))}{\epsilon}\!=\!\limsup_{\epsilon\downarrow 0}\mathbb{E}\left[\int_{0}^{\tau}e^{-\gamma t}\frac{\tilde{u}(\lambda^{*} e^{\gamma t}D^{*}(t)(1\!+\!\epsilon)H(t))\!-\!\tilde{u}(\lambda^{*} e^{\gamma t}D^{*}(t)H(t))}{\epsilon}dt\right.\\
	&-\left.\int_{0}^{\tau}(d-w\bar{L})\lambda^{*} D^{*}(t) H(t)dt+e^{-\gamma \tau}\frac{\tilde{U}(\lambda^{*} e^{\gamma \tau}D^{*}(\tau)(1+\epsilon)H(\tau))-\tilde{U}(\lambda^{*} e^{\gamma\tau}D^{*}(\tau)H(\tau))}{\epsilon}\right]\\
	&+\lambda^{*}\mathbb{E}\left[\int_{0}^{\tau}(F+\eta)H(t)dD^{*}(t)\right]+\lambda^{*} x D^{*}(0)\ge 0,
	\end{split}
	\end{equation*}
	\begin{equation*}
	\begin{split}
	\liminf_{\epsilon\uparrow 0}&\frac{\tilde{\mathfrak{L}}(\bar{D}^{\epsilon}(t))\!-\!\tilde{\mathfrak{L}}(D^{*}(t))}{\epsilon}\!=\!\liminf_{\epsilon\uparrow 0}\mathbb{E}\left[\int_{0}^{\tau}e^{-\gamma t}\frac{\tilde{u}(\lambda^{*} e^{\gamma t}D^{*}(t)(1+\epsilon)H(t))-\tilde{u}(\lambda^{*} e^{\gamma t}D^{*}(t)H(t))}{\epsilon}dt\right.\\
	&-\left.\int_{0}^{\tau}(d-w\bar{L})\lambda^{*} D^{*}(t) H(t)dt+e^{-\gamma \tau}\frac{\tilde{U}(\lambda^{*} e^{\gamma \tau}D^{*}(\tau)(1+\epsilon)H(\tau))-\tilde{U}(\lambda^{*} e^{\gamma\tau}D^{*}(\tau)H(\tau))}{\epsilon}\right]\\
	&+\lambda^{*}\mathbb{E}\left[\int_{0}^{\tau}(F+\eta)H(t)dD^{*}(t)\right]+\lambda^{*} xD^{*}(0)\leq  0.
	\end{split}
	\end{equation*}
	Applying the Fatou's lemma, we obtain separately
	\begin{equation*}
	\begin{split}
	& \mathbb{E}\!\left[\!\int_{0}^{\tau}\!\!\!\!\!D^{*}\!(t)H\!(t)(\!c^{*}\!(t)\!+\!wl^{*}\!(t)\!-\!w\bar{L}\!+\!d)dt\!+\!\!D^{*}\!(\tau)H\!(\tau)X^{x\!,c^{*}\!,\pi^{*}\!,l^{*}}(\tau)\!\right]\!\!\leq\!\! xD^{*}\!(0)\!+\!\mathbb{E}\!\left[\!\int_{0}^{\tau}\!\!\!\!\!(F+\eta)H\!(t)dD^{*}\!(t)\right],\\
	&\mathbb{E}\!\left[\!\int_{0}^{\tau}\!\!\!\!\!D^{*}\!(t)H\!(t)(\!c^{*}\!(t)\!+\!wl^{*}\!(t)\!-\!w\bar{L}\!+\!d)dt\!+\!\!D^{*}\!(\tau)H\!(\tau)X^{x\!,c^{*}\!,\pi^{*}\!,l^{*}}(\tau)\!\right]\!\!\ge\!\! xD^{*}\!(0)\!+\!\mathbb{E}\!\left[\!\int_{0}^{\tau}\!\!\!\!\!(F+\eta)H\!(t)dD^{*}\!(t)\right],
	\end{split}
	\end{equation*}
	which give us
	\begin{equation*}
	\mathbb{E}\!\left[\!\int_{0}^{\tau}\!\!\!\!\!D^{*}\!(t)H\!(t)(\!c^{*}\!(t)\!+\!wl^{*}\!(t)\!-\!w\bar{L}\!+\!d)dt\!+\!\!D^{*}\!(\tau)H\!(\tau)X^{x\!,c^{*}\!,\pi^{*}\!,l^{*}}(\tau)\!\right]\!\!=\!\! xD^{*}\!(0)\!+\!\mathbb{E}\!\left[\!\int_{0}^{\tau}\!\!\!\!\!(F+\eta)H\!(t)dD^{*}\!(t)\right].
	\end{equation*}
	Moreover, we define a new optimization problem named $(P_{\tau}^{\prime})$ as
	\begin{equation*}
	\max_{c(t)\ge 0,l(t)\ge 0} \mathbb{E}\left[\int_{0}^{\tau}e^{-\gamma t}u(c(t),l(t))dt+e^{-\gamma \tau}U(X^{x,c,l}(\tau))\right], \tag{$P_{\tau}^{\prime}$}
	\end{equation*}
	subject to
	\begin{equation*} \mathbb{E}\!\left[\!\int_{0}^{\tau}\!\!\!\!\!D^{*}(t)H(t)(c(t)\!+\!wl(t)\!-\!w\bar{L}\!+\!d)dt+D^{*}(\tau)H(\tau)X^{x,c,l}(\tau)\!\right]\!\!\leq\! xD^{*}(0)\!+\!\mathbb{E}\!\left[\!\int_{0}^{\tau}\!\!\!\!\!(F\!+\!\eta)H(t)dD^{*}(t)\!\right] .
	\end{equation*}
	We denote the optimal solutions of the above problem as $\tilde{c}^{*}(t)$ and $\tilde{l}^{*}(t)$.The Lagrange method endows us  
	\begin{equation*}
	\tilde{c}^{*}(t)+w\tilde{l}^{*}(t)=-\tilde{u}^{\prime}(\tilde{\lambda}e^{\gamma t}D^{*}(t)H(t)),\quad X^{x,\tilde{c}^{*},\tilde{l}^{*}}(\tau)=-\tilde{U}^{\prime}(\tilde{\lambda}e^{\gamma \tau}D^{*}(\tau)H(\tau)),
	\end{equation*}
	where $\tilde{\lambda}>0$ is the Lagrange multiplier. The constraint of Problem $(P_{\tau}^{\prime})$ takes equality when $\tilde{\lambda}=\lambda^{*}$. Then, the condition $c^{*}(t)+wl^{*}(t)=-\tilde{u}^{\prime}(\lambda^{*} e^{\gamma t}D^{*}(t)H(t))$ 
	implies that $c^{*}(t)$ and $l^{*}(t)$ are the optimal control policies of the problem $(P_{\tau}^{\prime})$. Moreover, since the maximum utility of primal problem $(P_{\tau})$ is upper bounded by the maximum utility of $(P_{\tau}^{\prime})$, we can conclude that $c^{*}(t)$ and $l^{*}(t)$ are also the optimal consumption and leisure solution to the problem $(P_{\tau}).$
	
\subsection{Proof of Lemma \ref{Lemma 3}} \label{C.3}
Referring to \cite[Section 10.3, Example 10.3.1]{oksendal2013stochastic}, the notations here almost coincide with the ones there. Let us first introduce two new functions:
	\begin{equation*}
	g(y)\triangleq g(t,z)=e^{-\gamma t}\tilde{U}(z),\quad \mbox{and} \quad G(t,z,\bar{w})\triangleq g(t,z)+\bar{w}=e^{-\gamma t}\tilde{U}(z)+\bar{w}.
	\end{equation*}
	Defining an operator $\mathcal{A}_{P}G
	\triangleq\frac{\partial G}{\partial t}+(\gamma-r)z\frac{\partial G}{\partial z}+\frac{\theta^{2}}{2}z^{2}\frac{\partial^{2}G}{\partial z^{2}}+e^{-\gamma t}\tilde{u}(z)-e^{-\gamma t}(d-w\bar{L})z$, the continuous region of the corresponding optimal stopping time problem is expressed as 
	\begin{equation*}
	\Omega_{1}=\{(t,z,\bar{w}):\mathcal{A}_{P}G(t,z,\bar{w})>0\}.
	\end{equation*}
	Since 
	\begin{equation*}
	\mathcal{A}_{P}G(t,z,\bar{w})
	=-\gamma e^{-\gamma t}\tilde{U}(z)+(\gamma-r)z e^{-\gamma t}\tilde{U}^{\prime}(z)+\frac{\theta^{2}}{2}z^{2} e^{-\gamma t}\tilde{U}^{\prime\prime}(z)+e^{-\gamma t}\tilde{u}(z)-e^{-\gamma t}(d-w\bar{L})z,
	\end{equation*}
	defining a new function
	$h(z)
	\!=\!-\!\gamma \tilde{U}(z)\!+\!(\gamma\!-\!r)z \tilde{U}^{\prime}(z)\!+\!\frac{\theta^{2}}{2}z^{2}\tilde{U}^{\prime\prime}(z)\!+\!\tilde{u}(z)\!-\!(d\!-\!w\bar{L})z$,
	the continuous region can be rewritten as $	\Omega_{1}=\{z:h(z)>0\}.$ Because $\tilde{U}(z)$ takes the piecewise form of
	\begin{equation*}
	\tilde{U}(z)=\begin{cases}
	v_{\scriptscriptstyle PB}\left(\frac{z}{\alpha}\right)-Fz, & 0<z<\alpha \hat{z}_{\scriptscriptstyle PB},\\
	v_{\scriptscriptstyle PB}(\hat{z}_{\scriptscriptstyle PB})-Fz, & z\ge \alpha\hat{z}_{\scriptscriptstyle PB},
	\end{cases}
	\end{equation*}
	we can rewrite $h(z)$ as
	\begin{equation*}
	\begin{split}
	h(z)&=\begin{cases}
	-\gamma v_{\scriptscriptstyle PB}\left(\frac{z}{\alpha}\right)\!+\!(\gamma\!-\!r)\frac{z}{\alpha}v_{\scriptscriptstyle PB}^{\prime}\left(\frac{z}{\alpha}\right)\!+\!\frac{\theta^{2}}{2}\frac{z^{2}}{\alpha^{2}}v_{\scriptscriptstyle PB}^{\prime\prime}\left(\frac{z}{\alpha}\right)\!+\!rFz\!+\!\tilde{u}(z)\!-\!(d\!-\!w\bar{L})z,& 0\!<\!z\!<\!\alpha \hat{z}_{\scriptscriptstyle PB},\\
	-v_{\scriptscriptstyle PB}(\hat{z}_{\scriptscriptstyle PB})+\tilde{u}(z)+(rF-d+w\bar{L})z, & z\!\ge\! \alpha\hat{z}_{\scriptscriptstyle PB},
	\end{cases}\\
	&=\begin{cases}
	\tilde{u}(z)-\tilde{u}\left(\frac{z}{\alpha}\right)+(rF-d+w\bar{L}-\frac{w\bar{L}}{\alpha})z,& 0<z<\alpha \hat{z}_{\scriptscriptstyle PB},\\
	-v_{\scriptscriptstyle PB}(\hat{z}_{\scriptscriptstyle PB})+\tilde{u}(z)+(rF-d+w\bar{L})z, &z\ge \alpha\hat{z}_{\scriptscriptstyle PB},
	\end{cases}
	\end{split}
	\end{equation*}
	the last equality comes from Condition (\ref{Ap.18}) on the interval $0<z<\alpha \hat{z}_{\scriptscriptstyle PB}$.
	From Condition (\ref{1.13}), it can be observed that the function $h(z)$ only takes the form 
	\begin{equation*}
	h(z)=	\tilde{u}(z)-\tilde{u}\left(\frac{z}{\alpha}\right)+(rF-d+w\bar{L}-\frac{w\bar{L}}{\alpha})z,
	\end{equation*}
	on the considering interval $0<z<\hat{z}$.
	We first extend the domain of $h(z)$ to the whole positive real line, and discover the existence and uniqueness of its zero $\bar{z}$, then discuss the magnitude between $\bar{z}$ and $\hat{z}$. The decreasing property of $\tilde{u}(z)$ leads to $\tilde{u}(z)-\tilde{u}\left(\frac{z}{\alpha}\right)>0$; hence, a necessary condition to ensure that there is at least one zero point is put forward as
	$rF-d+w\bar{L}-\frac{w\bar{L}}{\alpha}<0.$
	Afterwards, for the sake of determining the curvature of $h(z)$, we take the second-order derivative and obtain
	\begin{equation*}
	h^{\prime\prime}(z)=
	\tilde{u}^{\prime\prime}(z)-\frac{1}{\alpha^{2}}\tilde{u}^{\prime\prime}\left(\frac{z}{\alpha}\right),
	\end{equation*}
	therefore we focus on the sign of function $\tilde{u}^{\prime\prime}(z)-\frac{1}{\alpha^{2}}\tilde{u}^{\prime\prime}\left(\frac{z}{\alpha}\right)$. Since $\tilde{u}(z)$ is strictly decreasing and convex on $(0,\infty)$ and adopts the piecewise form, we discover the sign of $\tilde{u}^{\prime\prime}(z)-\frac{1}{\alpha^{2}}\tilde{u}^{\prime\prime}\left(\frac{z}{\alpha}\right)$ on three different intervals: $z<\alpha\tilde{y}$, $\alpha\tilde{y}\leq z<\tilde{y}$ and $z\ge\tilde{y}$.
	\begin{itemize}
		\item On  $z<\alpha\tilde{y}$: from $\frac{\delta(1-k)}{1-\delta(1-k)}<0$, $0<\alpha <1$ and $A_{1}=\frac{1-\delta+\delta k}{\delta (1-k)}L^{-\frac{(1-k)(1-\delta)}{\delta(1-k)-1}}<0$, we get
		\begin{equation*}
		\tilde{u}^{\prime\prime}(z)-\frac{1}{\alpha^{2}}\tilde{u}^{\prime\prime}\left(\frac{z}{\alpha}\right)=\frac{\delta(1-k)}{(\delta(1-k)-1)^{2}}A_{1}z^{\frac{2-\delta(1-k)}{\delta(1-k)-1}}\left(1-\alpha^{\frac{\delta(1-k)}{1-\delta(1-k)}}\right)<0.
		\end{equation*}
		\item On $\alpha \tilde{y}\leq z<\tilde{y}$: the function $\tilde{u}^{\prime\prime}(z)-\frac{1}{\alpha^{2}}\tilde{u}^{\prime\prime}\left(\frac{z}{\alpha}\right)$ is rewritten as
		\begin{equation*}
		\tilde{u}^{\prime\prime}(z)-\frac{1}{\alpha^{2}}\tilde{u}^{\prime\prime}\left(\frac{z}{\alpha}\right)=\frac{1}{z^{2}}\left[\frac{\delta(1-k)}{(\delta(1-k)-1)^{2}}A_{1}z^{\frac{\delta(1-k)}{\delta(1-k)-1}}-\frac{1-k}{k^{2}}A_{2}z^{-\frac{1-k}{k}}\alpha^{\frac{1-k}{k}}\right].
		\end{equation*}
		The first term inside the square bracket has the following inequality relationship
		\begin{equation*}
		\frac{\delta(1\!-\!k)}{(\delta(1\!-\!k)\!-\!1)^{2}}A_{1}z^{\frac{\delta(1\!-\!k)}{\delta(1\!-\!k)\!-\!1}}\!<\!\frac{\delta(1\!-\!k)}{(\delta(1\!-\!k)\!-\!1)^{2}}A_{1}\tilde{y}^{\frac{\delta(1\!-\!k)}{\delta(1\!-\!k)\!-\!1}}\!=\!\frac{1}{1\!-\!\delta(1\!-\!k)}L^{1\!-\!k}\left(\frac{1\!-\!\delta}{\delta w}\right)^{\!-\!\delta(1\!-\!k)},
		\end{equation*}
		as for the second term, we have
		\begin{equation*}
		-\frac{1-k}{k^{2}}A_{2}z^{-\frac{1-k}{k}}\alpha^{\frac{1-k}{k}}\leq-\frac{1-k}{k^{2}}A_{2}(\alpha\tilde{y})^{-\frac{1-k}{k}}\alpha^{\frac{1-k}{k}}=-\frac{L^{1-k}}{\delta}\left(\frac{1-\delta}{\delta w}\right)^{-\delta(1-k)}.
		\end{equation*}
		Then we can determine the sign of $\tilde{u}^{\prime\prime}(z)-\frac{1}{\alpha^{2}}\tilde{u}^{\prime\prime}\left(\frac{z}{\alpha}\right)$ as
		\begin{equation*}
		\begin{split}
		\tilde{u}^{\prime\prime}(z)-\frac{1}{\alpha^{2}}\tilde{u}^{\prime\prime}\left(\frac{z}{\alpha}\right)&<\frac{1}{z^{2}}\bigg[\frac{1}{1-\delta(1-k)}L^{1-k}\left(\frac{1-\delta}{\delta w}\right)^{-\delta(1-k)}-\frac{L^{1-k}}{\delta}\left(\frac{1-\delta}{\delta w}\right)^{-\delta(1-k)}\bigg]\\
		&=\frac{1}{z^{2}}\bigg[\left(\frac{1-\delta}{\delta w}\right)^{-\delta(1-k)}L^{1-k}\left(\frac{\delta-1+\delta(1-k)}{(1-\delta(1-k))\delta}\right)\bigg]<0.
		\end{split}
		\end{equation*}
		\item On $z\ge\tilde{y}$: The conditions $\frac{1-k}{k}\!<\!0$, $0\!<\!\alpha\!<\!1$ and $A_{2}\!=\!\frac{k}{\delta(1\!-\!k)}\left(\frac{1\!-\!\delta}{\delta w}\right)^{\frac{(1\!-\!k)(1\!-\!\delta)}{k}}\!<\!0$ endows us
		\begin{equation*}
		\tilde{u}^{\prime\prime}(z)-\frac{1}{\alpha^{2}}\tilde{u}^{\prime\prime}\left(\frac{z}{\alpha}\right)=\frac{1-k}{k^{2}}A_{2}z^{-\frac{1}{k}-1}\left(1-\alpha^{\frac{1-k}{k}}\right)<0.
		\end{equation*} 
	\end{itemize}
	Hence, we can summarize that $h^{\prime\prime}(z)<0$, which means the function $h(z)$ is strictly concave for $z>0$. Additionally, combining with the subsequent facts
	\begin{equation*}
	\lim\limits_{z\downarrow 0}h(z)=\lim\limits_{z\downarrow 0}A_{1}z^{\frac{\delta(1-k)}{\delta(1-k)-1}}\left(1-\alpha^{\frac{\delta(1-k)}{1-\delta(1-k)}}\right)+wLz\left(\frac{1}{\alpha}-1\right)+\left(rF-d+w\bar{L}-\frac{w\bar{L}}{\alpha}\right)z=0,
	\end{equation*}
	\begin{equation*}
	\lim\limits_{z\downarrow 0}h^{\prime}(z)=\lim\limits_{z\downarrow 0}\frac{\delta(1-k)}{\delta(1-k)-1}A_{1}z^{\frac{1}{\delta(1-k)-1}}\left(1-\alpha^{\frac{\delta(1-k)}{1-\delta(1-k)}}\right)+rF-d+w(\bar{L}-L)\left(1-\frac{1}{\alpha}\right)=\infty,
	\end{equation*}
	\begin{equation*}
	\lim\limits_{z\to\infty}h(z)=\lim\limits_{z\to\infty}A_{2}z^{-\frac{1-k}{k}}\left(1-\alpha^{\frac{1-k}{k}}\right)+\left(rF-d+w\bar{L}-\frac{w\bar{L}}{\alpha}\right)z=-\infty,
	\end{equation*}
	we can conclude that there exists a unique $\bar{z}$ such that $\{z>0:h(z)\!>\!0\}\!=\!\{0\!<\!z\!<\!\bar{z}\}$. Moreover, solutions of the primal optimization problem are discussed in two different cases, $\bar{z}\!<\!\hat{z}$ and $\bar{z}\!\ge\!\hat{z}$.
	
\subsection{Proof of Lemma \ref{Lemma 4}} \label{C.4}
Referring to \cite[Appendix, A.2]{jeanblanc2004optimal}, we make use of the scale function to calculate the probability of the stopping time. Related to the diffusion process of $Z(t)$, the scale function is
	\begin{equation*}
	s(z)=\int_{a}^{z}e^{-2\int_{a}^{y}\frac{f(x)}{g^{2}(x)}dx}dy,
	\end{equation*}
	where $f(x)=(\gamma-r)x$, $g(x)=-\theta x$ are the drift and diffusion coefficients of $dZ(t)$ and $a$ is an arbitrary constant located in $(0,\bar{z})$. Then we obtain
	\begin{equation*}
	\mathbb{P}(\tau_{\bar{z}}<\tau_{0})=\frac{s(z)-s(0)}{s(\bar{z})-s(0)},
	\end{equation*}
	with $s(0)=\lim\limits_{b\to 0^{+}}\int_{a}^{b}e^{-2\int_{a}^{y}\frac{\gamma-r}{\theta^{2}x}dx}dy$.
	It follows that in order to prove that the above probability is one instead of depending on the initial value $z$, it suffices to show that $s(0)=-\infty$. Since
	\begin{equation*}
	s(0)=\lim_{b\to 0^{+}}\int_{a}^{b}e^{-2\int_{a}^{y}\frac{\gamma-r}{\theta^{2}x}dx}dy=\lim_{b\to 0^{+}}\left(\frac{a^{\frac{2(\gamma-r)}{\theta^{2}}}}{1-\frac{2(\gamma-r)}{\theta^{2}}}b^{1-\frac{2(\gamma-r)}{\theta^{2}}}-\frac{a}{1-\frac{2(\gamma-r)}{\theta^{2}}}\right),
	\end{equation*}
	a sufficient condition to make $s(0)=-\infty$ is $1-\frac{2(\gamma-r)}{\theta^{2}}<0$, which is equal to $\gamma>r+\frac{\theta^{2}}{2}.$ Moreover, following the same argument, it can be easily proven that $\mathbb{P}\left(\tau_{0}<\infty\right)=1$ with the condition (\ref{1.15}), which also gives us $\mathbb{P}\left(\tau_{\bar{z}}<\infty\right)=1$.
	
\subsection{Solutions of Variational Inequalities (\ref{1.16}) and (\ref{1.17})} \label{C.5}

\paragraph{\textbf{Case 1. $0<\bar{z}<\hat{z}<\alpha\tilde{y}\leq\alpha\hat{z}_{\scriptscriptstyle PB}$:}}
\noindent The differential equation generated from Condition $(V3)$ in the system (\ref{1.16}) takes the solution as
\begin{equation*}
v(z)=B_{1}z^{n_{1}}+B_{2}z^{n_{2}}+\frac{A_{1}}{\Gamma_{1}}z^{\frac{\delta(1-k)}{\delta(1-k)-1}}+\frac{w(\bar{L}-L)-d}{r}z, \qquad 0<z<\bar{z}.
\end{equation*}
Since $n_{1}<0$, the term $z^{n_{1}}$ will suffer the explosion as $z$ goes to 0. Therefore, we set the coefficient $B_{1}=0$ to meet the boundedness assumption of $v(z)$.
Forward, the smooth fit condition $v(\bar{z})=\tilde{U}(\bar{z})$ leads to the subsequent two-equation system to resolve the parameters $B_{2}$ and $\bar{z}$.
\begin{itemize}
	\item  $\mathcal{C}^{0}$ condition at $z=\bar{z}$
	\begin{equation}
	B_{2}\bar{z}^{n_{2}}\!+\!\frac{A_{1}}{\Gamma_{1}}\bar{z}^{\frac{\delta(1\!-\!k)}{\delta(1\!-\!k)\!-\!1}}\!+\!\frac{w(\bar{L}\!-\!L)\!-\!d}{r}\bar{z}\!=\!B_{21,\scriptscriptstyle PB}\left(\frac{\bar{z}}{\alpha}\right)^{n_{2}}\!+\!\frac{A_{1}}{\Gamma_{1}}\left(\frac{\bar{z}}{\alpha}\right)^{\frac{\delta(1\!-\!k)}{\delta(1\!-\!k)\!-\!1}}\!+\!\left(\frac{w(\bar{L}\!-\!L)}{r\alpha}\!-\!F\right)\bar{z}.\label{Ap.19}
	\end{equation}
	\item $\mathcal{C}^{1}$ condition at $z=\bar{z}$
	\begin{equation}
	\begin{split}
	n_{2}B_{2}\bar{z}^{n_{2}-1}+&\frac{\delta(1-k)}{\delta(1-k)-1}\frac{A_{1}}{\Gamma_{1}}\bar{z}^{\frac{1}{\delta(1-k)-1}}+\frac{w(\bar{L}-L)-d}{r}\\
	&=n_{2}\frac{B_{21,\scriptscriptstyle PB}}{\alpha}\left(\frac{\bar{z}}{\alpha}\right)^{n_{2}-1}+\frac{\delta(1-k)}{\delta(1-k)-1}\frac{A_{1}}{\Gamma_{1}\alpha}\left(\frac{\bar{z}}{\alpha}\right)^{\frac{1}{\delta(1-k)-1}}+\left(\frac{w(\bar{L}-L)}{r\alpha}-F\right).
	\end{split}\label{Ap.20}
	\end{equation}
\end{itemize}
Besides, the condition $\tilde{U}^{'}(\hat{z})=-(F+\eta)$ generates the following equation
\begin{equation*}
\tilde{U}(\hat{z})=n_{2}\frac{B_{21,\scriptscriptstyle PB}}{\alpha}\left(\frac{\hat{z}}{\alpha}\right)^{n_{2}-1}+\frac{\delta(1-k)}{\delta(1-k)-1}\frac{A_{1}}{\Gamma_{1}\alpha}\left(\frac{\hat{z}}{\alpha}\right)^{\frac{1}{\delta(1-k)-1}}+\frac{w(\bar{L}-L)}{r\alpha}=-\eta,
\end{equation*} 
which endows us the value of $\hat{z}$.

\paragraph{\textbf{Case 2. $0<\bar{z}<\alpha\tilde{y}\leq\hat{z}\leq\alpha\hat{z}_{\scriptscriptstyle PB}$:}}
\noindent The Condition $(V3)$ from (\ref{1.16}) leads to the following differential equation
\begin{equation*}
-\gamma v(z)+(\gamma-r)z v^{\prime}(z)+\frac{1}{2}\theta^{2}z^{2}v^{\prime\prime}(z)+\tilde{u}(z)-(d-w\bar{L})z=0, \quad 0<z<\bar{z},
\end{equation*}
since the condition $\bar{z}<\tilde{y}$ remains unchanged, the above equation keeps the same compared with the previous sections, hence, takes the identical solution as
\begin{equation*}
v(z)=B_{2}z^{n_{2}}+\frac{A_{1}}{\Gamma_{1}}z^{\frac{\delta(1-k)}{\delta(1-k)-1}}+\frac{w(\bar{L}-L)-d}{r}z, \qquad 0<z<\bar{z}.
\end{equation*}
Furthermore, because the condition $\bar{z}<\alpha\tilde{y}$ also coincides with Case 1, we have the same smooth fit condition $v(\bar{z})=\tilde{U}(\bar{z})$, which results in the identical values of parameters $B_{2}$ and $\bar{z}$. The difference from the previous case occurs at the boundary $\hat{z}$. Because of $\alpha\tilde{y}\leq\hat{z}\leq\alpha\hat{z}_{\scriptscriptstyle PB}$, the function $\tilde{U}(z)$ adopts a different form at the point $\hat{z}$; furthermore, 
$(V1)$ and $(V6)$ in the system (\ref{1.16}) produces a different equation, that is,
\begin{equation}
n_{1}\frac{B_{12,\scriptscriptstyle PB}}{\alpha}\left(\frac{\hat{z}}{\alpha}\right)^{n_{1}-1}+n_{2}\frac{B_{22,\scriptscriptstyle PB}}{\alpha}\left(\frac{\hat{z}}{\alpha}\right)^{n_{2}-1}-\frac{1-k}{k}\frac{A_{2}}{\Gamma_{2}\alpha}\left(\frac{\hat{z}}{\alpha}\right)^{-\frac{1}{k}}+\frac{w\bar{L}}{r\alpha}=-\eta,\label{Ap.21}
\end{equation}
to acquire the value of $\hat{z}$.

\paragraph{\textbf{Case 3. $0<\alpha\tilde{y}<\bar{z}<\hat{z}\leq\alpha\hat{z}_{\scriptscriptstyle PB}$ $\&$ $\bar{z}<\tilde{y}$:}}
\noindent Under the same condition $\bar{z}<\tilde{y}$ with the previous cases, the differential equation generating from $(V3)$ of (\ref{1.16}) follows the identical solution,
\begin{equation*}
v(z)=B_{2}z^{n_{2}}+\frac{A_{1}}{\Gamma_{1}}z^{\frac{\delta(1-k)}{\delta(1-k)-1}}+\frac{w(\bar{L}-L)-d}{r}z, \qquad 0<z<\bar{z}.
\end{equation*}
Then the smooth fit condition $v(\bar{z})=\tilde{U}(\bar{z})$ enables us to determine the values of parameters $B_{2}$ and $\bar{z}$ with the following two-equation system:
\begin{itemize}
	\item  $\mathcal{C}^{0}$ condition at $z=\bar{z}$
	\begin{equation*}
	B_{2}\bar{z}^{n_{2}}\!+\!\frac{A_{1}}{\Gamma_{1}}\bar{z}^{\frac{\delta(1\!-\!k)}{\delta(1\!-\!k)\!-\!1}}\!+\!\frac{w(\bar{L}\!-\!L)\!-d}{r}\bar{z}\!=\!B_{12,\scriptscriptstyle PB}\left(\!\frac{\bar{z}}{\alpha}\!\right)^{n_{1}}\!+\!B_{22,\scriptscriptstyle PB}\left(\!\frac{\bar{z}}{\alpha}\!\right)^{n_{2}}\!+\!\frac{A_{2}}{\Gamma_{2}}\!\left(\!\frac{\bar{z}}{\alpha}\!\right)^{\!-\!\frac{1\!-\!k}{k}}\!+\!\left(\!\frac{w\bar{L}}{r\alpha}\!-\!F\!\right)\bar{z};
	\end{equation*}
	\item $\mathcal{C}^{1}$ condition at $z=\bar{z}$
	\begin{equation*}
	\begin{split}
	n_{2}B_{2}\bar{z}^{n_{2}\!-\!1}\!\!+\!\frac{\delta(\!1\!-\!k\!)}{\delta(\!1\!-\!k\!)\!-\!1}\!&\frac{A_{1}}{\Gamma_{1}}\bar{z}^{\frac{1}{\delta(1\!-\!k)\!-\!1}}\!\!+\!\frac{w(\!\bar{L}\!-\!L\!)\!-\!d}{r}\!=\!\\
	&n_{1}\!\frac{B_{12,\scriptscriptstyle PB}}{\alpha}\!\left(\!\frac{\bar{z}}{\alpha}\!\right)^{n_{1}\!-\!1}\!\!+\!n_{2}\!\frac{B_{22,\scriptscriptstyle PB}}{\alpha}\!\left(\!\frac{\bar{z}}{\alpha}\!\right)^{n_{2}\!-\!1}\!\!-\!\frac{1\!-\!}{k}\!\frac{A_{2}}{\Gamma_{2}\alpha}\!\left(\!\frac{\bar{z}}{\alpha}\!\right)^{\!-\!\frac{1}{k}}\!+\!\left(\!\frac{w\bar{L}}{r\alpha}\!-\!F\!\right).
	\end{split}
	\end{equation*}
\end{itemize}
Besides, from Condition $(V1)$ and $(V6)$ of (\ref{1.16}), we have the smooth fit condition $\tilde{U}^{\prime}(\hat{z})=-(F+\eta)$. Combining with the prerequisite $\alpha\tilde{y}\leq\hat{z}\leq\alpha\hat{z}_{\scriptscriptstyle PB}$, we have the subsequent equation,
\begin{equation*}
n_{1}\frac{B_{12,\scriptscriptstyle PB}}{\alpha}\left(\frac{\hat{z}}{\alpha}\right)^{n_{1}-1}+n_{2}\frac{B_{22,\scriptscriptstyle PB}}{\alpha}\left(\frac{\hat{z}}{\alpha}\right)^{n_{2}-1}-\frac{1-k}{k}\frac{A_{2}}{\Gamma_{2}\alpha}\left(\frac{\hat{z}}{\alpha}\right)^{-\frac{1}{k}}+\frac{w\bar{L}}{r\alpha}=-\eta,
\end{equation*}
which gives us the exact value of parameter $\hat{z}$. 

\paragraph{\textbf{Case 4. $0<\bar{z}<\hat{z}\leq\alpha\hat{z}_{\scriptscriptstyle PB}<\alpha\tilde{y}<\tilde{y}$:}}
\noindent Same with the previous cases, Condition $(V3)$ of (\ref{1.16}) forces the function $v(z)$ to take the solution as
\begin{equation*}
v(z)=B_{2}z^{n_{2}}+\frac{A_{1}}{\Gamma_{1}}z^{\frac{\delta(1-k)}{\delta(1-k)-1}}+\frac{w(\bar{L}-L)-d}{r}z, \qquad 0<z<\bar{z}.
\end{equation*}
Then, considering the smooth fit condition $v(\bar{z})=\tilde{U}(\bar{z})$, a two-equation system is established to resolve the parameters $B_{2}$ and $\bar{z}$:
\begin{itemize}
	\item  $\mathcal{C}^{0}$ condition at $z=\bar{z}$
	\begin{equation*}
	B_{2}\bar{z}^{n_{2}}\!+\!\frac{A_{1}}{\Gamma_{1}}\bar{z}^{\frac{\delta(1-k)}{\delta(1-k)-1}}+\frac{w(\bar{L}\!-\!L)\!-\!d}{r}\bar{z}\!=\!B_{2,\scriptscriptstyle PB}\left(\frac{\bar{z}}{\alpha}\right)^{n_{2}}\!+\!\frac{A_{1}}{\Gamma_{1}}\left(\frac{\bar{z}}{\alpha}\right)^{\frac{\delta(1-k)}{\delta(1-k)-1}}\!+\!\left(\frac{w(\bar{L}-L)}{r\alpha}\!-\!F\right)\bar{z};
	\end{equation*}
	\item $\mathcal{C}^{1}$ condition at $z=\bar{z}$
	\begin{equation*}
	\begin{split}
	n_{2}B_{2}\bar{z}^{n_{2}-1}+&\frac{\delta(1-k)}{\delta(1-k)-1}\frac{A_{1}}{\Gamma_{1}}\bar{z}^{\frac{1}{\delta(1-k)-1}}+\frac{w(\bar{L}-L)-d}{r}\\
	&=n_{2}\frac{B_{2,\scriptscriptstyle PB}}{\alpha}\left(\frac{\bar{z}}{\alpha}\right)^{n_{2}-1}+\frac{\delta(1-k)}{\delta(1-k)-1}\frac{A_{1}}{\Gamma_{1}\alpha}\left(\frac{\bar{z}}{\alpha}\right)^{\frac{1}{\delta(1-k)-1}}+\left(\frac{w(\bar{L}-L)}{r\alpha}-F\right).
	\end{split}
	\end{equation*}
\end{itemize}
Combining the condition $\tilde{U}^{\prime}(\hat{z})=-(F+\eta)$ with the prerequisite $\hat{z}\leq\hat{z}_{\scriptscriptstyle PB}$, we get the following equation
\begin{equation*}
n_{2}\frac{B_{2,\scriptscriptstyle PB}}{\alpha}\left(\frac{\hat{z}}{\alpha}\right)^{n_{2}-1}+\frac{\delta(1-k)}{\delta(1-k)-1}\frac{A_{1}}{\Gamma_{1}\alpha}\left(\frac{\hat{z}}{\alpha}\right)^{\frac{1}{\delta(1-k)-1}}+\frac{w(\bar{L}-L)}{r\alpha}=-\eta,
\end{equation*}
which enables us to obtain the exact value of $\hat{z}$.

\paragraph{\textbf{Case 5. $0<\alpha\tilde{y}<\tilde{y}\leq\bar{z}<\hat{z}<\alpha\hat{z}_{\scriptscriptstyle PB}$:}}
\noindent We begin with Condition $(V3)$ of (\ref{1.16}), which leads to the subsequent differential equation on the interval $0<z<\bar{z}$,
\begin{equation*}
-\gamma v(z)+(\gamma-r)z v^{\prime}(z)+\frac{1}{2}\theta^{2}z^{2}v^{\prime\prime}(z)+\tilde{u}(z)-(d-w\bar{L})z=0.
\end{equation*} 
Since $\bar{z}>\tilde{y}$, Lemma \ref{Lemma 1} shows that the function $\tilde{u}(z)$ takes two different forms on the corresponding interval with $\tilde{y}$ as the separating threshold. Hence, the solution of the above differential equation inherits this piecewise property and has the following resolution,
\begin{equation*}
v(z)=\begin{cases}
B_{11}z^{n_{1}}+B_{21}z^{n_{2}}+\frac{A_{1}}{\Gamma_{1}}z^{\frac{\delta(1-k)}{\delta(1-k)-1}}+\frac{w(\bar{L}-L)-d}{r}z, & 0<z<\tilde{y},\\

B_{12}z^{n_{1}}+B_{22}z^{n_{2}}+\frac{A_{2}}{\Gamma_{2}}z^{-\frac{1-k}{k}}+\frac{w\bar{L}-d}{r}z,& \tilde{y}\leq z<\bar{z}.
\end{cases}
\end{equation*}
To meet the boundedness assumption of $v(z)$, we set $B_{11}=0$ to avoid the explosion of term $z^{n_{1}}$ as $z$ goes to 0. Then, the same argument with the previous cases, we need to construct a four-equation system, which resorts to the smooth conditions at the point $\tilde{y}$ and $\bar{z}$, to determine the boundary $\bar{z}$, and the coefficient of function $v(z)$, i.e., $B_{21}$, $B_{12}$, $B_{22}$:
\begin{itemize}
	\item $\mathcal{C}^{0}$ condition at $z=\bar{z}$
	\begin{equation*}
	B_{12}\bar{z}^{n_{1}}\!+\!B_{22}\bar{z}^{n_{2}}\!+\!\frac{A_{2}}{\Gamma_{2}}\bar{z}^{-\frac{1\!-\!k}{k}}\!+\!\frac{w\bar{L}\!-\!d}{r}\bar{z}\!=\!B_{12,\scriptscriptstyle PB}\left(\!\frac{\bar{z}}{\alpha}\!\right)^{n_{1}}\!\!+\!B_{22,\scriptscriptstyle PB}\left(\!\frac{\bar{z}}{\alpha}\!\right)^{n_{2}}\!\!+\!\frac{A_{2}}{\Gamma_{2}}\left(\!\frac{\bar{z}}{\alpha}\!\right)^{\!-\!\frac{1\!-\!k}{k}}\!\!+\!\left(\!\frac{w\bar{L}}{r\alpha}\!-\!F\!\right)\bar{z};
	\end{equation*}
	\item $\mathcal{C}^{1}$ condition at $z=\bar{z}$
	\begin{equation*}
	\begin{split}
	n_{1}B_{12}\bar{z}^{n_{1}-1}+n_{2}&B_{22}\bar{z}^{n_{2}-1}-\frac{1-k}{k}\frac{A_{2}}{\Gamma_{2}}\bar{z}^{-\frac{1}{k}}+\frac{w\bar{L}-d}{r}=\\
	&n_{1}\frac{B_{12,\scriptscriptstyle PB}}{\alpha}\left(\frac{\bar{z}}{\alpha}\right)^{n_{1}\!-\!1}\!+\!n_{2}\frac{B_{22,\scriptscriptstyle PB}}{\alpha}\left(\frac{\bar{z}}{\alpha}\right)^{n_{2}\!-\!1}\!-\!\frac{1\!-\!k}{k}\frac{A_{2}}{\Gamma_{2}\alpha}\left(\frac{\bar{z}}{\alpha}\right)^{\!-\!\frac{1}{k}}\!+\!\frac{w\bar{L}}{r\alpha}\!-\!F;
	\end{split}
	\end{equation*}
	\item $\mathcal{C}^{0}$ condition at $z=\tilde{y}$
	\begin{equation*}
	B_{21}\tilde{y}^{n_{2}}+\frac{A_{1}}{\Gamma_{1}}\tilde{y}^{\frac{\delta(1-k)}{\delta(1-k)-1}}-\frac{wL}{r}\tilde{y}=B_{12}\tilde{y}^{n_{1}}+B_{22}\tilde{y}^{n_{2}}+\frac{A_{2}}{\Gamma_{2}}\tilde{y}^{-\frac{1-k}{k}};
	\end{equation*}
	\item $\mathcal{C}^{1}$ condition at $z=\tilde{y}$
	\begin{equation*}
	n_{2}B_{21}\tilde{y}^{n_{2}-1}+\frac{\delta(1-k)}{\delta(1-k)-1}\frac{A_{1}}{\Gamma_{1}}\tilde{y}^{\frac{1}{\delta(1-k)-1}}-\frac{wL}{r}= n_{1}B_{12}\tilde{y}^{n_{1}-1}+n_{2}B_{22}\tilde{y}^{n_{2}-1}-\frac{1-k}{k}\frac{A_{2}}{\Gamma_{2}}\tilde{y}^{-\frac{1}{k}}.
	\end{equation*}	
\end{itemize}
Meanwhile, considering the prerequisite $\alpha\tilde{y}\!<\!\hat{z}\!<\!\alpha\hat{z}_{\scriptscriptstyle PB}$ and the smooth fit condition $\tilde{U}^{'}(\hat{z})\!=\!-(F\!+\!\eta)$, we have
\begin{equation*}
n_{1}\frac{B_{12,\scriptscriptstyle PB}}{\alpha}\left(\frac{\hat{z}}{\alpha}\right)^{n_{1}-1}+n_{2}\frac{B_{22,\scriptscriptstyle PB}}{\alpha}\left(\frac{\hat{z}}{\alpha}\right)^{n_{2}-1}-\frac{1-k}{k}\frac{A_{2}}{\Gamma_{2}\alpha}\left(\frac{\hat{z}}{\alpha}\right)^{-\frac{1}{k}}+\frac{w\bar{L}}{r\alpha}=-\eta,
\end{equation*}
which gives us the value of $\hat{z}$.

\paragraph{\textbf{Case 6. $0<\tilde{y}\leq\hat{z}$ $\&$ $\bar{z}\ge\hat{z}$:}}
The Condition $(V3)$ of (\ref{1.17}) endows us a partial differential equation with the solution
\begin{equation*}
v(z)=\begin{cases}
B_{11}z^{n_{1}}+B_{21}z^{n_{2}}+\frac{A_{1}}{\Gamma_{1}}z^{\frac{\delta(1-k)}{\delta(1-k)-1}}+\frac{w(\bar{L}-L)-d}{r}z, & 0<z<\tilde{y},\\

B_{12}z^{n_{1}}+B_{22}z^{n_{2}}+\frac{A_{2}}{\Gamma_{2}}z^{-\frac{1-k}{k}}+\frac{w\bar{L}-d}{r}z,& \tilde{y}\leq z< \hat{z}.
\end{cases}
\end{equation*}
Due to the same considerations as before, we set $B_{11}=0$ to meet the boundedness of $v(z)$. Then, with the smooth fit conditions at $\hat{z}$ and $\tilde{y}$, we can construct a four-equation system to determine the values of parameters $B_{21}$, $B_{12}$, $B_{22}$ and $\hat{z}$:
\begin{itemize}
	\item $\mathcal{C}^{0}$ condition at $z=\tilde{y}$
	\begin{equation*}
	B_{21}\tilde{y}^{n_{2}}+\frac{A_{1}}{\Gamma_{1}}\tilde{y}^{\frac{\delta(1-k)}{\delta(1-k)-1}}-\frac{wL}{r}\tilde{y}=B_{12}\tilde{y}^{n_{1}}+B_{22}\tilde{y}^{n_{2}}+\frac{A_{2}}{\Gamma_{2}}\tilde{y}^{-\frac{1-k}{k}};
	\end{equation*}
	\item $\mathcal{C}^{1}$ condition at $z=\tilde{y}$
	\begin{equation*}
	n_{2}B_{21}\tilde{y}^{n_{2}-1}+\frac{\delta(1-k)}{\delta(1-k)-1}\frac{A_{1}}{\Gamma_{1}}\tilde{y}^{\frac{1}{\delta(1-k)-1}}-\frac{wL}{r}= n_{1}B_{12}\tilde{y}^{n_{1}-1}+n_{2}B_{22}\tilde{y}^{n_{2}-1}-\frac{1-k}{k}\frac{A_{2}}{\Gamma_{2}}\tilde{y}^{-\frac{1}{k}};
	\end{equation*}	
	\item $\mathcal{C}^{1}$ condition at $z=\hat{z}$
	\begin{equation*}
	n_{1}B_{12}\hat{z}^{n_{1}-1}+n_{2}B_{22}\hat{z}^{n_{2}-1}-\frac{1-k}{k}\frac{A_{2}}{\Gamma_{2}}\hat{z}^{-\frac{1}{k}}+\frac{w\bar{L}-d}{r}+F+\eta=0;
	\end{equation*}
	\item $\mathcal{C}^{2}$ condition at $z=\hat{z}$
	\begin{equation*}
	n_{1}(n_{1}-1)B_{12}\hat{z}^{n_{1}-2}+n_{2}(n_{2}-1)B_{22}\hat{z}^{n_{2}-2}+\frac{1-k}{k^{2}}\frac{A_{2}}{\Gamma_{2}}\hat{z}^{-\frac{1+k}{k}}=0.
	\end{equation*}
\end{itemize}

\paragraph{\textbf{Case 7. $0<\hat{z}<\tilde{y}$ $\&$ $\bar{z}\ge\hat{z}$:}}
\noindent The prerequisite $\hat{z}<\tilde{y}$ makes the function $\tilde{u}(z)$ of the form $A_{1}z^{\frac{\delta(1-k)}{\delta(1-k)-1}}-w L z$. Then Condition $(V3)$ in (\ref{1.17}) has the solution
\begin{equation*}
v(z)=
B_{1}z^{n_{1}}+B_{2}z^{n_{2}}+\frac{A_{1}}{\Gamma_{1}}z^{\frac{\delta(1-k)}{\delta(1-k)-1}}+\frac{w(\bar{L}-L)-d}{r}z, \quad 0<z<\hat{z}.
\end{equation*}
We set $B_{1}=0$ for avoiding the explosion of the term $B_{1}z^{n_{1}}$ when $z$ goes to 0. Then a two-equation system is created to solve the value of parameters $B_{2}$ and $\hat{z}$ explicitly:
\begin{itemize}
	\item $\mathcal{C}^{1}$ condition at $z=\hat{z}$
	\begin{equation*}
	n_{2}B_{2}\hat{z}^{n_{2}-1}+\frac{\delta(1-k)}{\delta(1-k)-1}\frac{A_{1}}{\Gamma_{1}}\hat{z}^{\frac{1}{\delta(1-k)-1}}+\frac{w(\bar{L}-L)-d}{r}+F+\eta=0;
	\end{equation*}
	\item $\mathcal{C}^{2}$ condition at $z=\hat{z}$
	\begin{equation*}
	n_{2}(n_{2}-1)B_{2}\hat{z}^{n_{2}-2}+\frac{\delta(1-k)}{(\delta(1-k)-1)^{2}}\frac{A_{1}}{\Gamma_{1}}\hat{z}^{\frac{2-\delta(1-k)}{\delta(1-k)-1}}=0.
	\end{equation*}
\end{itemize}

\section{No Optimal Bankruptcy Problem}\label{nob}
To study the influence of introducing the bankruptcy option, we also solve a pure optimal control problem without optimal stopping, which is defined subsequently as
\begin{equation*}
V_{\scriptscriptstyle nob}(x)=\sup_{\{c_{\scriptscriptstyle nob}(t),\pi_{\scriptscriptstyle nob}(t),l_{\scriptscriptstyle nob}(t)\}_{\scriptscriptstyle t\ge 0}\in\mathcal{A}_{\scriptscriptstyle nob}(x)}\mathbb{E}\left[\int_{0}^{\infty}e^{-\gamma t}u(c_{\scriptscriptstyle nob}(t),l_{\scriptscriptstyle nob}(t))dt\right].
\end{equation*}
The subscript $\scriptstyle nob$ indicates that the considering variables and functions are concerned with the no bankruptcy option model. Moreover, the admissible control set $\mathcal{A}_{\scriptscriptstyle nob}(x)$ is almost consistent with the definition of $\mathcal{A}_{\scriptscriptstyle PB}(x)$ except the liquidity condition. $\mathcal{A}_{\scriptscriptstyle nob}(x)$ adopts ``$X_{\scriptscriptstyle nob}(t)\ge F+\eta$, a.s., $\forall t\!\ge\! 0$'' instead of ``$X(t)\!\ge \!0$, a.s., $\forall t\!\ge\! 0$''.
\noindent Then we provide the budget and liquidity constraints as:
\begin{equation*}
\begin{cases}
\mbox{Budget Constraint:}\qquad \mathbb{E}\left[\int_{0}^{\infty}H(t)(c_{\scriptscriptstyle nob}(t)+d+wl_{\scriptscriptstyle nob}(t)-w\bar{L})dt\right]\leq x,\\
\mbox{Liquidity Constraint:}\qquad \mathbb{E}\left[\left.\int_{t}^{\infty}\frac{H(s)}{H(t)}(c_{\scriptscriptstyle nob}(s)+d+wl_{\scriptscriptstyle nob}(s)-w\bar{L})ds\right|\mathcal{F}_{t}\right]\ge F+\eta.
\end{cases}
\end{equation*}
Following the same argument with the post-bankruptcy part, we can solve the optimal control problem with two different cases:
\begin{itemize}
	\item Case 1. $0<\tilde{y}\leq\hat{z}_{\scriptscriptstyle nob}$: the optimal consumption-portfolio-leisure strategy is 
	\begin{equation*}
	c_{\scriptscriptstyle nob}^{*}(t)=\begin{cases}
	L^{-\frac{(1-k)(1-\delta)}{\delta(1-k)-1}}(Z_{\scriptscriptstyle nob}^{*}(t))^{\frac{1}{\delta(1-k)-1}}, &0<Z_{\scriptscriptstyle nob}^{*}(t)<\tilde{y},\\
	\left(\frac{1-\delta}{\delta w}\right)^{\frac{(1-\delta)(1-k)}{k}}(Z_{\scriptscriptstyle nob}^{*}(t))^{-\frac{1}{k}}, & \tilde{y}\leq Z_{\scriptscriptstyle nob}^{*}(t)\leq\hat{z}_{\scriptscriptstyle nob},\\
	\end{cases}
	\end{equation*}
	\begin{equation*}
	l_{\scriptscriptstyle nob}^{*}(t)=\begin{cases}
	L, & 0<Z_{\scriptscriptstyle nob}^{*}(t)<\tilde{y},\\
	\left(\frac{1-\delta}{\delta w}\right)^{-\frac{\delta(1-k)-1}{k}}(Z_{\scriptscriptstyle nob}^{*}(t))^{-\frac{1}{k}}, & \tilde{y}\leq Z_{\scriptscriptstyle nob}^{*}(t)\leq\hat{z}_{\scriptscriptstyle nob},\\
	\end{cases}
	\end{equation*}
	\begin{equation*}
	\pi_{\scriptscriptstyle nob}^{*}(t)\!=\!\begin{cases}
	\!\frac{\theta}{\sigma}\!\bigg[\!n_{2}(n_{2}\!-\!1)B_{\scriptscriptstyle 21,nob}(Z_{\scriptscriptstyle nob}^{*}(t))^{n_{2}\!-\!1}\! & \\
	\qquad\qquad\qquad\qquad\qquad\quad+\frac{\delta(1-k)}{(\delta(1-k)-\!1)^{2}}\frac{A_{1}}{\Gamma_{1}}(Z_{\scriptscriptstyle nob}^{*}(t))^{\frac{1}{\delta(1-k)\!-\!1}}\bigg], 
	&0\!<\!Z_{\scriptscriptstyle nob}^{*}(t)\!<\!\tilde{y},\\
	\!\frac{\theta}{\sigma}\!\bigg[\!n_{1}\!(n_{1}\!-\!1)\!B_{\scriptscriptstyle 12,\!nob}\!(Z_{\scriptscriptstyle nob}^{*}(t))\!^{n_{1}\!-\!1}\!+\!n_{2}(n_{2}\!-\!1)\!B_{\scriptscriptstyle 22,\!nob}\!(Z_{\scriptscriptstyle nob}^{*}(t))^{n_{2}\!-\!1} & \\
	\qquad\qquad\qquad\qquad\qquad\qquad\qquad\quad\quad +\frac{1-k}{k^{2}}\frac{A_{2}}{\Gamma_{2}}(Z_{\scriptscriptstyle nob}^{*}(t))^{-\frac{1}{k}}\!\bigg], & \tilde{y}\!\leq \!Z_{\scriptscriptstyle nob}^{*}\!(t)\!\leq\!\hat{z}_{\scriptscriptstyle nob},\\
	\end{cases}
	\end{equation*}
	and the optimal wealth process is
	\begin{equation*}
	X_{\scriptscriptstyle nob}^{*}(t)\!=\!\begin{cases}
	-n_{2}B_{\scriptscriptstyle 21, nob}(Z_{\scriptscriptstyle nob}^{*}(t))^{n_{2}-1}\!-\frac{\delta(1-k)}{\delta(1-k)-1}\frac{A_{1}}{\Gamma_{1}}(Z_{\scriptscriptstyle nob}^{*}(t))^{\frac{1}{\delta(1-k)-1}}& \\
	\qquad\qquad\qquad\qquad\qquad\qquad\qquad\qquad\qquad-\frac{w(\bar{L}-L)-d}{r},& 0<Z_{\scriptscriptstyle nob}^{*}(t)<\tilde{y},\\
	-n_{1}B_{\scriptscriptstyle 12,nob}(Z_{\scriptscriptstyle nob}^{*}(t))^{n_{1}-1}-n_{2}B_{\scriptscriptstyle 22,nob}(Z_{\scriptscriptstyle nob}^{*}(t))^{n_{2}-1} & \\
	\qquad\qquad\qquad\qquad\qquad\quad+\frac{1-k}{k}\frac{A_{2}}{\Gamma_{2}}(Z_{\scriptscriptstyle nob}^{*}(t))^{-\frac{1}{k}}\!-\frac{w\bar{L}-d}{r},&\tilde{y}\leq Z_{\scriptscriptstyle nob}^{*}(t)\leq\hat{z}_{\scriptscriptstyle nob}.
	\end{cases}
	\end{equation*}
	
	\item Case 2. $0<\hat{z}_{\scriptscriptstyle nob}<\tilde{y}$: the optimal consumption-portfolio-leisure strategy defined on the interval $0<Z_{\scriptscriptstyle nob}^{*}(t)\leq\hat{z}_{\scriptscriptstyle nob}$ is
	\begin{equation*}
	c_{\scriptscriptstyle nob}^{*}(t)=L^{-\frac{(1-k)(1-\delta)}{\delta(1-k)-1}}(Z_{\scriptscriptstyle nob}^{*}(t))^{\frac{1}{\delta(1-k)-1}},\qquad l_{\scriptscriptstyle nob}^{*}(t)=L,
	\end{equation*}
	\begin{equation*}
	\pi_{\scriptscriptstyle nob}^{*}(t)=	\frac{\theta}{\sigma}\left[n_{2}(n_{2}-1)B_{\scriptscriptstyle 2,nob}(Z_{\scriptscriptstyle nob}^{*}(t))^{n_{2}-1}+\frac{\delta(1-k)}{(\delta(1-k)-1)^{2}}\frac{A_{1}}{\Gamma_{1}}(Z_{\scriptscriptstyle nob}^{*}(t))^{\frac{1}{\delta(1-k)\!-1}}\right],
	\end{equation*}
	and the optimal wealth process is
	\begin{equation*}
	X_{\scriptscriptstyle nob}^{*}(t)=-n_{2}B_{\scriptscriptstyle 2,nob}(Z_{\scriptscriptstyle nob}^{*}(t))^{n_{2}-1}-\frac{A_{1}}{\Gamma_{1}}\frac{\delta(1-k)}{\delta(1-k)-1}(Z_{\scriptscriptstyle nob}^{*}(t))^{\frac{1}{\delta(1-k)-1}}-\frac{w(\bar{L}-L)}{r}.
	\end{equation*}
\end{itemize}

\end{document}